\documentclass[10pt]{article} 
\usepackage[preprint]{tmlr}

\usepackage{amsmath,amsfonts,bm}

\def\eqref#1{equation~\ref{#1}}

\def\1{\bm{1}}

\def\eps{{\epsilon}}

\DeclareMathAlphabet{\mathsfit}{\encodingdefault}{\sfdefault}{m}{sl}
\SetMathAlphabet{\mathsfit}{bold}{\encodingdefault}{\sfdefault}{bx}{n}

\usepackage{nicematrix}
\usepackage[pdftex, pdfstartview={FitV}, plainpages = false, colorlinks=true, linkcolor=black, citecolor = green!50!black, urlcolor = black,filecolor=black, pagebackref=false,hypertexnames=false, plainpages=false, pdfpagelabels ]{hyperref}
\usepackage{url}
\usepackage{graphicx}
\usepackage{float,enumerate}
\usepackage{tcolorbox}
\usepackage{booktabs}
\usepackage[disable]{todonotes}
\usepackage[table]{xcolor}
\usepackage{mathtools}
\usepackage{mathrsfs}
\newcommand{\ignore}[1]{}

\usepackage{tabularx}
\usepackage[inline]{enumitem}
\usepackage{multirow}
\usepackage{algorithm}
\usepackage{diagbox}
\usepackage{algpseudocode}
\usepackage{graphicx,wrapfig,lipsum}
\usepackage{amsthm}
\newtheorem{lemma}{Lemma}
\newtheorem{proposition}{Proposition}
\usepackage{enumitem}
\newtheorem{definition}{Definition}[section]

\newtheorem{corollary}{Corollary}[section]
\newtheorem{experiment}{Experiment}[section]
\DeclareMathOperator*{\argsup}{argsup}
\DeclareMathOperator*{\arginf}{arginf}
\usepackage{amsmath, amssymb}
\usepackage{tcolorbox}
\usepackage{tabularx}
\usepackage[normalem]{ulem} 
\usepackage{tikz} 
\newtcolorbox{optbox}[1][]{
  colback=black!3,
  colframe=black!30,
  coltitle=black,
  colbacktitle=black!8,
  fonttitle=\bfseries,
  title=#1,
  boxrule=0.4pt,
  arc=1mm
}

\title{Beyond Bayesian Nash: Learning  Minimax-Regret Equilibria for Adversarial Team Games under Asymmetric Information}

\author{\name Naman Aggarwal \email namanagg@mit.edu \\
      \addr Aerospace Control Laboratory \\ Laboratory of Information and Decision Systems\\
      Massachusetts Institute of Technology
      \AND
      \name Jonathan P.\ How \email jhow@mit.edu \\
      \addr Aerospace Control Laboratory \\ Laboratory of Information and Decision Systems \\ Massachusetts Institute of Technology
    }

\def\month{MM}  
\def\year{YYYY} 
\def\openreview{\url{https://openreview.net/forum?id=XXXX}} 

\begin{document}
\maketitle
\begin{abstract}
Adversarial team games (ATGs) with asymmetric information, such as adversarial path-finding, goal search, and reachability games on graphs, require strategies that are robust to hidden opponent types, such as a hidden goal flag, and to deception. Under asymmetric information, deception is seen as strategic shifts in the type distribution such that the omniscient opponent can collude with Nature and condition its play on the observed type. Existing risk-neutral solution concepts, such as Bayesian Nash equilibrium (BNE), are sensitive to distribution shifts, while distributionally robust approaches provide guarantees only within a prescribed ambiguity set. To address these limitations, we introduce Probabilistically Robust Minimax-Regret Equilibrium (PR-MRE), a novel equilibrium concept that combines the distribution-free robustness of minimax-regret reasoning with probabilistic information from a nominal type distribution. PR-MRE minimizes worst-case regret over a high-confidence subset of the type space, providing protection against strategic redistribution of probability mass while avoiding the conservatism of fully distribution-free approaches. 
We show that, for finite normal-form Bayesian games, {\color{black}approximate} PR-MRE computation can be formulated as a robust bilinear program and derive a tractable semidefinite relaxation. We then adapt this relaxation into a novel meta-solver within a robust double-oracle framework, PRMRE-PSRO, enabling population-based learning of approximate PR-MRE strategies via deep reinforcement learning best responses.
Experiments on graph-structured adversarial team games demonstrate that PR-MRE discovers strategies with substantially improved worst-case performance across hidden types compared to risk-neutral Bayesian Nash equilibrium solutions, resulting in more robust behavior under strategic distribution shifts.
\end{abstract}

\section{Introduction}
Computation of equilibria in large imperfect-information games has been a central driver of recent advances in artificial intelligence. A key challenge in such settings is determining how strategic agents should reason about uncertainty. Existing equilibrium concepts differ fundamentally in their treatment of risk and robustness, ranging from risk-neutral formulations that optimize expected performance to risk-averse approaches that guard against uncertainty through worst-case reasoning. These differing notions of rationality lead to markedly different behaviors and robustness guarantees. Consequently, the literature on equilibria and equilibrium refinements for imperfect-information games is extensive (see Table~\ref{tab:taxonomy} for an overview).
\begin{figure}[t]
  \centering
  \includegraphics[scale=0.55]{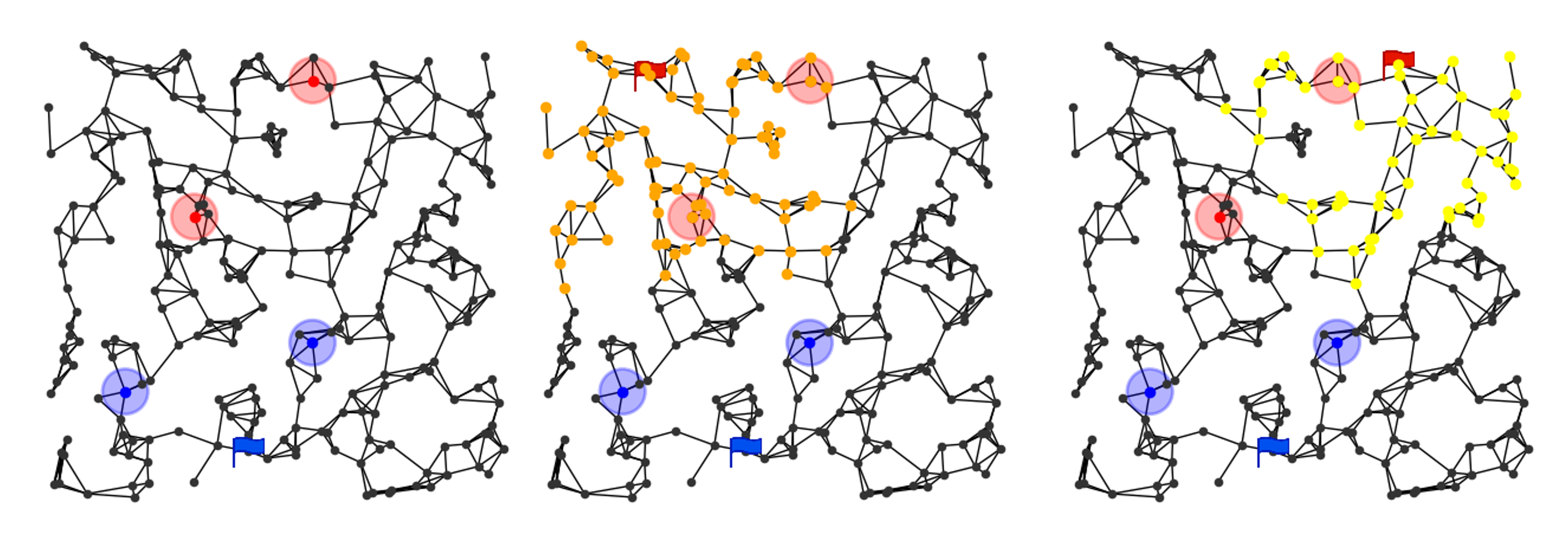}
   \caption{Graph-Structured Capture-the-Flag (Graph CtF) with Blue Team's Imperfect Information about the Red Flag Location: 2v2 CtF on a graph topology with Flag Uncertainty for Blue Team seen as two potential Red Flag locations. Limited Knowledge: Graph assumed known, but limited flag visibility (orange (yellow) indicates nodes on Left (Right) Flag are visible). Asymmetric: Blue agents cannot distinguish between Red flag locations at other nodes, whereas Red agents have perfect information about Blue flag. 
}
  \label{fig:logo}
\end{figure}
In Bayesian games, a Nature player moves before play begins and assigns each player a type drawn from a joint type distribution. Harsanyi's seminal framework \cite{harsanyi_1,harsanyi_2} models strategic interactions with incomplete information by introducing the Bayesian Nash Equilibrium (BNE), in which each player maximizes expected utility under the common prior distribution over types. While BNE provides a principled solution concept for games with hidden information, it relies critically on the common prior assumption: although players do not observe each other's private types, the underlying type distribution is assumed to be common knowledge. This assumption can be restrictive in real-world strategic interactions involving asymmetric information and deception. In many settings, the type distribution itself may be uncertain, evolve over time, or even be influenced by an adversary with an information advantage. We refer to such changes as \emph{strategic distribution shifts}. Unlike classical distributionally robust optimization (DRO), which models uncertainty through structured ambiguity sets such as KL-divergence, Wasserstein, or $f$-divergence balls around a nominal distribution, strategic distribution shifts need not admit a local divergence-based characterization. They may instead arise from regime changes, task-level shifts, or adaptive opponents that fundamentally alter the distribution of types.

Figure~\ref{fig:logo} provides a concrete example of such strategic distribution shifts in an adversarial search game. The Blue team must locate and capture a hidden Red flag before a timeout, while the Red team knows the true flag location and can defend accordingly. Blue only has access to a prior belief over possible flag locations (e.g., 60\% right and 40\% left), making the flag location a hidden type known only to Red. Consequently, a strategy that is optimized for a nominal prior may become vulnerable when the underlying type distribution changes or when Red exploits Blue’s prior assumptions.

We consider imperfect-information games with deception, specifically asymmetric-information games in which the type distribution may undergo strategic distribution shifts and the informed opponent can condition its play on the realized type. The Capture-the-Flag example illustrates a broader challenge: behaviors derived from risk-neutral solution concepts such as Bayesian Nash Equilibrium (BNE) can be highly exploitable because they rely on a nominal type distribution. This motivates the central question addressed in this work:
\textit{What is an appropriate notion of robust best response for the information-disadvantaged team, and the resulting equilibrium concept, that provides favorable worst-case performance guarantees under strategic distribution shifts?}

\begin{table}[t]
\centering
\caption{Taxonomy of robust Blue best-responses (BRs) according to the optimization objective and the amount of information trusted about Nature's nominal. Each robust Blue BR coupled with the informed Red BR leads to a unique equilibrium concept (see Section~\ref{sec:eqm_concepts_section} for details) and an associated robustness property (see Table~\ref{tab:rob_properties}) for the adversarial game with asymmetric information.}
\vspace{0.5em}
\small
\renewcommand{\arraystretch}{1.15}
\resizebox{\textwidth}{!}{%
\begin{tabular}{lcccc}
\toprule
& \multicolumn{4}{c}{\footnotesize High~~ $\xrightarrow{\hspace{2.25cm}\ \text{\fontsize{9.5}{11}\selectfont Trust in Nature's Nominal Distribution} \ \hspace{2.25cm}}$ ~~Low} \\[2pt]
\diagbox[width=3.4cm,height=1.15cm]{\textbf{Objective}}{\textbf{Trust}} 
&
\textbf{Full} &
\textbf{Local} &
\textbf{Coarse} &
\textbf{None} \\
\midrule

\textbf{Payoff}
&
\begin{tabular}{@{}c@{}}
Bayesian BR (BNE)\\
{\footnotesize Nominal Prior}
\end{tabular}
&
\begin{tabular}{@{}c@{}}
DRO-Ball\\
{\footnotesize Local Perturbations}
\end{tabular}
&
\begin{tabular}{@{}c@{}}
DRO-Structured\\
{\footnotesize Typicality-preserving shifts} \\
{\footnotesize (Structured type ambiguity)}
\end{tabular}
&
\begin{tabular}{@{}c@{}}
Worst-Case BR\\
{\footnotesize Entire type-simplex} \\
{\footnotesize (Adversarial Nature)}
\end{tabular} \\

\addlinespace[0.4em]

\textbf{Regret}
&
\cellcolor{gray!5}%
--
&
\cellcolor{gray!5}%
--
&
\cellcolor{gray!40}%
\begin{tabular}{@{}c@{}}
\textbf{PR-MRE}\\
{\footnotesize Typicality-preserving shifts} \\
{\footnotesize (Structured type ambiguity)}
\end{tabular}
&
\cellcolor{gray!20}
\begin{tabular}{@{}c@{}}
\textbf{MRE}\\
{\footnotesize Entire type-space}
\end{tabular}
\\

\bottomrule
\end{tabular}}
\label{tab:taxonomy}
\end{table}
\todo{Distributionally Robust Markov Games (either restrictive examples or heuristic training methods for the larger environments) -- also DROP for Markov Games is DRO and not worst-case robust, something we need for protection against deception (that's our threat model (we seek protection against)) -- we marry robust equilibria with adversarial training of best responses to robust equilibria and demonstrate learning approximate minimax robust policies in graph-structured imperfect-information POMGs.}

Distributionally robust optimization (DRO) protects against uncertainty through ambiguity sets around a nominal distribution. While effective for modeling local perturbations, DRO-based equilibrium concepts provide guarantees only within the specified ambiguity set and may fail under larger regime changes. For example, in a three-flag version of the search-and-capture game, a nominal type distribution of $(0.6,0.1,0.3)$ may shift to $(0.0,0.05,0.95)$, fundamentally altering the strategic landscape. Such strategic distribution shifts are not naturally captured by local divergence-based uncertainty models. At the other extreme, a fully worst-case robust approach reasons against an adversarial type distribution over the entire type simplex, but this can be overly conservative because it effectively concentrates on the most difficult type regardless of its likelihood.

To bridge this gap, we consider strategic distribution shifts that preserve \emph{type typicality}: types that are rare under the nominal distribution remain rare under admissible shifts, while probability mass may be redistributed arbitrarily among typical types. We then adopt a regret-based notion of robustness. Regret, or exploitability, measures the performance loss incurred by acting under imperfect information relative to a type-conditioned optimum. A Minimax-Regret Equilibrium (MRE) minimizes worst-case regret across all types and yields a distribution-free certificate in the form of a uniform lower bound on performance under arbitrary distribution shifts. While robust, MRE can be conservative because it ignores probabilistic information about type occurrence. We therefore introduce the Probabilistically Robust Minimax-Regret Equilibrium (PR-MRE), which incorporates coarse probabilistic information through the typicality-preserving threat model. PR-MRE interpolates between risk-neutral and worst-case approaches, yielding tighter performance guarantees while retaining robustness to strategic distribution shifts.

Computing MRE and PR-MRE in general adversarial team games remains an open challenge. Existing solution methods for ATGs \cite{anagnostides2026algorithms, atgs_neurips_2024, xiao2026solving, pmlr-v235-zhang24b, Celli2018-yn, pmlr-v162-carminati22a, anagnostides2023efficiently} and extensive-form games (EFGs), such as Counterfactual Regret Minimization (CFR) \cite{zinkevich2007regret}, are designed to compute Nash equilibria and do not naturally extend to regret-based robustness objectives under structured type uncertainty. Moreover, MRE was originally introduced in the context of finite mechanism-design settings \cite{mre, renou2010minimax, renou2011implementation}, and existing formulations do not readily generalize to the large strategy spaces such as team treeplexes encountered in adversarial team games.
To address this challenge, we propose PR-MRE PSRO, a robust double-oracle framework for computing approximate MRE and PR-MRE strategies via metagame approximation. We show that the resulting metagame problems can be {\color{black} approximated via surrogate} bilinear and robust bilinear programs, respectively, and derive tractable semidefinite relaxations that serve as novel meta-solvers within a population-based learning framework. Combined with deep reinforcement learning best-response oracles, the proposed approach enables scalable computation of robust strategies in graph-structured adversarial team games.

From Table~\ref{tab:taxonomy}, the optimization objective (payoff vs. regret) and the degree of trust placed in Nature's nominal type distribution define a two-dimensional design space of robust equilibrium concepts, with PR-MRE occupying the previously unexplored regime of regret-based best-response under coarse probabilistic information from the nominal.
We formally state our contributions as follows,
\begin{enumerate}
    \vspace{-0.75em}
    \item \textbf{New equilibrium concept, PR-MRE}. We introduce a novel robust equilibrium concept for asymmetric information games, Probabilistically-Robust Minimax-Regret equilibrium (PR-MRE) that interpolates between risk-neutral, distributionally-robust equilibria and the distribution-free Minimax-Regret equilibrium (MRE). {\color{black} In contrast to MRE that treats all types equally likely, PR-MRE utilizes coarse probabilistic information from the nominal and minimizes worst-case regret across an admissible set of test-time type distributions. We prove existence of a $\operatorname{PR-MRE}$ fixed-point for general asymmetric information games under stated assumptions on the strategy spaces, the type-conditioned value function and the admissible set of type distributions.}
    \item \textbf{Typicality-preserving threat model and robustness guarantees}. {\color{black} We show that a PR-MRE robust best-response for the information-disadvantaged team to a given informed opponent strategy provides a robustness guarantee in the form of the tightest benchmark-relative performance certificate under strategic type distribution shifts} that redistribute probability mass arbitrarily among high-confidence subsets of the type-space.
    \item \textbf{Scalable approximation algorithm, PR-MRE PSRO}. We propose a metagame formulation to learn approximate MRE and PR-MRE in adversarial team games under asymmetric information. For finite normal-form metagames, MRE and PR-MRE fixed-points {\color{black} can be approximated} via bilinear and robust bilinear surrogate programs respectively. We provide a tractable convex relaxation for the bilinear programs that is incorporated as a meta-equilibria within a proposed robust double-oracle framework, PR-MRE PSRO to compute robust team policies for graph-structured ATGs under asymmetric information.
    \item \textbf{Empirical validation in graph-structured adversarial team games}: We demonstrate empirically that PR-MRE PSRO discovers behaviors that are substantially more robust to distribution shifts than risk-neutral equilibrium solutions.
\end{enumerate}

\section{Graph-Structured Adversarial Team Games under Asymmetric Information}
\label{sec:atg}
\begin{figure}[t]
  \centering
  \includegraphics[scale=0.35]{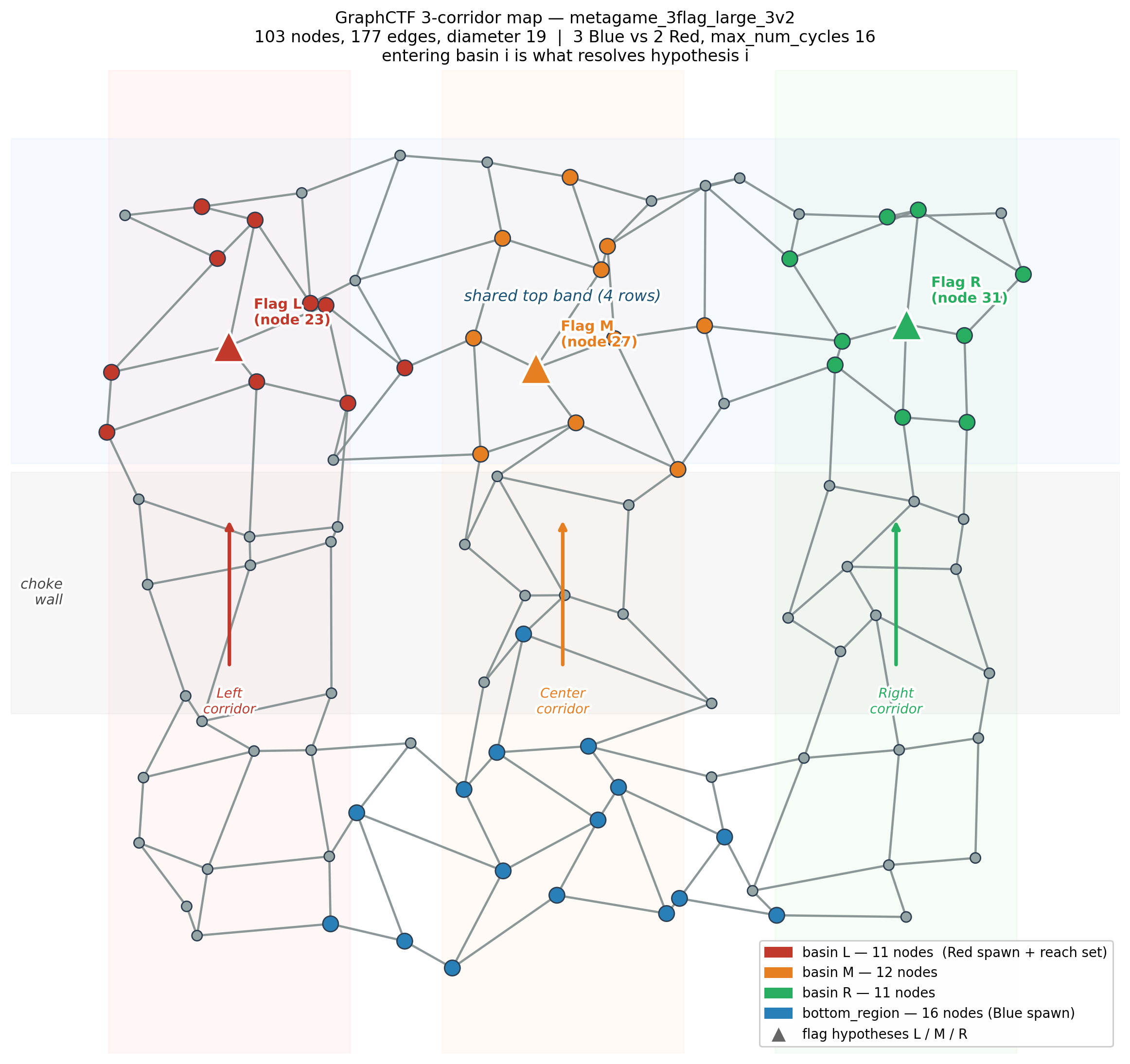}
  \caption{\color{black}{3 Blue vs 2 Red Graph Capture-the-Flag with three Red flag hypotheses -- $\theta_{0}$ (\textbf{L}eft), $\theta_{1}$ (\textbf{R}ight) and $\theta_{2}$ (\textbf{M}iddle), one flag at the end of each corridor (shaded). The red (green, orange) nodes represent the information frontier or the nodes from which the Left (Right, Middle) flag hypothesis is visible. Red team agents initialize in the top-basin whereas the Blue team agents initialize in the bottom-basin below the choke zone. The Blue team agents cannot distinguish between the three flag hypotheses before they reach the information frontier while passing through the choke zone. Tagging causes respawns. Mission length is 16 and the graph diameter is 19 for reference.}}
  \label{wrap-fig:atg}
\end{figure}
We study graph-structured adversarial team games under asymmetric information, instantiated as Graph Capture-the-Flag (Graph CtF). Two teams of agents move on an undirected graph over a finite horizon and optimize opposing objectives -- Blue team wins if the Blue agents capture the Red flag, and Red team wins if they successfully defend their flag for the time-horizon. The key asymmetry is that the Red team observes the true flag location, while the Blue team acts under uncertainty over a finite set of flag hypotheses.

Formally, the game is defined by the tuple $\mathscr{G} = \left< \mathcal{G}, n_r, n_b, \Theta, T \right>,
$ where \(\mathcal{G} = \langle \nu, \mathcal{E} \rangle\) is an undirected graph with node set \(\nu\), \(n_r\) and \(n_b\) are the number of Red and Blue agents, \(\Theta \subseteq \nu\) is the set of possible flag locations, and \(T\) is the time horizon. The hidden type is sampled as $
\theta \sim \mu \in \Delta(\Theta)
$. The game proceeds sequentially over \(T\) steps. At each time step, agents select actions based on their available information and move on the graph according to the transition dynamics. The Red team observes the realized type \(\theta\), whereas the Blue team does not, and must act based only on its observation history. Accordingly, we consider policy classes \(\Pi_r\) and \(\Pi_b\), where Red policies may condition on \(\theta\) and history, while Blue policies depend only on the Blue team’s observation history. Let \(V(\pi_r, \pi_b; \theta)\) denote the expected return of the Blue team under policies \(\pi_r \in \Pi_r\) and \(\pi_b \in \Pi_b\) when the underlying type is \(\theta\). {\color{black} For our development, we assume that $V(\cdot, \cdot, \cdot)$ is continuous, concave in $\pi_{b} \in \Pi_{b}$ and convex in $\pi_{r} \in \Pi_{r}$ for all types $\theta \in \Theta$. Moreover, we assume that all ambiguity sets $\mathcal{S} \subseteq \Delta(\lvert\Theta \rvert)$ considered in our discussion are compact and convex.} 

The nominal type distribution $\mu(\cdot)$ induces a standard Bayesian objective. In deceptive settings, however, the true distribution over types may differ substantially from $\mu(\cdot)$, rendering policies optimized for the nominal distribution highly exploitable. This motivates the search for equilibrium concepts and associated best-response notions that remain robust under strategic distribution shifts.
\subsection{Equilibrium Concepts and Robustness Properties}
\label{sec:eqm_concepts_section}
\begin{table}[t]
\centering
\caption{Robustness properties associated with the various Blue best-responses (BRs) corresponding to different optimization objectives and trust levels on Nature (see Table~\ref{tab:taxonomy} for taxonomy on robust Blue BRs).}
\vspace{0.5em}
\small
\renewcommand{\arraystretch}{1.15}
\resizebox{\textwidth}{!}{%
\begin{tabular}{lcccc}
\toprule
& \multicolumn{4}{c}{\footnotesize High~~ $\xrightarrow{\hspace{2.25cm}\ \text{\fontsize{9.5}{11}\selectfont Trust in Nature's Nominal Distribution} \ \hspace{2.25cm}}$ ~~Low} \\[2pt]
\diagbox[width=3.4cm,height=1.15cm]{\textbf{Objective}}{\textbf{Trust}} 
&
\textbf{Full} &
\textbf{Local} &
\textbf{Coarse} &
\textbf{None} \\
\midrule

\textbf{Payoff}
&
No Robustness
&
\begin{tabular}{@{}c@{}}
Worst-case  \\
payoff certificate \\
{\footnotesize Local perturbations}
\end{tabular}
&
\begin{tabular}{@{}c@{}}
Worst-case \\
payoff certificate \\
{\footnotesize Typicality-preserving shifts} \\
{\footnotesize (Structured type ambiguity) }
\end{tabular}
&
\begin{tabular}{@{}c@{}}
Worst-case\\
payoff certificate \\
{\footnotesize Entire type-simplex}
\end{tabular} \\

\addlinespace[0.4em]

\textbf{Regret}
&
\cellcolor{gray!5}%
--
&
\cellcolor{gray!5}%
--
&
\cellcolor{gray!40}%
\begin{tabular}{@{}c@{}}
\textbf{Worst-case} \\ \textbf{regret certificate} \\ $\Leftrightarrow$ \textbf{ULB} (Lemma~\ref{lemma:prmre_rob_property}) \\
{\footnotesize Typicality-preserving shifts} \\
{\footnotesize (Structured type ambiguity)}
\end{tabular}
&
\cellcolor{gray!20}
\begin{tabular}{@{}c@{}}
\textbf{Worst-case}\\ \textbf{regret certificate} \\
$\Leftrightarrow$ \textbf{ULB} (Lemma~\ref{lemma:ulb_mre}) \\
{\footnotesize Entire type-space}
\end{tabular}
\\

\bottomrule
\end{tabular}}
\label{tab:rob_properties}
\end{table}
In this section, we present a unified treatment of various equilibrium concepts in the specific context of adversarial team games and analyse their robustness properties with respect to shifts in the type distribution. For the case of asymmetric information, the Red team has perfect knowledge of Nature's hidden type (flag location) and can condition its play on the observed type.
\begin{definition}
\label{assumption:red_team_rationality} The Red team best-response operator $\operatorname{BR}^{\mathrm{(r)}}(\pi_{b}; \theta): \Pi_{\mathrm{b}} \times \Theta \rightarrow \Pi_{\mathrm{r}}$ to the Blue team strategy $\pi_{b} \in \Pi_{\mathrm{b}}$ and observed type $\theta$ is given as $\operatorname{BR}^{\mathrm{(r)}}(\pi_{b}; \theta) = \arginf\limits_{\pi \in \Pi_{\mathrm{r}}} V(\pi(\cdot), \pi_{b}; \theta)$.
\end{definition}

Thus, for a given Blue team strategy, the Red team computes a best-response that minimizes the expected Blue team return conditioned on the observed type. How the Blue team defines as `rational' play under imperfect information leads to a unique equilibrium concept and an associated robustness property. We now define various best-response operators for the Blue team, each of which, when coupled with the Red team best-response operator stated in Definition~\ref{assumption:red_team_rationality} defines a unique equilibrium.

\subsubsection{\textbf{Risk-Neutral Methods}}
\label{sec:bne}
\begin{definition}\label{def:bbr}(Bayesian Best-Response Operator under Asymmetric Information) The Blue team Bayesian best-response operator $\operatorname{BR}^{\mathrm{(b)}}(\pi_{r}(\cdot); p, \Theta): \Pi_{\mathrm{r}} \times \Delta(\lvert\Theta\rvert) \rightarrow \Pi_{\mathrm{b}}$ given Red team strategy $\pi_{r}(\cdot)$ and the nominal type distribution $p \in \Delta(\lvert\Theta\rvert)$ is given as follows,
\begin{align}
        \operatorname{BR}^{\mathrm{(b)}}(\pi_{r}(\cdot); p, \Theta) =  \argsup_{\pi \in \Pi_{\mathrm{b}}} \mathbb{E}_{\theta \sim p} [V(\pi_{r}(\cdot), \pi; \theta)].
    \end{align}
\end{definition}
The Bayesian best-response operator maximizes expected team performance under a nominal type distribution leading 
to the risk-neutral equilibrium concept of a Bayesian Nash equilibrium.
\begin{definition}\label{def:bayesian_nash}(Bayesian Nash Equilibrium, BNE) The ex-ante Bayesian Nash equilibrium $(\pi^{*}_{b, \mathrm{BNE}}(p), \pi^{*}_{r})$ for a nominal type distribution $p$ is defined as the fixed point of the best-response operators defined in Definition~\ref{assumption:red_team_rationality} and Definition~\ref{def:bbr} as the following,
$$\pi^{*}_{b, \mathrm{BNE}}(p) \in \operatorname{BR}^{\mathrm{(b)}}(\pi^{*}_{r}(\cdot); p, \Theta) \ \text{and} \ \pi^{*}_{r}(\cdot) \in \operatorname{BR}^{\mathrm{(r)}}(\pi^{*}_{b, \mathrm{BNE}}(p); \cdot).$$
\end{definition}
The BNE is an ex-ante equilibrium concept in which players commit to strategies before play commences and optimize expected utility under the nominal type distribution. We focus on this setting because it aligns naturally with population-based training methods in which meta-equilibria are computed over sets of competing policies. For ex-post notions that incorporate sequential rationality and belief consistency, we refer the reader to refinements such as the perfect Bayesian Nash equilibrium (PBNE) \cite{ouyang2016dynamic,konicki2026computingperfectbayesianequilibria}.

From Definition~\ref{def:bayesian_nash}, the BNE can equivalently be interpreted as a solution to the following saddle-point equation, 
\begin{align}
    \label{eq:bne_saddle_point}
    \mathbb{E}_{\theta \sim p} \left[ V(\pi^{*}_{r}(\cdot), \pi_{b}; \theta) \right] \leq \mathbb{E}_{\theta \sim p} \left[ V(\pi^{*}_{r}(\cdot), \pi^{*}_{b, \mathrm{BNE}}(p); \theta) \right] \leq \mathbb{E}_{\theta \sim p} \left[ V(\pi_{r}(\cdot), \pi^{*}_{b, \mathrm{BNE}}(p); \theta) \right].
\end{align}
At the saddle point given by the strategy-tuple $\left(\pi^{*}_{b}, \pi^{*}_{r}(\cdot)\right)$ (the BNE), no team has any incentive to deviate given the opponent team strategy. An exact computation of the saddle point eq.~(\ref{eq:bne_saddle_point}) for the general formulation where strategy spaces $\Pi_{r}$ and $\Pi_{b}$ are sequence-form polytopes (or team treeplexes) is intractable. Approaches in the MARL community to find tractable approximations to the equilibria of high-dimensional games such as the graph-structured adversarial team game introduced in Section~\ref{sec:atg} broadly consists of DO-based \cite{carminati2022marriageadversarialteamgames, tme_cor_psro}, CFR-based \cite{treeplexes_efg_kroer} and PPO-based \cite{sokkota2026}.
Risk-neutral equilibria, such as BNE (\ref{eq:bne_saddle_point}), do not take into account potential strategic shifts in the type distribution,  which motivates the following class of distributionally-robust and risk-averse methods to manage risk posed by ambiguity due to deception. 

\subsubsection{\textbf{Distributionally-Robust Methods}}
\label{sec:dro}
A distributionally-robust Blue team response given opponent team strategy $\pi_{r}(\cdot) \in \Pi_{\mathrm{r}}$ optimizes for the best worst-case performance for type distributions belonging to an ambiguity set. The distributionally-robust best-response operator corresponding to a distributional ambiguity set $\mathscr{P}$ is defined as $ \operatorname{BR}^{\mathrm{(b)}}_{\mathrm{DRO}}(\pi_{r}(\cdot); \mathscr{P}, \Theta) =  \argsup\limits_{\pi \in \Pi_{b}} \inf\limits_{q \in \mathscr{P}} \mathbb{E}_{\theta \sim q} [V(\pi_{r}(\cdot), \pi; \theta)]$. \todo{or, more generally, as $ \operatorname{BR}^{\mathrm{(b)}}_{\mathrm{RA}}(\pi_{r}(\cdot); \mathscr{P}, \Theta) =  \argsup\limits_{\pi \in \Pi_{\mathrm{b}}} \inf\limits_{q \in \mathscr{P}} \mathbb{\rho}_{\theta \sim q} \left(V(\pi_{r}(\cdot), \pi; \theta)\right)$ where $\rho(\cdot)$ is a general convex risk measure.}
\begin{align}
    \label{eq:dro_brittle}
    \mathbb{E}_{\theta \sim f}[V( \pi_{r}(\cdot), \pi_{b, \mathrm{DRO}}(\pi_{r}(\cdot), \mathscr{P}); \theta)] \leq \inf\limits_{q \in \mathscr{P}}\mathbb{E}_{\theta \sim q}[V( \pi_{r}(\cdot), \pi_{b, \mathrm{DRO}}(\pi_{r}(\cdot), \mathscr{P}); \theta)].
\end{align}
%
The robustness guarantee of the DRO best-response is restricted to distributions contained within the ambiguity set $\mathcal{P}$. Consequently, DRO does not provide performance guarantees under strategic distribution shifts that fall outside the prescribed uncertainty model.

\subsubsection{\textbf{Adversarial Nature and Worst-Case Robustness}}
\label{sec:wc_eqm}
A rational, albeit pessimistic Blue team best-response under adversarial type selection is given by $\operatorname{BR}^{\mathrm{(b)}}_{\mathrm{wc}}(\pi_{r}(\cdot)) =  \argsup\limits_{\pi \in \Pi_{\mathrm{b}}} \left( \inf_{q \in \Delta(\lvert\Theta\rvert)} \mathbb{E}_{\theta \sim q} [V(\pi_{r}(\cdot), \pi; \theta)] \right)$. Let $q_{\mathrm{wc}}$ be the adversarial type distribution corresponding to the above worst-case best-response operator on given $\pi_{r}(\cdot)$ such that for $\pi_{b, \mathrm{wc}} (\pi_{r}(\cdot)) \in \operatorname{BR}^{\mathrm{(b)}}_{\mathrm{wc}}(\pi_{r}(\cdot))$, $$ \sup\limits_{\pi \in \Pi_{\mathrm{b}}} \inf\limits_{q \in \Delta(\lvert\Theta\rvert)} \mathbb{E}_{\theta \sim q} [V(\pi_{r}(\cdot), \pi; \theta)] = \mathbb{E}_{\theta \sim q_{\mathrm{wc}}} [V(\pi_{r}(\cdot), \pi_{b, \mathrm{wc}}(\pi_{r}(\cdot)); \theta)].$$
For a given Red team strategy, the above worst-case best-response operator solves a zero-sum game between the Blue team and Nature to compute a robust Blue team best-response. This leads to the following three-player equilibrium formulation where Nature acts as an adversary and plays an adversarial type distribution such that the Red team and Nature collude (since Red can correlate play with Nature) against the Blue team,
$\pi^{*}_{b, \mathrm{wc}} \in \operatorname{BR}^{\mathrm{(b)}}_{\mathrm{wc}}(\pi^{*}_{r}(\cdot)) \ \text{and} \ \pi^{*}_{r}(\cdot) \in \operatorname{BR}^{\mathrm{(r)}}(\pi^{*}_{b, \mathrm{wc}}; \cdot).$
Worst-case robust equilibrium provides a global performance guarantee when the Nature is adversarial such that $\mathbb{E}_{\theta \sim q_{\mathrm{wc}}}[ V(\pi_{r}(\cdot), \pi^{'}; \theta)] \leq \mathbb{E}_{\theta \sim q_{\mathrm{wc}}}[ V(\pi_{r}(\cdot), \pi_{b, \mathrm{wc}}; \theta)]$. This however can lead to overly pessimistic strategies such that,
\begin{align}\label{eq:wc_dro_comparison_eq}
 \mathbb{E}_{\theta \sim q_{\mathrm{wc}}}[ V(\pi_{r}(\cdot), \pi_{b, \mathrm{wc}}; \theta)] \leq \mathbb{E}_{\theta \sim q \in \mathcal{Q} }[ V(\pi_{r}(\cdot), \pi_{b, \mathrm{wc}}; \theta)] \leq \mathbb{E}_{\theta \sim q \in \mathcal{Q} }[ V(\pi_{r}(\cdot), \pi_{b, \mathrm{DRO}}; \theta)].
\end{align}
The three-player game formulation treats all types with equal importance independent of likelihood under nominal distribution and computes a robust Blue response to an adversarial type distribution that concentrates mass potentially on the hardest type. Such an approach discards available probabilistic information about the occurrence of types and could lead to behavior with poor typical performance under non-adversarial distribution shifts. In the pursuit of global performance guarantees, we discuss a tangential approach in the following subsection that reasons about regret with respect to type-optimal best-responses.

\subsubsection{\textbf{Minimax-Regret Robust Best-Response and Minimax-Regret Equilibrium (MRE)}}
Distributionally robust methods offer robustness only up to an ambiguity set centered at the nominal distribution and are suitable for settings when protection is sought against mis-specification of the nominal distribution (prior) and not strategic distribution shifts. \textit{Strategic distribution shifts} refer to a \textit{smart} Nature running an oracle procedure to compute a type distribution that potentially places probability mass adversarially on high-regret types. We characterize regret (or exploitability) with respect to a type as follows, 
\begin{tcolorbox}[
colback=white,
colframe=black!20,
colbacktitle=white,
coltitle=black
]
\begin{definition}\label{def:exp}(Exploitability)
We define the exploitability (or regret) of Blue team policy $\pi$ on type $\theta$ given Red team policy $\pi_{r}(\cdot)$ as, 
\begin{align}\label{eq:exp}
    \mathcal{E} \left(\pi_{r}(\cdot), \pi; \theta\right) = \sup_{\pi^{*} \in \Pi_{b}}V(\pi_{r}(\cdot), \pi^{*}; \theta) - V(\pi_{r}(\cdot), \pi; \theta),
\end{align}
and the worst-case exploitability of $\pi$ given $\pi_{r}(\cdot)$ over the type-space as $\mathcal{E} \left(\pi_{r}(\cdot), \pi; \Theta \right) = \sup\limits_{\theta \in \Theta} \mathcal{E} \left(\pi_{r}(\cdot), \pi; \theta\right)$.
\end{definition}
\end{tcolorbox}
{\color{black} From our stated assumptions on $\Pi_{b}$ and $V(\cdot, \cdot, \cdot)$, type-conditioned regret $\mathcal{E}(\pi_{r}(\cdot), \pi_{b}; \theta)$ is continuous and convex in $\pi_{b}$. Moreover, from the application of Berge's maximum theorem on the benchmark term $\sup_{\pi^{*} \in \Pi_{b}}V(\pi_{r}(\cdot), \pi^{*}; \theta)$, type-conditioned regret is also a continuous function of $\pi_{r}(\cdot)$.} We define the following Blue team robust best-response given opponent team strategy $\pi_{r}(\cdot)$ that minimizes worst-case regret with respect to type-optimal best-responses across the type-space.

\begin{definition}\label{def:mmr}(Minimax-Regret Best-Response Operator under Asymmetric Information)
\begin{align}\label{eq:mre_br_def}
\operatorname{BR}^{\mathrm{(b)}}_{\mathrm{MMR}}(\pi_{r}; \Theta) =  \arginf_{\pi \in \Pi_{\mathrm{b}}} \sup_{\theta \in \Theta} \mathcal{E} \left(\pi_{r}(\cdot), \pi; \theta\right)
    \end{align}
\end{definition}
{\color{black} The minimax-regret best-response can be equivalently defined as $\arginf_{\pi \in \Pi_{\mathrm{b}}} \sup_{\mu \in \Delta(\vert\Theta\rvert)} \sum_{\theta \in \Theta}\mu(\theta)\mathcal{E} \left(\pi_{r}(\cdot), \pi; \theta\right)$.
The robust best-response set as defined above is non-empty from Sion's Minimax theorem which follows from finite $\Theta$ and our stated assumptions on $\Pi_{b}$.} In other words, the minimax-regret robust best-response penalizes deviation from the performance obtained under perfect information about the type. Such a rational attitude hedges across all types under imperfect information and prevents being overly sub-optimal on any single type. The ex-ante Minimax-Regret (MRE) equilibrium $(\pi^{*}_{b}, \pi^{*}_{r})$ is defined as the fixed point of the best-response operator defined in Def.~\ref{def:mmr},
$\pi^{*}_{b, \mathrm{MRE}} \in \operatorname{BR}^{\mathrm{(b)}}_{\mathrm{MMR}}(\pi^{*}_{r}; \Theta) \ \text{and} \ \pi^{*}_{r}(\cdot) \in \operatorname{BR}^{\mathrm{(r)}}(\pi^{*}_{b, \mathrm{MRE}}; \cdot).$
The MRE best-response is distribution-free and provides global performance guarantees independent of any restrictive assumptions about the type to lie within an ambiguity set around the nominal distribution.
The key consequence of minimax-regret reasoning is that it induces a performance certificate that is valid for every type distribution in the simplex. This certificate is formalized below.
\begin{tcolorbox}[
colback=white,
colframe=black!20,
colbacktitle=white,
coltitle=black
]
\begin{lemma}[\textit{{\color{black} Point-wise Benchmark Relative} Performance Lower Bound}]\label{lemma:ulb_mre} For a given Red team strategy $\pi_{r}(\cdot)$, the expected payoff under Blue team strategy $\pi_{b}$ and a type distribution $\mu(\cdot)$ can be bounded from below as follows. Moreover, $ \exists \ q \in \Delta(\lvert\Theta\rvert) $ such that the lower-bound is tight,
\begin{align}\label{eq:mre_perf_lb}
    \mathbb{E}_{\theta \sim \mu}[V(\pi_{r}(\cdot), \pi_{b}; \theta] \geq \mathcal{V}^{(-)}(\pi_{b}, \mu) =  \sum_{\theta \in \Theta} \mu(\theta) \sup_{\pi \in \Pi_{\mathrm{b}}} V(\pi_{r}, \pi; \theta) - \mathcal{E}(\pi_{r}(\cdot), \pi_{b}, \Theta) \ \forall \ \mu \in \Delta(\lvert\Theta\rvert).
\end{align}
\end{lemma}
\begin{proof}
    Proof in Appendix~\ref{sec:appendix_proofs}.
\end{proof}
\end{tcolorbox}
From (\ref{eq:mre_perf_lb}), the lower-bound on performance can be decoupled into a per-type component, free of the decision variable $\pi_{b}$, representing optimal decision under perfect information and a second term denoting the worst-case exploitability of the decision $\pi_{b}$ across the type-space. Thus, the MRE best-response $\pi_{b, \mathrm{MRE}} \in \operatorname{BR}^{\mathrm{(b)}}_{\mathrm{MMR}}(\pi_{r}; \Theta)$ from (\ref{eq:mre_br_def}) can equivalently be interpreted as the decision $\pi_{b, \mathrm{MRE}}$ that maximizes the {\color{black} benchmark-relative} performance lower-bound $\mathcal{V}^{(-)}(\cdot, \mu)$ for all type distributions $\mu(\cdot)$ in the type-simplex i.e., $$\operatorname{BR}^{\mathrm{(b)}}_{\mathrm{MMR}}(\pi_{r}; \Theta) =  \arginf\limits_{\pi \in \Pi_{\mathrm{b}}} \sup\limits_{\theta \in \Theta} \mathcal{E} \left(\pi_{r}(\cdot), \pi; \theta\right) = \argsup\limits_{\pi \in \Pi_{\mathrm{b}}} \mathcal{V}^{(-)}(\pi, \mu) \ \forall \ \mu(\cdot) \in \Delta(\lvert\Theta\rvert).$$ 

In contrast to DRO methods that do not provide guarantees for perturbations outside the ambiguity set and worst-case robust equilibria that concentrate mass on hard types agnostic to the likelihood information contained in the nominal type distribution leading to poor typical performance under non-adversarial distribution shifts, {\color{black}MRE robust best-response to a fixed Red provides a point-wise lower bound on performance across the type-simplex by hedging regret across all types}.

\textbf{Limitation}: MRE minimizes worst-case exploitability across the entire type-space and does not utilize useful probabilistic information from the nominal type distribution. For instances where a type with high exploitability is rare under the nominal type-distribution and admissible perturbations thereof, minimizing worst-case exploitability across the entire type-space might lead to conservative performance under typical distributions. This necessitates the need for a robust best-response and an associated equilibrium concept bridging distributionally robust and risk-averse methods and minimax-regret equilibria that utilizes probabilistic information about type occurrence whilst also providing robustness guarantees to non-adversarial distribution shifts in terms of a {\color{black}point-wise} performance lower bound for typical distributions.

\section{Probabilistically-Robust Minimax-Regret Equilibria}
\label{sec:PRMRE}
Motivated by the limitations of existing equilibrium concepts as discussed in Section~\ref{sec:eqm_concepts_section}, we introduce a novel robust best-response operator and the associated equilibrium concept -- Probabilistically-Robust Minimax-Regret equilibria for ATGs under asymmetric information that minimizes exploitability across subsets of the type-space weighted by known information about their probability mass under admissible distribution shifts.
\begin{figure}[t]
  \centering
  \includegraphics[scale=0.60]{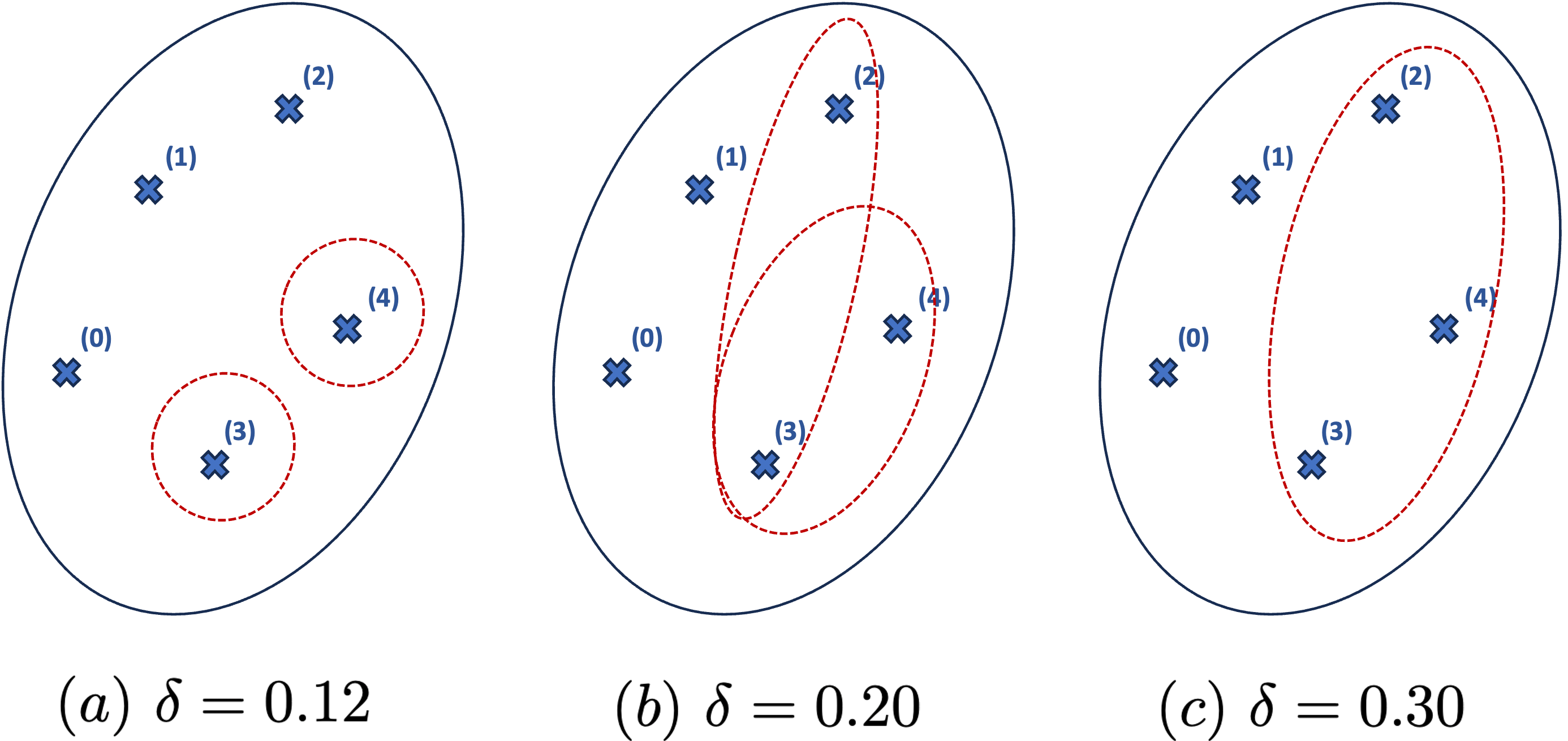}
  \caption{Illustration of PR-MRE, nominal distribution $\bar{\mu}(\cdot) = (0.35, 0.35, 0.15, 0.03, 0.12)$. The blue outer ellipse denotes the type space $\Theta = \{\theta_{0}, \cdots, \theta_{4}\}$, the crosses denote specific types and the red dotted ellipses denote the set of low-probability covers $ \Theta_{l}(\bar{\mu}, \delta)$ for various confidence levels $\delta$. For (a), $\Theta_{l}(\bar{\mu}, 0.12) = \{ \{\theta_{3}\}, \{\theta_{4}\} \}$ such that under admissible shifts $\mu \in \Delta(\lvert\Theta\rvert)$, $ \mu(\delta_{3}) \leq \delta \text{ and } \mu(\delta_{4}) \leq \delta$. Similarly for (b), $\Theta_{l}(\bar{\mu}, 0.20) = \{ \{\theta_{2}, \theta_{3}\}, \{ \theta_{3}, \theta_{4}\}\}$ such that $\mu(\theta_{2}) + \mu(\theta_{3}) \leq \delta$ and $ \mu(\delta_{3}) + \mu(\delta_{4}) \leq \delta$, and for (c) $\Theta_{l}(\bar{\mu}, 0.30) = \{ \theta_{2}, \theta_{3}, \theta_{4}\}$ such that $\mu(\theta_{2}) + \mu(\theta_{3}) + \mu(\theta_{4}) \leq \delta$. Rather than minimizing worst-case exploitability uniformly across the entire type space as in MRE, PR-MRE focuses on high-confidence subsets while retaining robustness to low-probability covers defined through the confidence level $\delta$.
  }
  \label{fig:prob_cover}
\end{figure}
We first discuss the threat model. We consider distribution shifts that preserve type typicality, meaning that types that are rare under the nominal distribution remain rare under admissible shifts. Admissible distributions therefore preserve coarse probabilistic information encoded in the nominal type distribution while allowing substantial redistribution of probability mass within high-probability regions of the type space.
This is formalized via covers of the type-space as follows,
\begin{tcolorbox}[
colback=white,
colframe=black!20,
colbacktitle=white,
coltitle=black
]
\begin{definition}\label{def:hp_covers}(Probabilistic Covers of Type-Space)
For a nominal type distribution $\bar{\mu}(\cdot)$ and confidence parameter $\delta$, a valid high-probability cover of the type-space $\mathscr{C} \subseteq \Theta$ is such that $ \sum_{\theta \in \mathscr{C}} \bar{\mu}(\theta) \geq 1 - \delta$. The associated low-probability cover $\mathscr{L}$ to the high-probability cover $\mathscr{C}$ is defined as $\mathscr{L} = \mathscr{C}^{c}$ such that $\sum_{\theta \in \mathscr{L}}\bar{\mu}(\theta) \leq \delta$.
\end{definition}
\end{tcolorbox}
We characterize permissible perturbations to the nominal type distribution as given by the set of distributions that preserve all high-probability covers (or equivalently, all low-probability covers) of the type-space. For a nominal type distribution $\bar{\mu}(\cdot)$ and confidence parameter $\delta$, the set of all $(1-
\delta)$- high-probability covers $\Theta$ is defined as $\Theta(\bar{\mu}, \delta) \coloneqq \{ \mathscr{C} \subseteq \Theta \text{\ such that \ } \sum_{\theta \in \mathscr{C}} \bar{\mu}(\theta) \geq 1 - \delta\}$ and the set of all $\delta$- low-probability covers is defined as $ \Theta_{\ell}(\bar{\mu}, \delta) \coloneqq \{ \mathscr{C}^{c} \text{ for all } \mathscr{C} \in \Theta(\bar{\mu}, \delta)\}$. We also denote the set of all rare types under the nominal distribution and confidence level $\delta$ as $\Theta_{r}(\bar{\mu}, \delta) = \{\theta \text{ such that } \bar{\mu}(\theta) \leq \delta\}$. We now formally characterize our threat model as follows,

\begin{tcolorbox}[
colback=white,
colframe=black!20,
colbacktitle=white,
coltitle=black
]
\begin{definition}(Threat Model)\label{def:threat_model}
The permissible set of perturbed distributions $\mathcal{S}(\bar{\mu}, \delta)$ is given as
$\mathcal{S}(\bar{\mu}, \delta) = \{ \mu(\cdot) \in \Delta(\lvert\Theta\rvert) \ \text{ such that } \ \sum\limits_{\theta \in \mathscr{L}} \mu(\theta) \leq \delta \ \text{ for all } \ \mathscr{L} \in \Theta_{\ell}(\bar{\mu}, \delta)\}$.
\end{definition}
\end{tcolorbox}
In other words, the permissible perturbed distributions under our threat model allocate probability mass arbitrarily up to respecting typicality of types such that low-probability covers remain low probability.

For instance, for a nominal distribution $\bar{\mu}(\cdot) = (0.35, 0.35, 0.15, 0.03, 0.12)$ and $\delta=0.12$, it is easy to verify that the low-probability cover is given by $\Theta_{\ell}(\bar{\mu}, 0.12) = \{(\theta_{3}),( \theta_{4})\}$ (as depicted in Figure~\ref{fig:prob_cover}a) such that for any permissible perturbed distribution $\mu(\cdot) \in \mathcal{S}(\bar{\mu}, \delta)$, $\mu(\theta_{3}) \leq \delta$ and $\mu(\theta_{4}) \leq \delta$. Similarly for $\delta=0.2$, the low probability cover is given by $\Theta_{\ell}(\bar{\mu}, 0.20) = \{ (\theta_{2}, \theta_{3}), (\theta_{3}, \theta_{4})\}$ and permissible distributions $\mu(\cdot)$ are such that $\mu(\theta_{2}) + \mu(\theta_{3}) \leq \delta$ and $\mu(\theta_{3}) + \mu(\theta_{4}) \leq \delta$ (Figure~\ref{fig:prob_cover}b),  and for $\delta = 0.3$ (Figure~\ref{fig:prob_cover}c), $\mu(\theta_{2}) + \mu(\theta_{3}) + \mu(\theta_{4}) \leq \delta$. The inequality constraints arise as a direct consequence of the threat model that preserves coarse probabilistic information about type occurrence (all high-probability covers of the type-space) under distribution shift. Such a threat model makes sense for scenarios when the task hypotheses map to actual real-world configurations. For instance, for capture-the-flag, certain flag hypotheses might actually be realized with a bounded probability because of terrain traversability, other physical constraints etc.

\textbf{Geometry of $\mathcal{S}(\bar{\mu}, \delta)$}: The threat model from Definition~\ref{def:threat_model} admits a meaningful class of strategic perturbations not captured by ball-shaped ambiguity sets centered at a nominal type distribution (based on for example KL- or generalized f-divergence). Unlike ambiguity-set methods that constrain perturbations to remain close to the nominal distribution, our threat model permits arbitrary redistribution of probability mass within a high-confidence subset of the type space while preserving its aggregate probability mass.
For the high-confidence subset $\Theta / \Theta_{r} \subseteq \Theta$, arbitrary allocations of probability mass across types within $\Theta / \Theta_{r}$ is allowed as long as the net mass over rare types $\Theta_{r}$ follows constraints posed by Definition~\ref{def:threat_model}. Consequently, the model admits highly concentrated (``peaky’’) distributions that place most of their mass on a single type, while still preserving coarse probabilistic information about which regions of the type space are typical. Figure~\ref{fig:prob_cover} illustrates the key distinction such that for $\bar{\mu} = (0.35, 0.35, 0.15, 0.03, 0.12)$ and $\delta=0.12$, $(0, 0, 1, 0, 0)$ is a permissible type distribution under the typicality-preserving threat model where $\mu(\theta_{1}), \mu(\theta_{2})$ shift from $0.35 \rightarrow 0$ and $\mu(\theta_{3})$ shifts from $0.15 \rightarrow 1$. Other example distributions contained in the permissible set $\mathcal{S}(\bar{\mu}, 0.12)$ include $(0.12, 0, 0.76, 0.06, 0.06), (0.90, 0, 0, 0.10, 0)$ and $(0, 1, 0, 0, 0)$.


This motivates a robust best-response operator that balances exploitability on high-probability regions of the type space against exploitability on atypical types. The resulting operator is defined as follows.
\begin{tcolorbox}[
title=Probabilistically-Robust Minimax-Regret (PR-MRE) Best-Response Operator,
colback=white,
colframe=black!40,
colbacktitle=white,
coltitle=black
]
\begin{definition}\label{def:prmre}(PR-MRE Best-Response Operator under Asymmetric Information)
For a given Red team strategy $\pi_{r}(\cdot)$, the probabilistically-robust MRE best-response operator for the Blue team is defined as follows,
\begin{align}\label{eq:prmre_br_def}
\operatorname{BR}^{\mathrm{(b)}}_{\mathrm{PRMRE}}(\pi_{r}(\cdot); \bar{\mu}, \delta) =  \arginf_{\pi \in \Pi_{\mathrm{b}}} \sup_{\mu \in \mathcal{S}(\bar{\mu}, \delta)} \sum_{\theta \in \Theta} \mu(\theta) \mathcal{E} \left(\pi_{r}(\cdot), \pi; \theta \right).
\end{align}
\end{definition}
\end{tcolorbox}
{\color{black}\textbf{Remark}: Existence of the robust best-response set as defined in eq.~(\ref{eq:prmre_br_def}) follows from Sion's Minimax Theorem under stated assumptions on $V(\cdot, \cdot, \cdot)$, $\Pi_{b}$ and $\mathcal{S}(\bar{\mu}, \delta)$.}
\begin{tcolorbox}[
colback=white,
colframe=black!40,
colbacktitle=white,
coltitle=black
]
{\color{black}
\begin{lemma}\label{lemma:threat_model_perf_lb}
Fix a Red team strategy $\pi_{r}(\cdot)$ and let
$\mathcal{S}=\mathcal {S}(\bar{\mu}, \delta)$. The expected payoff for Blue team strategy $\pi_{b}$ under any permissible test-time type distribution $\mu \in \mathcal{S}$ can be lower-bounded in terms of a point-wise type-conditioned benchmark $\mathscr{B}(\mu)$ evaluated at the test-time type distribution and a worst-case expected regret offset $\mathscr{R}(\pi_{b})$ incurred by the decision $\pi_{b}$ over $\mathcal{S}$ as follows,
\begin{align}\label{eq:prmre_perf_lb}
    \mathbb{E}_{\theta \sim \mu}[V(\pi_{r}(\cdot), \pi_{b}; \theta] \geq \mathcal{V}_{\mathcal{S}}^{(-)}(\pi_{b}, \mu) =  \underbrace{\sum_{\theta \in \Theta} \mu(\theta) \sup_{\pi^{*} \in \Pi_{\mathrm{b}}} V(\pi_{r}, \pi^{*}; \theta)}_{\mathscr{B}(\mu)} - \underbrace{\sup\limits_{q \in \mathcal{S}(\bar{\mu}, \delta)} \left( \sum\limits_{\theta \in \Theta} q(\theta) \mathcal{E}(\pi_{r}(\cdot), \pi_{b}; \theta) \right)}_{\mathscr{R}(\pi_{b})}.
\end{align}
\end{lemma}}
\end{tcolorbox}
It is straightforward to prove the following Lemma that the $\operatorname{PR-MRE}$ robust best-response from (\ref{eq:prmre_br_def}) is equivalent to (\ref{eq:mre_br_def_expanded}) through a water-filling argument by arguing that the inner supremum for any outer $\pi \in \Pi_{b}$ corresponds to a type distribution that places all permissible mass on the highest regret type within the high-confidence $\Theta / \Theta_{r}$ subset while respecting the remaining constraints posed by the low-probability covers. Lemma~\ref{lemma:mmr} shows that the $\operatorname{PR-MRE}$ robust best-response weighs regret over rare types by a maximum of $\delta$ in contrast to the distribution-free $\operatorname{MRE}$ (\ref{eq:mre_br_def}) that treats all types as equally likely and minimizes worst-case regret across the type-space.
\begin{lemma}\label{lemma:mmr}
For a given Red team strategy $\pi_{r}(\cdot)$, the probabilistically-robust MRE best-response operator for the Blue team is defined as follows,
\begin{align}\label{eq:mre_br_def_expanded}
\operatorname{BR}^{\mathrm{(b)}}_{\mathrm{PRMRE}}(\pi_{r}(\cdot); \bar{\mu}, \delta) =  \arginf_{\pi \in \Pi_{\mathrm{b}}} \sup_{\substack{ \mu(\Theta / \Theta_{r}) + \sum_{\theta \in \Theta_{r}} \mu(\theta) = 1, \\ \sum\limits_{\theta \in \mathscr{L}} \mu(\theta) \leq \delta \ \forall \ \mathscr{L} \in \Theta_{\ell}(\bar{\mu}, \delta)}} \ \mu(\Theta / \Theta_{r})\cdot\mathcal{E} \left(\pi_{r}(\cdot), \pi; \Theta / \Theta_{r} \right) + \sum_{\theta \in \Theta_{r}}  \mu(\theta) \cdot \mathcal{E} \left(\pi_{r}(\cdot), \pi; \theta \right),
\end{align}
where $\mathcal{E}(\pi_{r}(\cdot), \pi; \Theta / \Theta_{r})$ refers to the worst-case exploitability of the Blue team strategy $\pi$ over the type-space subset $\Theta / \Theta_{r}$  where $\Theta_{r}$ is the set of all rare types such that under the nominal distribution $\bar{\mu}(\cdot)$, $ \bar{\mu}(\theta) \leq \delta, \ \forall \ \theta \in \Theta_{r}$.
\end{lemma}
\begin{tcolorbox}[
title=Probabilistically-Robust Minimax-Regret Equilibrium (PR-MRE),
colback=white,
colframe=black!40,
colbacktitle=white,
coltitle=black
]
The objective in (\ref{eq:prmre_br_def}) minimizes worst-case exploitability on a high-confidence subset of the type space while discounting exploitability on atypical types by the confidence parameter $\delta$. From the novel Blue team best-response operator (\ref{eq:prmre_br_def}) as proposed above, we obtain the equilibrium concept -- Probabilistically-Robust Minimax-Regret equilibria (PR-MRE) characterized by the following fixed-point system,
\begin{align}
    \label{eq:prmre_fixed_point}
    \pi^{*}_{b, \mathrm{PRMRE}}(\bar{\mu}, \delta) \in \operatorname{BR}^{\mathrm{(b)}}_{\mathrm{PRMRE}}(\pi^{*}_{r}(\cdot); \bar{\mu}, \delta) \ \text{and} \ \pi^{*}_{r}(\cdot) \in \operatorname{BR}^{\mathrm{(r)}}(\pi^{*}_{b, \mathrm{PRMRE}}(\bar{\mu}, \delta); \cdot).
\end{align}
\end{tcolorbox}

\begin{tcolorbox}[
title=Existence of PR-MRE for General Asymmetric Information Games,
colback=white,
colframe=black!40,
colbacktitle=white,
coltitle=black
]
{\color{black}
Let $\Theta$ be finite. For each type
$\theta$, let the informed player's strategy space
$\Pi_{r} (\theta)$ and the uninformed player's strategy space $\Pi_b$
be nonempty, compact, convex subsets of finite-dimensional Euclidean
spaces. Let
$\Pi_r=\prod\limits_{\theta\in\Theta}\Pi_{r} (\theta)$, and let the admissible
distribution set $\mathcal S(\bar\mu,\delta)$ be nonempty and compact.
Assume type-conditioned payoffs $V(\pi_r (\cdot),\pi_b; \theta)$ are continuous, concave in Blue's strategy $\pi_{b}$ and convex in Red's strategy $\pi_{r}(\theta)$ for all types $\theta$. Then a PR-MRE as given by the fixed-point conditions of Eq.~(\ref{eq:prmre_fixed_point}) exists.
\begin{proof}
    Proof of existence in Appendix~\ref{sec:appendix_proof_existence}.
\end{proof}
}
\end{tcolorbox}
{\color{black} The above only proves existence of $\operatorname{PR-MRE}$ fixed-points for asymmetric information games based on the stated assumptions on the strategy spaces, $V(\cdot, \cdot, \cdot)$ and $\mathcal{S}(\bar{\mu}, \delta)$ and does not claim uniqueness.} The key robustness guarantee of PR-MRE is given below.
\begin{tcolorbox}[
colback=white,
colframe=black!40,
colbacktitle=white,
coltitle=black
]
{\color{black}
\begin{lemma}\label{lemma:prmre_rob_property}
($\operatorname{PR-MRE}$ Benchmark-Relative Performance Certificate)
Fix a Red team strategy $\pi_r(\cdot)$ and let $\mathcal{S}=\mathcal{S}(\bar{\mu},\delta)$. A $\operatorname{PR-MRE}$ robust best-response minimizes the worst-case expected regret $\mathcal{R}(\pi_b)$ over $\mathcal{S}$ and therefore maximizes the benchmark-relative lower bound
$\mathcal{V}_{\mathcal{S}}^{(-)}(\pi_b,\mu)$ simultaneously for every $\mu\in\mathcal{S}$. In particular, $\mathcal{V}_{\mathcal{S}}^{(-)}(\pi_{b,\mathrm{PRMRE}},\mu)\geq \mathcal{V}_{\mathcal{S}}^{(-)}(\pi_b,\mu)$ for all $\mu\in\mathcal{S}$ and $\pi_b\in\Pi_b$.
\end{lemma}
\begin{proof}
    Proof in Appendix~\ref{sec:appendix_proofs}.
\end{proof}
}
\end{tcolorbox}
{\color{black}
\textbf{Remark}: We remark that the lower-bound of the form $\mathcal V_{\mathcal S}^{(-)}(\pi_b,\mu)$ is a point-wise performance
lower-bound such that $V(\pi_{b}, \mu) \geq \mathcal V_{\mathcal S}^{(-)}(\pi_b,\mu)$ for all $\mu \in \mathcal{S}$. The type-conditioned benchmark in $\mathcal V_{\mathcal S}^{(-)}(\pi_b,\mu)$ varies point-wise
with the test-time type distribution $\mu$ while the certified worst-case regret offset $\inf\limits_{\pi_{b}} \mathscr{R}(\pi_{b})$ is obtained uniformly over the admissible set from PR-MRE. The DRO-MaxMin robust best-response instead provides a uniform lower bound on payoffs such that $V(\pi_{b}, \mu) \geq V^{-}(\pi_{b}) = \inf\limits_{\mu \in \mathcal{S}} V(\pi_{b}, \mu)$ against an adversarial type distribution. The certificate ordering in Lemma~\ref{lemma:prmre_rob_property} does not imply that PR-MRE has higher realized payoff at a given \(\mu\) or a higher worst-case payoff than DRO-MaxMin.}

A corollary of Lemma~\ref{lemma:prmre_rob_property} is that $\operatorname{PR-MRE}$, which minimizes a weighted sum of exploitabilities over high-confidence subsets of the type-space, obtains a tighter {\color{black}benchmark-relative} performance lower bound on permissible distributions under the threat model as compared to $\operatorname{MRE}$, which minimizes worst-case exploitability over the entire type-space, $ \mathcal{V}_{\mathcal{S}}^{(-)}( \pi_{b, \mathrm{PRMRE}}, \mu ) \geq \mathcal{V}_{\mathcal{S}}^{(-)}(\pi_{b, \mathrm{MRE}}, \mu ) \ \forall \ \mu \in \mathcal{S}(\bar{\mu}, \delta)$.

This motivates $\operatorname{PR\text{-}MRE}$ as an appropriate robust equilibrium concept under asymmetric information, providing a tighter {\color{black}benchmark-relative} lower bound on expected performance under admissible distribution shifts than distributionally robust, worst-case, and minimax-regret equilibrium concepts.
In the Graph CtF setting, this corresponds to favoring strategies that remain robust across plausible flag hypotheses while avoiding excessive conservatism induced by highly atypical flag locations.

The tighter lower bound does not come for free, as we establish in Section~\ref{sec:computation} that the {\color{black}approximate} computation of $\operatorname{PR-MRE}$ for a metagame approximation of the ATG can be formulated as a robust bilinear program that couples regrets across types as compared to a simpler bilinear program for $\operatorname{MRE}$ and the well-known linear programming formulation for $\operatorname{BNE}$. This is due to the fact that $\operatorname{PR-MRE}$ reasons about which regions of the type-space to minimize regret over as compared to $\operatorname{MRE}$ that minimizes worst-case regret over the entire type-space. 

\textbf{Programming Formulation}

Unlike BNE, whose best response is a single optimization over expected utility, regret-based reasoning introduces a nested optimization over type specific optima. Consequently, the optimization couples regrets across types, leading to a robust bilinear program rather than the linear programs familiar from Bayesian equilibria. We next derive this formulation and subsequently develop a tractable semidefinite relaxation for use within PSRO.
\begin{tcolorbox}[
colback=white,
colframe=black!40,
colbacktitle=white,
coltitle=black]
From (\ref{eq:prmre_br_def}) and Definition~\ref{def:threat_model}, the PR-MRE best-response operator can be reformulated as follows,

\begin{minipage}[t]{0.45\linewidth}
\begin{align}
\inf_{\pi \in \Pi_{\mathrm{b}}} \sup_{\substack{\mu \in \Delta(\lvert\Theta\rvert) \\ \sum\limits_{\theta \in \mathscr{L}} \mu(\theta) \leq \delta \ \forall \ \mathscr{L} \in \Theta_{l}(\bar{\mu}, \delta) }} \mu^{\top}\varepsilon(\pi)  \notag
\end{align}
\end{minipage}
\hfill\raisebox{-4em}{\Large$\Longleftrightarrow$}\hfill
\begin{minipage}[t]{0.45\linewidth}
\begin{align} 
\inf_{\tilde{\varepsilon}, \pi} \quad &\tilde{\varepsilon} \notag \\ \sup_{\substack{\mu \in \Delta(\lvert\Theta\rvert) \\ \sum\limits_{\theta \in \mathscr{L}} \mu(\theta) \leq \delta \ \forall \ \mathscr{L} \in \Theta_{l}(\bar{\mu}, \delta) }} \hspace{-1em} \mu^{\top}&\varepsilon(\pi)  = \tilde{\varepsilon}(\pi) \leq \tilde{\varepsilon} \label{eq:prmre_br_reformulation}
\end{align}
\end{minipage}
where $\mu^{\top}\varepsilon(\pi)  = \sum_{\theta \in \Theta} \mu(\theta) \mathcal{E} \left(\pi_{r}(\cdot), \pi; \theta \right)$ is the expected regret under $\mu(\cdot)$ and Blue team policy $\pi$.
\end{tcolorbox}
The robust average regret condition couples regrets across types. Moreover, optimizing over regrets is a nested optimization problem as opposed to BNE, distributionally-robust and worst-case notions of best-response since to optimize the regret, deviation from type-wise optimal solutions is minimized. For $\delta = 0$, $ \Theta_{l}(\bar{\mu}, 0) = \Phi$ such that robust average regret under the threat model reduces to worst-case regret corresponding to MRE best-response.
The nested infimum and supremum can be reformulated as shown above. Eq.~(\ref{eq:prmre_br_reformulation}) is a robust formulation for expected regret under type distributions permissible under the threat model.
\begin{tcolorbox}[
title=PR-MRE Best-Response Programming Formulation for Asymmetric Information Games,
colback=white,
colframe=black!40,
colbacktitle=white,
coltitle=black
]
For Red team strategy $\pi_{r}(\cdot)$,
\begin{align}\label{eq:prmre_br_general}
\inf_{\tilde{\varepsilon}, \pi} \quad &\tilde{\varepsilon} \tag{PRMRE-BR} \\ 
\delta \lambda^{\top}\mathbf{1}_{\lvert \mathscr{L} \rvert \times 1} &+ \lambda_{\mathrm{u}} \leq \tilde{\varepsilon}, \\
-\varepsilon^{\top}(\pi) + \lambda^{\top} \mathcal{A}_{\mathscr{L}} &- \lambda^{\top}_{\mathrm{p}} + \lambda_{\mathrm{u}}\mathbf{1} = 0, \\
\lambda \in \mathbb{R}^{\lvert\mathscr{L}\rvert}_{+}, \lambda_{\mathrm{p}} \in  &\mathbb{R}^{\lvert\Theta\rvert}_{+}; \ \lambda_{\mathrm{u}} \in \mathbb{R},
\end{align}
\vspace{-2.75em}
\begin{align}
    \text{where } \varepsilon(\pi) = \left[ \sup\limits_{\pi^{*} \in \Pi_{b}} V(\pi_{r}(\cdot), \pi^{*}; \theta) - V(\pi_{r}(\cdot), \pi; \theta)  \right]_{\theta \in \Theta},
\end{align}
is the regret vector corresponding to Blue team policy $\pi$.
\end{tcolorbox}
{\color{black}
\begin{proposition}\label{prop:prmre_br_general}
    For a fixed Red strategy $\pi_{r}(\cdot)$, \ref{eq:prmre_br_general} is a convex program. The returned optimal solution $\pi^{*}_{b}$ of the above program is an exact best-response to $\pi_{r}(\cdot)$ per Definition~\ref{def:prmre} such that,
    \begin{align}
        \pi^{*}_{b} \in \operatorname{arginf}\limits_{\pi_{b} \in \Pi_{b}} \sup\limits_{\mu \in \mathcal{S}(\bar{\mu}, \delta)} \sum\limits_{\theta \in \Theta} \mu(\theta)\mathcal{E}( \pi_{r}(\cdot), \pi_{b}; \theta).
    \end{align}
\end{proposition}
\begin{proof}
     Proposition follows from linear-programming duality on the epigraph form for the $\operatorname{inf}-\operatorname{sup}$ program and the concavity of $V(\pi_{r}(\cdot), \pi; \theta)$ in Blue's strategy $\pi$ for a fixed Red strategy $\pi_{r}(\cdot)$ for all types.
\end{proof}

While the above program formulates $\operatorname{PR-MRE}$ best-response computation for a fixed Red, computation of a $\operatorname{PR-MRE}$ fixed-point requires solving a coupled system for both Blue and Red best-responses (Red's per-type best-response defined in Definition~\ref{assumption:red_team_rationality}). We instead consider a joint Blue-Red formulation that minimizes worst-case expected regret subject to Red's best-response conditions. Although this relaxation does not guarantee an exact fixed point, it provides a lower bound on equilibrium regret and yields candidate profiles whose $\operatorname{PR-MRE}$ equilibrium residuals can be evaluated independently.
}
\begin{tcolorbox}[
title={\color{black} PR-MRE Joint Red-Blue Optimization Formulation},
colback=white,
colframe=black!40,
colbacktitle=white,
coltitle=black
]
\vspace{-1.25em}
\begin{align}\label{eq:prmre_eqm}
\inf_{\tilde{\varepsilon}, \pi_{\mathrm{PRMRE}}, \pi^{*}_{r}(\cdot)} \quad &\tilde{\varepsilon} \tag{PRMRE-EQM} \\
\delta \lambda^{\top}\mathbf{1}_{\lvert \mathscr{L} \rvert \times 1} + &\lambda_{\mathrm{u}} \leq \tilde{\varepsilon},
\end{align}
\vspace{-2.75em}
\begin{align}
 \varepsilon(\pi_{\mathrm{PRMRE}}) &= \left[ \sup\limits_{\pi^{*} \in \Pi_{b}}  V(\pi^{*}_{r}(\cdot), \pi^{*}; \theta) - V(\pi^{*}_{r}(\cdot), \pi_{\mathrm{PRMRE}}; \theta) \right]_{\theta \in \Theta}, \\
    -\varepsilon^{\top}&(\pi_{\mathrm{PRMRE}}) + \lambda^{\top} \mathcal{A}_{\mathscr{L}} - \lambda^{\top}_{\mathrm{p}} + \lambda_{\mathrm{u}}\mathbf{1} = 0, \\
    V(\pi_{r}(\cdot), \pi_{\mathrm{PRMRE}}; \theta) &\geq V(\pi^{*}_{r}(\cdot), \pi_{\mathrm{PRMRE}}; \theta) \text{ for all } \pi_{r}(\cdot) \in \Pi_{r} \text{ and } \theta \in \Theta.
\end{align}
\end{tcolorbox}
{\color{black}
\begin{proposition}\label{prop:prmre_eqm_lb}
    The returned globally optimal solution $\varepsilon^{*}$ of the program \ref{eq:prmre_eqm} provides a lower-bound on the worst-case expected regret over $\mathcal{S}(\bar{\mu}, \delta)$ for every $\operatorname{PR-MRE}$ fixed-point $(\pi_{r, \mathrm{PRMRE}}(\cdot), \pi_{b, \mathrm{PRMRE}})$ such that,
    \begin{align}\label{eq:eqm_regret_lb}
        \varepsilon^{*} \leq \mathscr{R}(\pi_{r, \mathrm{PRMRE}}(\cdot), \pi_{b, \mathrm{PRMRE}}).
    \end{align}
\end{proposition}
\begin{proof}
    \ref{eq:prmre_eqm} optimizes over the set $ \mathcal{F} \coloneqq \{ (\pi_{r}(\cdot), \pi_{b}) \in \Pi_{r} \times \Pi_{b} \ \vert \ \pi_{r}(\theta) \in \operatorname{BR}^{r}(\pi_{b}, \theta)\}$ as $ \inf_{\mathcal{F}} \mathscr{R}(\pi_{r}(\cdot), \pi_{b})$ where $\mathscr{R}(\pi'_r(\cdot), \pi'_{b}) \coloneqq \sup_{\mu \in \mathcal{S}(\bar{\mu}, \delta)} \sum_{\theta \in \Theta} \mu(\theta) \mathcal{E}( \pi'_r(\cdot), \pi'_{b}; \theta)$ is the worst-case expected regret objective. Moreover, the set of all $\operatorname{PR-MRE}$ fixed-points $\mathcal{F}_{\mathrm{PRMRE}} \coloneqq \{ (\pi_{r}(\cdot), \pi_{b}) \in \Pi_{r} \times \Pi_{b} \ \vert \ \pi_{r}(\theta) \in \operatorname{BR}^{r}(\pi_{b}, \theta), \pi_{b} \in \operatorname{BR}^{b}_{\mathrm{PRMRE}}(\pi_{r}(\cdot), \bar{\mu}, \delta)\}$ is a subset of $ \mathcal{F}$ such that $\mathcal{F}_{\mathrm{PRMRE}} \subseteq \mathcal{F} $. Optimizing worst-case expected regret $\mathscr{R}(\cdot, \cdot)$ over a larger set provides the lower bound (\ref{eq:eqm_regret_lb}).
\end{proof}
Section~\ref{sec:computation} formulates \ref{eq:prmre_br_general} and \ref{eq:prmre_eqm} as linear (exact) and robust bilinear programs (approximation per Proposition~\ref{prop:prmre_eqm_lb}) respectively for finite normal-form games.

\textbf{Equilibrium Residuals}. For any returned solution $(\pi_{r}(\cdot), \pi_{b})$ from the \ref{eq:prmre_eqm} program, gap from a $\operatorname{PR-MRE}$ fixed-point is characterized via the residual, $\eps_{\operatorname{PRMRE}} (\pi_{r}(\cdot), \pi_{b}) = \max(\eps_{r}, \eps_{b})$ where $\eps_{r}$ is the maximum over the per-type Red residuals defined as $\eps_{r}(\theta) = V(\pi_{r}(\cdot), \pi_{b}; \theta) - \inf_{\pi \in \Pi_{r}}V(\pi, \pi_{b}; \theta)$ and $\eps_{b} = \mathscr{R}(\pi_{r}(\cdot), \pi_{b}) - \inf_{\pi \in \Pi_{b}}\mathscr{R}(\pi_{r}(\cdot), \pi)$.
}

\textbf{Illustrative Example}

Consider an example normal-form game to illustrate the distinction between $\operatorname{PR-MRE}$ and the various equilibrium concepts discussed in Section~\ref{sec:atg}. Consider the $1\times5\times3$ game from Table~\ref{tab:1row-game} with type-space $\Theta = \{\theta_{0}, \theta_{1}, \theta_{2}\}$, five Blue actions $b_{0}$ to $b_{4}$ on the columns, a single Red action $r_{0}$ as the row with three rows one for each type and a nominal prior on types $ \bar{\mu} = (0.25, 0.70, 0.05)$. We use a single Red action in this illustrative example so that the informed Red best-response is unique, allowing differences between equilibrium fixed-points to be attributed solely to the Blue team's varied robust best-responses. A more complicated 2×5×3 example with the added complexity of Red's best-response for determining the respective equilibrium fixed-points is provided in Appendix~\ref{sec:illus_example_appendix} (see Table~\ref{tab:2row-game}). From the example game shown below in Table~\ref{tab:1row-game}, action $b_{0}$ is a $\theta_{0}$-specialist such that $V^{*}(\theta_{0}) = V(b_{0}; \theta_{0})$, $b_{1}$ is a $\theta_{1}$-specialist and $b_{2}$ is a $\theta_{2}$-specialist that is rare under the nominal ($\bar{\mu}(\theta_{2}) = 0.05$). Moreover, actions $b_{0}$ and $b_{1}$ are catastrophic on the rare type $\theta_{2}$ and action $b_{2}$ is catastrophic on typical types $\theta_{0}$ and $\theta_{1}$. Actions $b_{3}$ and $b_{4}$ are generalists such that they are not too bad for any particular type. In other words, $b_{3}$ and $b_{4}$ have markedly better worst-case regret over types as compared to the first-three specialist columns, 1.1 and 1.6 against 2.1, 2.3 and 2.0 for $b_{0}$, $b_{1}$ and $b_{2}$ respectively. Moreover, $b_{3}$ is the best generalist with lowest worst-case regret of $1.1$ across types (see column $b_{3}$ in Table~\ref{tab:1row-game}).

From the taxonomy on robust Blue best-responses under asymmetric information and associated equilibria in Tables~\ref{tab:taxonomy} and \ref{tab:rob_properties}, each best-response corresponds to a degree of trust in Nature's nominal type distribution and an optimization objective resulting in a unique robustness property. We now contrast $\operatorname{PR-MRE}$ and the various equilibrium concepts with respect to robustness to shifts from the nominal on the example game and visualize their comparative advantage over different regions of the type-simplex (see Figure~\ref{fig:1row_triangles} and Table~\ref{tab:1row-game}). The nominal type distribution $(0.25, 0.70, 0.05)$ represents a single point in the triangular barycentric representation of the three-dimensional simplex (Figure~\ref{fig:1row_triangles}) and we evaluate the performance (Blue payoff) of the Blue decision corresponding to different robust equilibrium fixed-points as the test type distribution probes over the entire simplex. 

\textbf{Evaluation Protocol}. For our comparison and analysis (Figures~\ref{fig:1row_triangles} and \ref{fig:1row_tables}), the performance of the Blue decision corresponding to different equilibria for any test type distribution is measured against an adaptive Red decision, specifically the Red $\operatorname{BNE}$ corresponding to the test type distribution (Red's decision corresponding to the BNE fixed-point for the test type distribution). Figure~\ref{fig:1row_triangles} shows the winner map and payoff difference heatmaps for $\operatorname{PR-MRE}$ and the different methods against such a Red opponent that adapts as the test type distribution probes over the simplex. For the single Red action example in Table~\ref{tab:1row-game}, the Red decision is unique everywhere by construction and the robust Blue decisions are evaluated against variations in just the Nature's decision / type distribution (Red decision is same everywhere). For the $2\times5\times3$ example in Appendix~\ref{sec:illus_example_appendix}, evaluation and visualization is done against an adaptive Red opponent.

\begin{table}[t]
\caption{Example $1{\times}5{\times}3$ game. Columns $\{b_{0} \cdots b_{4}\}$ are Blue actions, row $r_{0}$ corresponds to the Red action and there are three total rows, one per type $\theta \in \{\theta_{0}, \theta_{1}, \theta_{2}\}$ with a nominal distribution $\bar{\mu} = (0.25, 0.70, 0.05)$. Each cell shows the Blue payoff $V(r, b; \theta)$, shorthand $V(b; \theta)$ for the unique Red choice $r_{0}$, the type-conditioned optimum values $V^{*}(\theta)$ read left to right per $\theta$ row are marked in bold (black), and the
per-type regret for any Blue action (or column) given by  $\varepsilon(b; \theta) = V^{*}(\theta) - V(b; \theta)$ is in red subscript. Worst-case regret over types $\sup_{\theta \in \Theta} \varepsilon(b; \theta)$ for any Blue action $b$ is read top to bottom per column and marked in bold (red).}
\vspace{-1.5em}
\centering
\[
\setlength{\arraycolsep}{6pt}
\renewcommand{\arraystretch}{1.3}
\begin{array}{r |c|c|c|c|c| l}
 \multicolumn{1}{r}{}
   & \multicolumn{1}{c}{b_0} & \multicolumn{1}{c}{b_1} & \multicolumn{1}{c}{b_2}
   & \multicolumn{1}{c}{b_3} & \multicolumn{1}{c}{b_4} & \multicolumn{1}{l}{} \\[4pt]
\cline{2-6}
 r_0 &  \textbf{1.5}_{\ \color{red}{0.0}} &  0.6_{\ \color{red}{0.9}} & -0.5_{\ \color{red}{\textbf{2.0}}} &  0.4_{\ \color{red}{\textbf{1.1}}} &  0.7_{\ \color{red}{0.8}} & \big\}\;\theta_0, 0.25 \\
\cline{2-6}
\noalign{\vskip 7pt}
\cline{2-6}
 r_0 &  0.3_{\ \color{red}{0.7}} &  \textbf{1.0}_{\ \color{red}{0.0}} & -0.8_{\ \color{red}{1.8}} &  0.8_{\ \color{red}{0.2}} &  0.8_{\ \color{red}{0.2}} & \big\}\;\theta_1, 0.70 \\
\cline{2-6}
\noalign{\vskip 7pt}
\cline{2-6}
 r_0 & -0.6_{\ \color{red}{\textbf{2.1}}} & -0.8_{\ \color{red}{\textbf{2.3}}} &  \textbf{1.5}_{\ \color{red}{0.0}} &  0.5_{\ \color{red}{1.0}} & -0.1_{\ \color{red}{\textbf{1.6}}} & \big\}\;\theta_2, 0.05 \\
\cline{2-6}
\end{array}
\]
\label{tab:1row-game}
\end{table}

\begin{figure}[h] 
  \centering
  \includegraphics[scale=0.25]{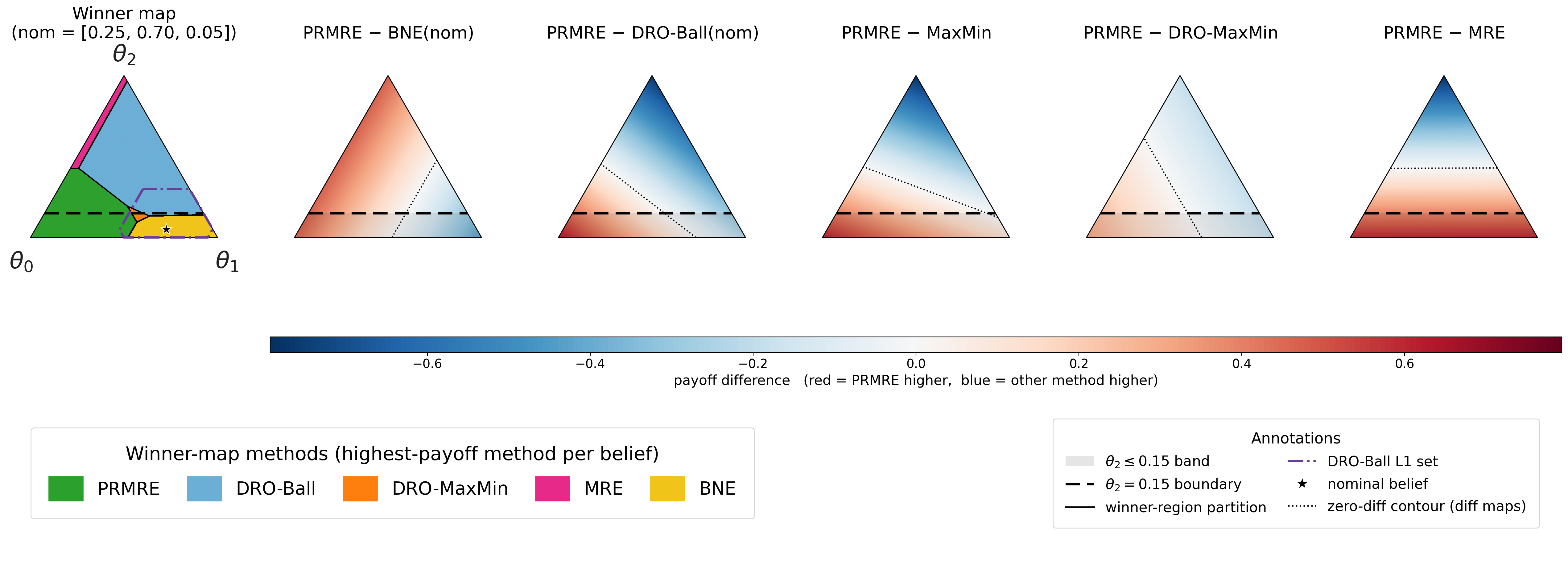}
  \vspace{-1.5em}
  \caption{Comparison of $\operatorname{PRMRE}$ ($\delta = 0.15$) with $\operatorname{BNE}$, $\operatorname{DRO-Ball}$, $\operatorname{MaxMin}$, $\operatorname{DRO-MaxMin}$ and $\operatorname{MRE}$  for the synthetic three-type game in Table~\ref{tab:1row-game}: winner map (first triangle) and payoff difference heatmaps (remaining triangles). Triangles represent three-dimensional simplex in barycentric coordinates with vertices corresponding to full probability mass on the respective types. Black dotted line and the trapezoidal band beneath it represent the high-confidence under nominal structured type-ambiguity set $\mu(\theta_{2}) \leq 0.15$, the star in the winner map corresponds to the nominal distribution $\bar{\mu} = (0.25, 0.70, 0.05)$ and the dotted hexagon centered at the star corresponds to an $L1$-ball with radius $\rho=0.5$ around it. For the trapezoidal band, $\operatorname{PR-MRE}$ dominates all other methods which can be seen from the green volume share in the winner-map and the red volume share in the payoff difference heatmaps. $\operatorname{BNE}$ and $\operatorname{DRO-Ball}$ perform better close to the nominal ($\theta_{1}$-corner, mass 0.70 under nominal) but lose out on performance to $\operatorname{PR-MRE}$ for a type distribution chosen at random within the trapezoidal band. $\operatorname{PR-MRE}$ has the best average-case performance over the structured ambiguity set as compared to any other method (see Figure~\ref{fig:1row_tables} for the numerical result, consistent with visual inspection of the winner-map and payoff difference heatmaps).}
\label{fig:1row_triangles}
\end{figure}
\textbf{Remark}. We note that we compare $\operatorname{PR-MRE}$ to two types of distributionally-robust equilibria, $\operatorname{DRO-Ball}$ based on an $L1$-ball around the nominal (local trust) and $\operatorname{DRO-MaxMin}$ based on the structured type ambiguity induced by the threat model in Definition~\ref{def:threat_model} (coarse trust) that preserves the set of all high-probability covers under admissible distribution shifts. This provides a direct comparison of our regret-based approach  $\operatorname{PR-MRE}$  to the payoff-based distributionally-robust approach $\operatorname{DRO-MaxMin}$  over the same structured ambiguity set.

$\operatorname{BNE}$, $\operatorname{DRO-Ball}$, $\operatorname{MaxMin}$ and $\operatorname{DRO-MaxMin}$ are all computed via linear-programs whereas $\operatorname{MRE}$ and $\operatorname{PR-MRE}$ ($\delta = 0.15$) are computed via semidefinite relaxations to the corresponding {\color{black} surrogate} bilinear and robust bilinear programs as derived in Section~\ref{sec:computation} for normal-form asymmetric information games. 

{\color{black} \textbf{Equilibrium residual of the recovered candidate PR-MRE strategy}. Note that the program \ref{eq:prmre_eqm_normal_form} does not return exact $\operatorname{MRE}$ and $\operatorname{PR-MRE}$ fixed-points and all the results discussed here are for candidate strategy profiles recovered from a convex relaxation of a lifted rank-constrained semidefinite reformulation of the surrogate program (refer to \ref{eq:prmre_eqm_normal_form} and Propositions~\ref{prop:robust_bilinear_sufficiency}--\ref{prop:rank_one_recovery_normal_form_games} in Section~\ref{sec:computation}). We characterize the gap of the recovered strategy from a $\operatorname{PR-MRE}$ fixed-point by computing the equilibrium residual and remark the two sources of the residual as: the \ref{eq:prmre_eqm_normal_form} program is a surrogate and does not certifiably return exact fixed-points. Moreover, the convex-relaxation of the rank-1 constrained surrogate semidefinite program could itself be lossy. For the example game in Table~\ref{tab:1row-game}, the recovered semidefinite matrix from our convex relaxation (details in Appendix~\ref{sec:prmre_sdp_relax}) resides in $\mathcal{S}^{+}_{7}$ (lifted space, $\mathcal{S}^{+}_{\lvert\mathcal{X}\rvert + \lvert\mathcal{Y}\rvert + 1}$ psd cone for finite normal-form asymmetric information games with Red and Blue strategy spaces $\mathcal{X}$ and $\mathcal{Y}$ respectively) is rank-$3$. The three leading eigenvalues are approximately $(2.540,0.188,0.135)$, with the remaining four below $8\times10^{-7}$. The computed Red residual $\eps_{r}$ is zero since the game has only one Red action. To compute the Blue residual, we fix the extracted Red strategy and solve the PR-MRE robust best-response problem exactly (see Proposition~\ref{prop:prmre_br_nfg_exact} that states that the formulated program returns exact robust best-responses). The extracted Blue strategy has worst-case expected regret $0.641137$, versus the exact robust best-response value $0.640769$, giving a residual $\eps_{b}$ of $3.7\times10^{-4}$. Thus, although the relaxed SDP solution is not rank one, the extracted strategy satisfies the PR-MRE fixed-point condition approximately, with a residual of $\operatorname{max}(\eps_{r}, \eps_{b}) = 3.7\times10^{-4}$.
}

\textbf{Results}. For the example game in Table~\ref{tab:1row-game}, the computed $\operatorname{BNE}$ Blue decision concentrates on the specialist $b_{1}$ for the highest-confidence type $\theta_{1}$ under nominal $\bar{\mu} = (0.25, 0.70, 0.05)$ as $(0, 1, 0, 0, 0)$ whereas the computed $\operatorname{DRO-Ball}$ decision for an $L1$-ball of radius $\rho = 0.5$ around the nominal (dotted hexagon around the nominal star in Figure~\ref{fig:1row_triangles}) hedges probability mass across generalist actions $b_{3}, b_{4}$ as $(0, 0, 0, 0.89, 0.11)$. $\operatorname{MRE}$ that discards probabilistic information from the nominal and minimizes worst-case regret across all types results in a decision $(0.39, 0, 0.39, 0.22, 0)$ with $0.22$ mass to the best generalist $b_{3}$ and a significant probability mass of $0.39$ to the rare-type specialist $b_{2}$. This comes at the expense of performance over typical types $\theta_{0}$ and $\theta_{1}$ as the poorest worst-case and average payoffs over the high-confidence under nominal structured type ambiguity set $\mu(\theta_{2}) \leq 0.15$ (the black dotted line and the trapezoidal band beneath it in Figure~\ref{fig:1row_triangles}). The $\operatorname{MaxMin}$ decision that assumes an adversarial type distribution (Section~\ref{sec:wc_eqm}) and optimizes for worst-case payoff across the type-simplex is the only other decision besides $\operatorname{MRE}$ that puts mass on the rare-type specialist $b_{2}$ as $(0.18, 0, 0.16, 0.66, 0)$. In contrast to distribution-free $\operatorname{MRE}$, $\operatorname{PR-MRE}$ utilizes coarse information from the nominal and optimizes for worst-case regret across a structured type ambiguity set that respects typicality of types. For $\delta = 0.15$, the $\operatorname{PR-MRE}$ best-response discounts regret incurred on the rare-type $\theta_{2}$  (nominal mass $0.05$) by $0.15$ at maximum (follows from Lemma~\ref{lemma:mmr}) such that the $\operatorname{PR-MRE}$ decision for our example game is $(0.46, 0, 0, 0, 0.54)$.

\begin{figure}[ht]
  \centering
  \includegraphics[scale=0.3]{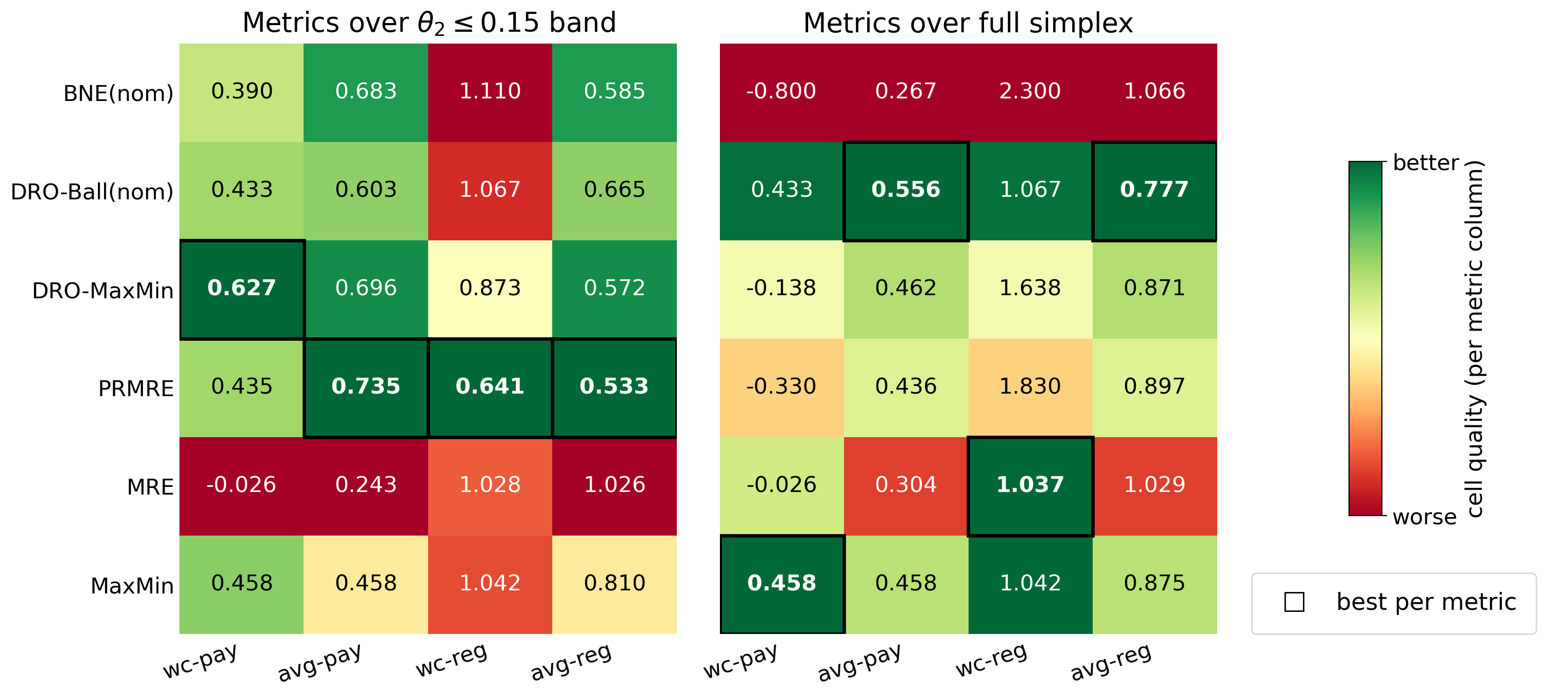}
  \caption{Comparison of $\operatorname{PRMRE}$ ($\delta = 0.15$) with $\operatorname{BNE}$, $\operatorname{DRO-Ball}$, $\operatorname{MaxMin}$, $\operatorname{DRO-MaxMin}$ and $\operatorname{MRE}$  for the synthetic three-type game in Table~\ref{tab:1row-game} with respect to average and worst-case, payoff and regret over the structured ambiguity set $\mu(\theta_{2}) \leq 0.15$ and the entire type-simplex. Over the $\mu(\theta_{2}) \leq 0.15$ trapezoidal band (see Figure~\ref{fig:1row_triangles}), $\operatorname{PR-MRE}$ has the best worst-case regret as compared to any other method. $\operatorname{DRO-MaxMin}$ on the other hand beats $\operatorname{PR-MRE}$ on worst-case payoff over the $\mu(\theta_{2}) \leq 0.15$ band which is consistent with its payoff-based objective (0.627 vs 0.435) but loses average-case performance for a type distribution chosen at random from the trapezoidal band such that $\operatorname{PR-MRE}$ attains the highest average payoff (0.735) as compared to any other method (0.696 for $\operatorname{DRO-MaxMin}$). The regret-based objective of the $\operatorname{PR-MRE}$ best-response provides the best point-wise lower bound on payoff for type distributions belonging to the structured ambiguity set (as stated in Lemma~\ref{lemma:prmre_rob_property}) which in this example game manifests as a superior average-case performance for a type-distribution chosen at random from the $\mu(\theta_{2}) \leq 0.15$ structured ambiguity set.}
  \label{fig:1row_tables}
\end{figure}
The $\operatorname{PR-MRE}$ decision favours $b_{4}$ over $b_{3}$ despite having a higher worst-case regret over the entire type-space (1.6 vs 1.1 for $\operatorname{MRE}$) due to higher payoff and lower worst-case regret as compared to $b_{3}$ when restricted to typical types $\theta_{0}$ and $\theta_{1}$. Blue action $b_{3}$ is the best overall generalist and $b_{4}$ is a more specialized generalist that performs better over regions which matter probabilistically. The payoff-based $\operatorname{DRO-MaxMin}$ decision $(0.08, 0, 0, 0, 0.92)$ concentrates over an adversarial type distribution within the structured type ambiguity set as compared to the regret-based $\operatorname{PR-MRE}$ that hedges across the $\mu(\theta_{2}) \leq 0.15$ trapezoidal band by having roughly equal mass $0.46$ to the $\theta_{0}$-corner that constitutes the left side of the band in Figure~\ref{fig:1row_triangles}. From Figure~\ref{fig:1row_tables}, $\operatorname{PR-MRE}$ has the best average-case performance over the structured ambiguity set as compared to any other method which is consistent with a visual inspection of the green share in the winner-map and the red share in the payoff difference heatmaps in Figure~\ref{fig:1row_triangles}. Overall, the example game illustrates that the taxonomy in Table~\ref{tab:taxonomy}, defined by the optimization objective and degree of trust in Nature's nominal type distribution, induces qualitatively different robust Blue best-responses and corresponding equilibria, resulting in distinct robustness properties and complementary regions of advantage across the type simplex. We refer the reader to Appendix~\ref{sec:illus_example_appendix} for results on a more complicated $2\times5\times3$ example game that extends the illustrative $1\times5\times3$ game by restoring the informed Red player's strategic choices when determining the corresponding robust equilibrium fixed points.

\subsection{Computation of Probabilistically-Robust Minimax-Regret Equilibria for Normal-Form Games}
\label{sec:computation}
Computing equilibria in graph-structured adversarial team games is already intractable due to the size of the extensive-form representation, necessitating approximation methods such as CFR, PPO, and double-oracle approaches. In contrast, PR-MRE introduces a regret-based robustness objective over structured type uncertainty, which does not admit standard sequence-form linear programming formulations or straightforward extensions of CFR-style regret updates. 
To address this challenge, we derive a tractable optimization formulation of PR-MRE together with a semidefinite relaxation that can be embedded within a population-based learning framework.

\textbf{For Finite Normal-Form Games}
\begin{tcolorbox}[
title=PR-MRE Best-Response Programming Formulation for Normal-Form Games,
colback=white,
colframe=black!40,
colbacktitle=white,
coltitle=black
]
For a finite normal-form game defined by the payoff tensor $\{ A(\cdot, \cdot; \theta) \}_{\theta \in \Theta} \in \mathbb{R}^{\lvert\mathcal{X}\rvert \times \lvert\mathcal{Y}\rvert \times \lvert\Theta\rvert}$ and the fixed omniscient Red strategy $x(\cdot): \Theta \rightarrow \Delta(\lvert\mathcal{X}\rvert)$, the Blue PR-MRE best-response $y_{\mathrm{PRMRE}} \in \Delta(\lvert\mathcal{Y}\rvert)$ denoted by $y_{\mathrm{PRMRE}} \longleftarrow \operatorname{PRMRE-BR}\{A, x(\cdot)\}$ can be formulated as the following program,
\begin{align}
\inf_{\tilde{\varepsilon}, y_{\mathrm{PRMRE}}} \quad &\tilde{\varepsilon} \tag{PRMRE-BR-NF} \label{eq:prmre_br_normal_form} \\ 
\delta \lambda^{\top}\mathbf{1}_{\lvert \mathscr{L} \rvert \times 1} &+ \lambda_{\mathrm{u}} \leq \tilde{\varepsilon}, \\
x(\theta)A(\theta)e_{j} - x(\theta)A(\theta)y_{\mathrm{PRMRE}} \leq \{ \lambda^{\top} \mathcal{A}_{\mathscr{L}} &- \lambda^{\top}_{\mathrm{p}} + \lambda_{\mathrm{u}}\mathbf{1} \}_{\theta} \text{ for all } j \in \left[\lvert\mathcal{Y}\rvert\right] \text{ and } \theta \in \Theta, \\
\lambda \in \mathbb{R}^{\lvert\mathscr{L}\rvert}_{+}, \lambda_{\mathrm{p}} \in  \mathbb{R}^{\lvert\Theta\rvert}_{+}, \ &\lambda_{\mathrm{u}} \in \mathbb{R}, y_{\mathrm{PRMRE}} \in \Delta(\lvert\mathcal{Y}\rvert).
\end{align}
\end{tcolorbox}
{\color{black}
\begin{proposition}\label{prop:prmre_br_nfg_exact}
    For a finite normal-form game and a fixed Red strategy $x(\cdot)$, the returned optimal solution $y^{*}$ of the \ref{eq:prmre_br_normal_form} linear program is an exact $\operatorname{PR-MRE}$ best-response to $x(\cdot)$ per Definition~\ref{def:prmre} such that,
    \begin{align}
        y^{*} \in \operatorname{arginf}\limits_{y \in \Delta(\lvert\mathcal{Y}\rvert)} \sup\limits_{\mu \in \mathcal{S}(\bar{\mu}, \delta)} \sum\limits_{\theta \in \Theta} \mu(\theta)\mathcal{E}( x(\cdot), y; \theta).
    \end{align}
\end{proposition}
\begin{proof}
    Follows from Proposition~\ref{prop:prmre_br_general} for general asymmetric information games.
\end{proof}
For approximate joint Red-Blue strategy computation, we adapt \ref{eq:prmre_eqm} into the following robust bilinear program for finite normal-form games.
}
\begin{tcolorbox}[
title={\color{black} PR-MRE Joint Blue–Red Optimization Formulation},
colback=white,
colframe=black!40,
colbacktitle=white,
coltitle=black
]
{\color{black} We consider the following joint Blue–Red optimization program, whose globally optimal value lower-bounds the worst-case expected regret of every PR-MRE fixed point.}
\begin{align}
\inf_{\tilde{\varepsilon}, y, x(\cdot)} \quad &\tilde{\varepsilon} \tag{PRMRE-EQM-NF} \label{eq:prmre_eqm_normal_form} \\ 
\delta \lambda^{\top}\mathbf{1}_{\lvert \mathscr{L} \rvert \times 1} &+ \lambda_{\mathrm{u}} \leq \tilde{\varepsilon}, \\
x(\theta)A(\theta)e_{j} - x(\theta)A(\theta)y \leq \{ \lambda^{\top} \mathcal{A}_{\mathscr{L}} &- \lambda^{\top}_{\mathrm{p}} + \lambda_{\mathrm{u}}\mathbf{1} \}_{\theta} \text{ for all } j \in \left[\lvert\mathcal{Y}\rvert\right] \text{ and } \theta \in \Theta, \label{eq:y_switch} \\
-x(\theta)A(\theta)y \geq -e_{k} A(\theta) &y \text{ for all } k \in [\lvert\mathcal{X}\rvert] \text{ and } \theta \in \Theta, \label{eq:x_switch} \\
\lambda \in \mathbb{R}^{\lvert\mathscr{L}\rvert}_{+}, \lambda_{\mathrm{p}} \in  \mathbb{R}^{\lvert\Theta\rvert}_{+}; \ \lambda_{\mathrm{u}} \in \mathbb{R}, y &\in \Delta(\lvert\mathcal{Y}\rvert), x(\theta) \in \Delta(\lvert\mathcal{X}\rvert) \text{ for all } \theta \in \Theta.
\end{align}
\end{tcolorbox}
{\color{black}
\begin{proposition}\label{prop:robust_bilinear_sufficiency}
    The returned globally optimal solution $\varepsilon^{*}$ of the robust bilinear program \ref{eq:prmre_eqm_normal_form} provides a lower-bound on the worst-case expected regret over $\mathcal{S}(\bar{\mu}, \delta)$ for every $\operatorname{PR-MRE}$ fixed-point $(x_{\mathrm{PRMRE}}(\cdot), y_{\mathrm{PRMRE}})$ for the finite normal-form game such that,
    \begin{align}\label{eq:eqm_regret_lb_nfg}
        \varepsilon^{*} \leq \mathscr{R}(x_{\mathrm{PRMRE}}(\cdot), y_{\mathrm{PRMRE}}).
    \end{align}
\end{proposition}
\begin{proof}
    Follows from Proposition~\ref{prop:prmre_eqm_lb}.
\end{proof}
}
\textbf{Rank-constrained Semidefinite Reformulation of \ref{eq:prmre_eqm_normal_form}}.
The details on the rank-constrained semidefinite formulation and convex relaxation are relegated to Section~\ref{sec:prmre_sdp_relax}.
{\color{black}
\begin{proposition}
    The rank-constrained semidefinite program (SDP) developed in Appendix~\ref{sec:prmre_sdp_relax} is an exact reformulation of the robust bilinear program \ref{eq:prmre_eqm_normal_form} such that an exact solution of the non-convex SDP recovers an exact solution of \ref{eq:prmre_eqm_normal_form}.
\end{proposition}
\begin{proposition}\label{prop:rank_one_recovery_normal_form_games}
    Any rank-1 solution of the convex relaxation detailed in Appendix~\ref{sec:prmre_sdp_relax} returns an exact solution of the non-convex SDP and hence recovers an exact solution of the robust bilinear program \ref{eq:prmre_eqm_normal_form}.
\end{proposition}
\begin{proof}
    This follows from the eigen decomposition of a rank-1 positive semidefinite matrix. See Proposition~2.1 in Ref.~\cite{ali_zhang_bimatrix_nash_sdp} for a related result on Nash equilibria in a bimatrix game with perfect information.
\end{proof}
We emphasize that rank-one recovery certifies global optimality for PRMRE-EQM-NF, but does not necessarily imply that the recovered strategy profile constitutes a PR-MRE fixed point, as discussed before Proposition~\ref{prop:prmre_eqm_lb}. 
}
\subsection{{\color{black} PR-MRE Robust Double Oracle for Asymmetric Information Games}}
Double-oracle (DO) \cite{mcmahan2003planning} is an algorithm for computing Nash Equilibrium (NE) in two-player zero-sum normal-form games. The algorithm proceeds by maintaining a population of strategies for each player, referred to as a metagame that abstracts away the full game, computes a NE over the metagame and iteratively adds strategies to the player populations by best-responding to the opponent meta-equilibria. The DO approach in principle recovers the NE for the full game. Policy Space Response Oracles (PSRO) \cite{psro_silver} extends DO to large games by employing reinforcement-learning to compute best-responses. The metagame NE is computed on the empirical game matrix derived from the matchup between each policy in the population set against all opponent policies and tracking the average payoff in a metagame payoff matrix.

As discussed in Section~\ref{sec:eqm_concepts_section}, NE (BNE) is risk-neutral and offers no robustness and performance guarantees under strategic distribution shifts for asymmetric information games such as Graph Capture-the-Flag where Nature's type is hidden from the Blue team and Red can condition it's play on the realized type and / or potentially control the type distribution. For the DO- and PSRO- policy expansion steps, best-responding to a NE (BNE) meta-equilibria over the opponent policy set might add brittle team policies to the population that suffer under distribution shift. Existing work on incorporating robustness at the metagame level and risk-aware policy expansion is limited \cite{slumbers_rae_psro} and does not formally reason about exploitability of the population under hidden information.

\todo{brittle to distribution shift NE expansion; incorporating risk such as Slumbers; no formal methodology that reasons about exploitability and exploitability minimazation based expansion. MRE and PR-MRE are concepts that formalize robustness to distribution shifts in the form of uniform performance lower bounds over the simplex and typical-preserving distribution shifts respectively.}
\begin{algorithm}[t]
\caption{{\color{black} Exact PR-MRE Robust Double Oracle for Finite Games}}
\label{alg:doubleoracle}
\begin{algorithmic}[1]
\State \textbf{Result:} Probabilistically-Robust Minimax-Regret Equilibrium
\State \textbf{Input:} $\Theta$, $\bar{\mu}$, $\delta$, Initial population $\Pi_{0} = (\mathcal{X}_{0}, \mathcal{Y}_{0})$, Normal-form game payoff tensor $A_{\mathcal{X} \times \mathcal{Y}} \in \mathbb{R}^{\lvert\mathcal{X}\rvert \times \lvert\mathcal{Y}\rvert \times \lvert\Theta\rvert}$

\Repeat
    \State $\blacktriangleright$ {\color{black} Compute an exact $\operatorname{PR-MRE}$ fixed-point $(x^{{k}}(\cdot), y^{{k}})$ for the current metagame $A_{\mathcal{X}^{{k}} \times  \mathcal{Y}^{{k}}}$ with a full-game benchmark $\mathcal{Y}$}
    \For{$i \in \{\mathrm{R}, \mathrm{B}\}$}

        \If{$i \text{ is } \mathrm{B}$}
            \State $\blacktriangleright$ {\color{black} Compute an exact $\operatorname{PR-MRE}$ best-response $y^{{k + 1}} \gets \text{\ref{eq:prmre_br_normal_form}}\{A_{\mathcal{X}^{{k}} \times \mathcal{Y}}, x^{{k}}(\cdot)\}$  with a full-game benchmark $\mathcal{Y}$}
            \State $\blacktriangleright$ $\pi_{i} \gets \operatorname{supp}(y^{{k + 1}}) \setminus \mathcal{Y}^{{k}}$
            
        \Else
        \State $\blacktriangleright$ $ x^{{k+1}}(\theta) \gets \operatorname{BR}\{ A_{\mathcal{X} \times \mathcal{Y}^{{k}}}, y^{{k}}\} \text{ for all } \theta \in \Theta$
        \State $\blacktriangleright$ $\pi_{i} \gets \left\{\bigcup\limits_{\theta \in \Theta}\operatorname{supp}(x^{{k + 1}}(\theta))\right\} \setminus \mathcal{X}^{{k}}$
        \EndIf
        

        \State $\Pi_i^{{k}+1} \gets \Pi_i^{k} \cup \{\pi_i\}$
    \EndFor

\Until{no novel best response exists for either player}

\State \Return $x^{{k_{\mathrm{final}}}}(\cdot), y^{{k_{\mathrm{final}}}}$
\end{algorithmic}
\end{algorithm}
$\operatorname{MRE}$ and $\operatorname{PR-MRE}$ are regret-based notions of equilibria that provide robustness guarantees in the form of {\color{black} point-wise benchmark-relative performance certificate} (see Lemmas~\ref{lemma:ulb_mre} and \ref{lemma:threat_model_perf_lb}).
\todo{MRE and PRMRE} We propose the robust double-oracle method (RDO) to compute MRE and PR-MRE for large normal-form asymmetric information games. {\color{black} $\operatorname{PR-MRE\text{ }RDO}$ (see Algorithm~\ref{alg:doubleoracle}) is based on exact $\operatorname{PR-MRE}$ fixed-points as the restricted-game meta-equilibria and \ref{eq:prmre_br_normal_form} as the best-response for policy set expansion of the uninformed player}. \ref{eq:prmre_br_normal_form} {\color{black}(see Proposition~\ref{prop:prmre_br_nfg_exact}, exact best-response program)} provides a principled best-response that adds strategies with minimum worst-case exploitability against the omniscient opponent type-conditioned meta-equilibria. The proposed RDO expansion results {\color{black} in a robust population that supports the $\operatorname{PR-MRE}$ of the underlying normal-form game} and is stable under distribution shifts of the hidden type. {\color{black} The RDO procedure terminates in a finite number of iterations such that an exact $\operatorname{PR-MRE}$ meta-equilibria over the restricted metagame (computed w.r.t full-game benchmarks) at the final iteration is a $\operatorname{PR-MRE}$ fixed-point of the full-game as well.}
{\color{black}
\begin{proposition}[Finite convergence of PR-MRE RDO with a full-game benchmark]
\label{prop:rdo-finite-convergence}
Let $\mathcal{X}$, $\mathcal{Y}$, and $\Theta$ be finite, and let
$\mathcal{S}(\bar{\mu},\delta)$ be a nonempty compact set of
type distributions. For a type-conditioned Red strategy profile
$x(\cdot)$, define the Blue benchmark $b_\theta(x(\cdot), \mathcal{Y}')$ for any $\mathcal{Y}' \subseteq \mathcal{Y}$ and the associated robust regret objective $\mathscr{R}(x(\cdot),y, \mathcal{Y}')$ as follows,
\begin{align}
    b_\theta(x(\cdot), \mathcal{Y}')
    &:= \sup_{y^{*}\in\Delta(\mathcal{Y}')}
        V(x(\cdot), y^{*}; \theta), \\
    \mathscr{R}(x(\cdot),y, \mathcal{Y}')
    &:= \sup_{\mu\in\mathcal{S}(\bar{\mu},\delta)}
        \sum_{\theta\in\Theta}\mu(\theta)
        \bigl[b_\theta(x(\cdot), \mathcal{Y}')-V(x(\cdot),y; \theta)\bigr],
\end{align}
Moreover, we denote the full-game benchmark and the associated robust regret objective as $b_\theta(x(\cdot))$ and $\mathscr{R}(x(\cdot), y)$ respectively.
Consider robust double oracle with nonempty initial populations
$\mathcal{X}_0\subseteq \mathcal{X}$ and $\mathcal{Y}_0\subseteq \mathcal{Y}$. At each iteration $k$,
suppose $(x^k(\cdot),y^k)$ be the $\operatorname{PR-MRE}$ fixed point for the $( \Pi_{\theta \in \Theta} \mathcal{X}_{k}, \mathcal{Y}_{k})$ restricted game with a full-game benchmark satisfying,
\begin{align}
    x^k(\theta)
    &\in \arginf_{x_\theta \in\Delta(\mathcal{X}_k)}
        V(x_\theta, y^k; \theta)
        \text{ for all } \theta\in\Theta, \label{eq:rdo-red-meta}\\
    y^k
    &\in \arginf_{y\in\Delta(\mathcal{Y}_k)} \
        \mathscr{R}(x^k(\cdot),y, \mathcal{Y}). \label{eq:rdo-blue-meta}
\end{align}
In particular, the benchmark in
eq.~(\ref{eq:rdo-blue-meta}) is computed over $\Delta(\mathcal{Y})$,
rather than $\Delta(\mathcal{Y}_k)$.
Suppose further that the best-response oracles return exact full-game best responses $\hat{x}^{k}(\cdot)$ and $\hat{y}^{k}$ such that,
\begin{align}
    \hat x^k(\theta)
    &\in \arginf_{x_\theta\in\Delta(\mathcal{X})}
        V(x_\theta, y^k; \theta) \text{ for all } \theta\in\Theta, \\
    \hat y^k
    &\in \arginf_{y\in\Delta(\mathcal{Y})} \
        \mathscr{R}(x^k(\cdot),y, \mathcal{Y}).
\end{align}
Any new pure strategies discovered by the best-response oracles are added to the player populations such that $\mathcal{X}_{k+1}
    =\mathcal{X}_k\cup
      \bigcup\limits_{\theta\in\Theta}\operatorname{supp}(\hat x_\theta^k)$ and $\mathcal{Y}_{k+1}
    =\mathcal{Y}_k\cup\operatorname{supp}(\hat y^k)$. The algorithm terminates when there are no new added pure strategies after at most $
    (|\mathcal{X}|-|\mathcal{X}_0|)+(|\mathcal{Y}|-|\mathcal{Y}_0|)$
population-expanding iterations. Upon termination, the returned
meta-equilibrium is a $\operatorname{PR-MRE}$ fixed point of the full game.
\end{proposition}

\begin{proof}
The populations grow monotonically. Every non-terminal iteration
adds at least one previously absent pure strategy to $\mathcal{X}_k$ or
$\mathcal{Y}_k$. Since $\mathcal{X}$ and $\mathcal{Y}$ are finite, there can be at most
$(|\mathcal{X}|-|\mathcal{X}_0|)+(|\mathcal{Y}|-|\mathcal{Y}_0|)$ such population-expanding iterations.

Let $k$ be the terminal iteration such that the full-game best-response oracles to the meta-equilibrium $(x^{k}(\cdot), y^{k})$ do not add any new pure strategies to the populations $\mathcal{X}^{k}$ and $\mathcal{Y}^{k}$. Denote by $\hat{x}^{k}(\cdot)$ and $\hat{y}^{k}$ the solutions returned by the exact best-response oracles such that,
\begin{align}
    \hat{y}^{k} \in \operatorname{arginf}\limits_{y \in \Delta(\mathcal{Y})} \mathscr{R}(x^{k}(\cdot), y, \mathcal{Y}) \text{ and } \operatorname{supp}(\hat{y}^{k}) \subseteq \mathcal{Y}_{k}.
\end{align}
The support of the exact full-game best-response contained within the current population implies that $\nexists \ y' \in \Delta(\mathcal{Y}) \setminus \Delta(\mathcal{Y}_{k})$ such that $ \mathscr{R}(x^{k}(\cdot), y') < \mathscr{R}(x^{k}(\cdot), \hat{y}^{k})$. Since infimum of the robust regret objective over $\mathcal{Y}$ is atmost equal to the infimum over its subset $\mathcal{Y}_{k}$, we obtain the following equality,
\begin{align}\label{eq:y_eq}
    \operatorname{inf}\limits_{y \in \Delta(\mathcal{Y})} \mathscr{R}(x^{k}(\cdot), y, \mathcal{Y}) = \underbrace{\operatorname{inf}\limits_{y \in \Delta(\mathcal{Y}_{k})} \mathscr{R}(x^{k}(\cdot), y, \mathcal{Y}).}_{\text{obtained at the meta-equilibrium } y^{k}}
\end{align}
Therefore, $\operatorname{inf}\limits_{y \in \Delta(\mathcal{Y})} \mathscr{R}(x^{k}(\cdot), y, \mathcal{Y}) = \mathscr{R}(x^{k}(\cdot), y^{k}, \mathcal{Y})$, which implies that,
\begin{align}\label{eq:full_game_meta_yk}
    y^{k} \in \operatorname{arginf}\limits_{y \in \Delta(\mathcal{Y})} \mathscr{R}(x^{k}(\cdot), y, \mathcal{Y}).
\end{align}
In summary, the support of the full-game best-response $\hat{y}^{k}$ contained within the existing population implies that Blue's meta-equilibrium marginal $y^{k}$ is also a full-game $\operatorname{PR-MRE}$ best-response to Red's meta-equilibrium marginal $x^{k}(\cdot)$ as described by eq.~(\ref{eq:full_game_meta_yk}).
Similarly,
\begin{align}
    \hat{x}^{k}(\theta) \in \operatorname{arginf}\limits_{x \in \Delta(\mathcal{X})} V(x, y^{k}; \theta) \text{ and } \operatorname{supp}(\hat{x}^{k}(\theta)) \subseteq \mathcal{X}_{k} \text{ for each } \theta \in \Theta.
\end{align}
As argued for Blue's side, the support of the per-type full-game best-response $\hat{x}^{k}(\theta)$ contained within the current population implies that $\operatorname{inf}\limits_{x \in \Delta(\mathcal{X})} V(x, y^{k}; \theta) = \operatorname{inf}\limits_{x \in \Delta(\mathcal{X}_{k})} V(x, y^{k}; \theta)$. The $\operatorname{RHS}$ is obtained at the Red meta-equilbrium $x^{k}(\cdot)$ by definition. Therefore $x^{k}(\cdot)$ is a full-game best-response to Blue's meta-equilibrium marginal $y^{k}$ as follows,
\begin{align}\label{eq:full_game_meta_xk}
    x^{k}(\theta) \in \operatorname{arginf}\limits_{x \in \Delta(\mathcal{X})} V(x, y^{k}; \theta) \text{ for all } \theta \in \Theta.
\end{align}
Together, eqs.~(\ref{eq:full_game_meta_yk}) and (\ref{eq:full_game_meta_xk}) satisfy the mutual best-response conditions of the full-game and establish that the meta-equilibrium $(x^k(\cdot),y^k)$ at the terminal $k$-th iteration constitutes a PR-MRE fixed-point of the full game.
\end{proof}
}
\todo{Pseudocode; A subscript notation; define Red BR point to its definition at the beginning of Section 2 }
{\color{black} \textbf{Remark on population-restricted benchmarks}: The finite convergence claim of $\operatorname{PRMRE-RDO}$ is established for a full-game benchmark defining Blue's robust regret objective for meta-equilibrium computation, eq.~(\ref{eq:rdo-blue-meta}). The full-game convergence guarantee does not generally extend without additional assumptions to the setting with a population-restricted benchmark $y^{k} \in \operatorname{arginf}_{y \in \Delta(\mathcal{Y}_{k})} \mathscr{R}(x^{k}(\cdot), y, \mathcal{Y}_{k})$, such that both the search for Blue's decision and the benchmark term (the inner supremum, $\sup_{y^{*} \in \Delta(\mathcal{Y}_{k})} V(x^{k}(\cdot), y^{*}; \theta)$) are defined over the restricted population $\mathcal{Y}_{k}$. 

Specifically, the equality of eq.~(\ref{eq:y_eq}) need not hold due to the mismatch in objectives between the population-restricted best-response for the meta-equilibrium $\operatorname{arginf}\limits_{y \in \Delta(\mathcal{Y}_{k})} \mathscr{R}(x^{k}(\cdot), y, \mathcal{Y}_{k})$ and the full-game best-response $\operatorname{arginf}\limits_{y \in \Delta(\mathcal{Y})} \mathscr{R}(x^{k}(\cdot), y, \mathcal{Y})$. Therefore, termination of $\operatorname{PRMRE-RDO}$ with a population-restricted benchmark need not certify a PRMRE fixed point of the full game.} 


\begin{algorithm}[t]
\caption{PR-MRE Team PSRO}
\label{alg:team-psro-mm}

\noindent\textbf{Result:} Approximate PRMRE Learning for Adversarial Team Games under Asymmetric Information

\noindent\textbf{Input:} $\Theta$, $\bar{\mu}$, $\delta$, Initial population $\Pi_{{r}}^0, \Pi_{{b}}^0$

\begin{tabbing}
\hspace{1.5em}\=\hspace{2.5em}\=\kill

\textbf{repeat} \{for $\mathrm{k}=0,1,\ldots$\} \\

\> $\blacktriangleright$ Compute metagame payoff tensor $A^{k}  \in \mathbb{R}^{\lvert\Pi^{{k}}_{{r}}\rvert \times \lvert\Pi^{{k}}_{{b}}\rvert \times \lvert\Theta\rvert}$ with a third dimension for types \\

\> $\blacktriangleright$ {\color{black} Compute candidate meta-equilibria from the surrogate program on the restricted metagame,} \\
{\color{black} $(x^{{k}}(\cdot), y^{{k}})
\leftarrow$  \ref{eq:prmre_eqm_normal_form}$\{A_{ \Pi^{k}_{r} \times \Pi^{k}_{b}} \}$}
\\
\> \textbf{for} $m$ iterations \textbf{do} \\
\>\> $\blacktriangleright$ Update team best-response
$\pi_{{r}}$ toward
$\mathrm{BR}_{\mathrm{Val}}(y^{k}, \operatorname{Unif}(\Theta))$ via cooperative MARL \\
\>\> {\color{black}$\blacktriangleright$ Obtain final mixture $y^{*}_{k, \mathrm{auxiliary}}$ over robust team best-responses trained via an auxiliary} \\ {\color{black}double-oracle procedure for a zero-sum regret game between Blue and Nature with population-restricted} \\ {\color{black} benchmarks, $\operatorname{BR}_{\mathrm{PRMRE}}(x^{{k}}(\cdot), \bar{\mu}, \delta)$ (see page 21)} \\

\> $\Pi_{{r}}^{{k}+1}
\leftarrow \Pi_{{r}}^{k} \cup \{\pi_{{r}}\}$ \\

\> {\color{black} $\Pi_{{b}}^{{k}+1}
\leftarrow \Pi_{b}^{k} \cup \operatorname{supp}(y^{*})$} \\

\textbf{until} max $k$ number of iterations
 \\

\textbf{Return:} $(x^{{k}}(\cdot), y^{{k}})$
\end{tabbing}
\end{algorithm}

\section{PRMRE-PSRO: Learning Approximate {\color{black} Probabilistically-Robust} Minimax-Regret Equilibria for Adversarial Team Games under Asymmetric Information}

In Section~\ref{sec:PRMRE}, we proposed a novel equilibrium concept that combines the distribution-free robustness of minimax-regret reasoning with coarse probabilistic information about types from a nominal prior. The $\operatorname{PR-MRE}$ best-response operator for the uninformed Blue (see Definitions~\ref{def:mmr} and \ref{def:exp}) computes a team best-response to the omniscient Red by minimizing worst-case regret of the Blue team strategy across subsets of the type-space weighted by coarse probabilistic information about their occurrence from the prior thus providing a robustness guarantee to strategic distribution shifts in terms of {\color{black} a point-wise benchmark-relative performance certificate}. For a finite normal-form game with asymmetric information, the $\operatorname{PR-MRE}$ best-response couples regrets across types resulting in a robust bilinear program for the $\operatorname{PR-MRE}$ equilibrium as compared to a bilinear program for the distribution-free $\operatorname{MRE}$ equilibrium and the linear-programming formulation for Bayesian Nash equilibria (see Section~\ref{sec:eqm_concepts_section} for the various equilibrium concepts and their robustness properties).

$\operatorname{PR-MRE}$ (and $\operatorname{MRE}$) does not admit a known decomposition to local CFR-style regret updates for an extensive form representation of a game with hidden types. Standard variants in the CFR literature recover approximately the Bayesian Nash equilibrium of the asymmetric information game and handle hidden information as known upto a prior (or chance). To compute robust team behavior in large adversarial team games with hidden information and potential strategic distribution shifts, we propose $\operatorname{PRMRE-PSRO}$ based on the robust double oracle procedure $\operatorname{RDO}$ for learning approximate $\operatorname{PR-MRE}$ for ATGs via reinforcement learning based best-responses. 

Algorithm~\ref{alg:team-psro-mm} details the pseudocode for $\operatorname{PR-MRE}$ Team $\operatorname{PSRO}$. For the Red team (omniscient) and Blue team (uninformed) policy sets $\Pi_{{r}}^{{k}}$ and $\Pi_{{b}}^{{k}}$ at the $k$th iteration of the $\operatorname{PRMRE-PSRO}$ procedure, metagame payoff tensor {\color{black}$A^{k} \in \mathbb{R}^{\lvert \Pi^{{k}}_{{r}}\rvert \times \lvert \Pi^{{k}}_{{b}}\rvert \times \lvert\Theta\rvert}$} is maintained with an explicit third dimension corresponding to the type-space {\color{black} such that $ A^{k}_{i, j, l} = V(\pi_{r, i}, \pi_{b, j}; \theta_{l})$}. Maintaining a three-dimensional payoff tensor for the robust double oracle approach enables reasoning about the relative performance of policy sets across types as opposed to a two-dimensional empirical game matrix common in the PSRO literature \cite{psro_silver, slumbers_rae_psro, mcaleer2023teampsro}. For a game of Capture-the-Flag for instance, the three-dimensional payoff tensor measures performance of the Red defense and Blue attack-and-capture team policies across all possible flag hypotheses. The 3D payoff representation enables reasoning about regret across types for meta-equilibrium computation and robust policy set expansion.
We now detail the robust population-based approach $\operatorname{PRMRE-PSRO}$ for learning approximate minimax-regret equilibria for adversarial team games under asymmetric information as follows, 

\textbf{PR-MRE PSRO}, \textbf{Algorithm~\ref{alg:team-psro-mm}}. {\color{black} Policy set expansion proceeds  by best-responding to the candidate opponent team meta-equilibrium as computed via the surrogate robust bilinear program \ref{eq:prmre_eqm_normal_form} (described in Section~\ref{sec:computation}) on the empirical metagame payoff tensor}. 

The omniscient Red team uses a value-based team best-response $\operatorname{BR}_{\mathrm{Val}}$ criterion that maximizes expected Red team payoff against an opponent distribution given by the Blue $\operatorname{PR-MRE}$ {\color{black} candidate} meta-equilibrium $y^{k} \in \Delta(\Pi^{{k}}_{{b}})$ at the latest $\operatorname{PSRO}$ iteration. Value-based team best-responses can be computed via standard cooperative multi-agent reinforcement learning methods such as multi-agent proximal policy optimization (MAPPO) \cite{mappo}.

{\color{black}
As defined in Definition~\ref{def:prmre}, the Blue team $\operatorname{PR-MRE}$ robust best-response to the type-conditioned Red meta-equilibrium $x^{k}(\cdot) \in \Delta(\Pi^{k}_{r})$ is given by,
\begin{align}\label{eq:rl_prmre_br}
\operatorname{BR}^{\mathrm{(b)}}_{\mathrm{PRMRE}}(x^{k}(\cdot); \bar{\mu}, \delta) = \arginf_{\pi \in \Pi_{\mathrm{b}}} \mathscr{R}(x^{k}(\cdot), \pi, \Pi_{b}) = \arginf_{\pi \in \Pi_{\mathrm{b}}} \sup_{\mu \in \mathcal{S}(\bar{\mu}, \delta)} \sum_{\theta \in \Theta} \mu(\theta) \mathcal{E} \left(x^{k}(\cdot), \pi; \theta \right),
\end{align}
where $\mathcal{E} \left(x^{k}(\cdot), \pi; \theta \right) = \sup_{\pi^{*} \in \Pi_{b}} V(x^{k}(\cdot), \pi^{*}; \theta) - V(x^{k}(\cdot), \pi; \theta)$ is defined with respect to a full-game benchmark. Therefore for a fixed Red, the $\operatorname{PR-MRE}$ robust best-response as given by eq.~(\ref{eq:rl_prmre_br}) induces an auxiliary zero-sum regret-game between Blue and Nature. We propose the approximate computation of the above regret-game saddle-point via reinforcement-learning based best-response training through an auxiliary double oracle between Blue and Nature. Specifically, Blue is trained against an adaptively reweighted regret-maximizing Nature in the auxiliary double-oracle run and against an opponent distribution fixed at the type-conditioned Red obtained as the meta-equilibrium of the main double-oracle run in the outer-loop. 

Moreover, to aleviate the computational burden of training a per-type best-response to compute the full-game benchmark, we use a restricted-population benchmark for the auxiliary regret game between Blue and Nature such that $\mathcal{E} \left(x^{k}(\cdot), \pi; \theta \right) \approx \mathcal{E} \left(x^{k}(\cdot), \pi; \theta, \Pi^{k}_{b} \right) = \sup_{y^{*} \in \Delta(\Pi^{k}_{b})} V(x^{k}(\cdot), y^{*}; \theta) - V(x^{k}(\cdot), \pi; \theta)$ where $ \sup_{y^{*} \in \Delta(\Pi^{k}_{b})} V(x^{k}(\cdot), y^{*}; \theta) = \sup_{y^{*} \in \Delta(\Pi^{k}_{b})} x^{k}(\theta) A^{k} y^{*}$ in terms of the metagame payoff tensor $A^{k}$ and $V(x^{k}(\cdot), \pi; \theta) = \sum_{i \in \lvert \Pi^{k}_{r} \rvert} x^{k}_{i}(\theta) V(\pi_{r, i}, \pi; \theta) $. We first make the following proposition before the exposition of our proposed auxiliary regret-based double-oracle approach to compute approximate $\operatorname{PR-MRE}$.
\begin{proposition} For a fixed Nature's distribution $q$ under the Red opponent $x^{k}(\cdot)$,
\begin{align}
    \arginf\limits_{\pi_{b} \in \Pi_{b}} \sum\limits_{\theta \in \Theta} q(\theta) \mathcal{E} \left(x^{k}(\cdot), \pi_{b}; \theta \right) = \argsup\limits_{\pi_{b} \in \Pi_{b}} \sum\limits_{\theta \in \Theta} q(\theta) V \left(x^{k}(\cdot), \pi_{b}; \theta \right).
\end{align}
\end{proposition}
In other words, under a fixed Nature, the regret-minimizing Blue robust best-response is equivalent to a value-maximizing best-response $\mathrm{BR}_{\mathrm{Val}}$. This is beneficial as the robust best-response training reduces to cooperative MARL training against the regret-maximizing Nature's type distribution computed from the restricted regret game under the fixed Red meta-equilibria $x^{k}(\cdot) \in \Delta(\lvert \Pi^{k}_{r} \rvert)$.}

\begin{tcolorbox}[
title={\color{black} $\operatorname{BR}_{\mathrm{PRMRE}}$, Auxiliary Double Oracle Loop for the Blue-Nature Zero-Sum Regret-Game},
colback=white,
colframe=black!40,
colbacktitle=white,
coltitle=black]
{\color{black} Denote auxiliary-loop iterations by $t$ and main-loop iterations by $k$. The auxiliary run maintains it's own column set $\Pi^{k}_{b, t}$ warmstarted by the main-loop column set such that $\Pi^{k}_{b, 0} = \Pi^{k}_{b}$ and a metagame payoff matrix $A^{k}_{t} \in \mathbb{R}^{\lvert\Theta\rvert \times \lvert \Pi^{k}_{b,t}\rvert}$ such that $ A^{k}_{t} (i, j) = V(x^{k}(\cdot), \pi_{b, j}, \theta_{i})$. At auxiliary iteration $t$,}
\newline
\vspace{2 em}
\begin{minipage}[t]{0.45\linewidth}
{\color{black}
Compute Nature's regret-based adversarial type distribution $q^{*}_{k,t}$ as the Nash solution of the finite zero-sum regret game as follows,
\begin{align}\label{eq:rrg}
q^{*}_{k,t}(\cdot) \in  \argsup_{\mu \in \mathcal{S}(\bar{\mu}, \delta)} \inf_{y \in \Delta(\Pi^{k}_{{b}, t})} \sum_{\theta \in \Theta} \mu(\theta) \mathcal{E} \left(x^{k}(\cdot), y; \theta \right),
\end{align}
where $\mathcal{E}\left(x^{k}(\cdot), y; \theta \right)$ is shorthand for regret $\operatorname{w.r.t.}$ the auxiliary population-restricted benchmark $\mathcal{E}\left(x^{k}(\cdot), y; \theta, \Pi^{k}_{b, t} \right) = \sup_{y^{*} \in \Delta(\Pi^{k}_{b, t})} A^{k}_{t}(\theta, \cdot)y^{*} - A^{k}_{t}(\theta, \cdot) y$.

Since the Nash solution need not be unique, we solve the equilibrium selection problem by first computing an optimal minimax mixture over the auxiliary Blue population $y^{*}_{k, t}$ and the optimal robust regret objective $\varepsilon^{*}$ for the zero-sum regret game (\ref{eq:rrg}) followed by a maximum-entropy selection criteria for the adversarial type distribution: $\sup_{\mu \in \mathcal{S}(\bar{\mu}, \delta)} \sum_{k} -\mu_{k}\log(\mu_{k})$ with the linear constraint $\mu^{\top}\mathcal{E}(x^{k}(\cdot), y^{*}_{k, t}, \cdot) = \varepsilon^{*}$. Uniqueness of the adversarial type distribution to train the Blue team follows from strict concavity of the entropy objective.
}
\end{minipage}
{\color{black} \hfill\raisebox{-10em}{\Large$\Longrightarrow$}\hfill}
\begin{minipage}[t]{0.45\linewidth}
{\color{black}
Warmstart $\pi^{k}_{b, t+1}$ from $\pi^{k}_{b, t}$ where $\pi^{k}_{b, 0}$ is selected from the main-loop Blue population according to some prescribed initialization rule and resume training against the regret-adapted Nature $q^{*}_{k, t}$ for a pre-defined budget before the next re-weighing trigger,
\begin{align}
    \pi^{k}_{b, t+1} \longleftarrow \mathrm{BR}_{\mathrm{Val}}(x^{k}(\cdot), q^{*}_{k, t}).
\end{align}
}
{\color{black} Add the latest Blue checkpoint as a new column to the auxiliary pool $\Pi^{k}_{b, t+1} = \Pi^{k}_{b, t} \ \cup \ \{\pi^{k}_{b, t+1}\}$. Compute per-type values of the newly added Blue column against the fixed Red opponent from the outer main loop $ V(x^{k}(\cdot), \pi^{k}_{b, t+1}; \theta)$ proceeded by updating the regret values for the expanded auxiliary game. Adding a new column can change the benchmark values due to potentially discovering new per-type suprema leading to an update (increase) in the per-type regret of the older columns.}
\end{minipage}
{\color{black} Recompute Nature's regret-based adversarial type distribution over the expanded auxiliary column set and continue training. The auxiliary expansion occurs for a fixed total training budget such that Nature is re-weighed at a pre-defined frequency. At the end of the total budget, only the final $\operatorname{PR-MRE}$ mixture's support over the auxiliary columns is promoted to the main pool and the remaining columns are discarded.

\textbf{Return}: $ y^{*}_{k, \mathrm{auxiliary}} \in \Delta( \Pi^{k}_{b, t_{\mathrm{final}}})$.
}
\end{tcolorbox}
{\color{black}For our $\lvert\Theta\rvert = 3$ experiment (Capture-the-Flag with three possible flags), a 3M budget was used whereby for the first 600K steps, Nature's type distribution was recomputed every 100K steps (6 times) and from 600K-2200K every 200K steps (8 times) and frozen for the run between 2200K-3000K steps. Therefore, each auxiliary regret-game expansion added $6 + 8 = 14$ Blue checkpoints as columns out of which only the support of the final PR-MRE mixture over the 21 columns (14 + 7 from the seed) was promoted to the main column set.}

{\color{black} \textbf{Note on MRE, DRO-MaxMin and BNE PSRO}. The MRE and DRO-MaxMin PSRO procedures follow the same population-based expansion framework as PRMRE-PSRO, with meta-equilibria computed according to their respective solution concepts. For the uninformed Blue team, the corresponding robust best-response operators are approximated through auxiliary Blue–Nature double-oracle procedures. While PR-MRE and MRE employ a regret-maximizing Nature, with admissible type distributions restricted to $S(\bar\mu,\delta)$ and $\Delta(\Theta)$, respectively, DRO-MaxMin employs a payoff-minimizing Nature over $S(\bar\mu,\delta)$. For a fixed Nature distribution, Blue's best-response training reduces to value maximization in all three procedures. In contrast, BNE-PSRO computes Bayesian meta-equilibria and trains both teams against their respective opponent meta-equilibria under the nominal type distribution $\bar\mu$.


We now evaluate the different robust best-response operators for the uninformed Blue team on an example Capture-the-Flag game and compare the results.} For our experiment, all Red and Blue team policies are parameterized as Graph Neural Networks (GNNs) that directly take in as input the graph-structured game state of the ATG. We note that the proposed $\operatorname{PRMRE-PSRO}$ is invariant to the choice of neural network architecture and provides a general methodology to train robust team policies to distribution shifts under hidden information.
\subsection{Graph-structured Capture-the-Flag: Experiment and Results}
{\color{black}
\begin{experiment}
    Compare Blue team policies learned via $\operatorname{MRE}$, $\operatorname{PR-MRE}$, $\operatorname{DRO-MaxMin}$ robust best-response operators and a value-based best-response against different Red opponents on an example graph-structured adversarial team game under asymmetric information, Capture-the-Flag (Figure~\ref{wrap-fig:atg}), in terms of robustness to distribution shift of the hidden flag hypotheses.
\end{experiment}
}

{\color{black} For our experiment, we compare one half of the expansion iteration of the robust $\operatorname{PSRO}$ procedure for the uninformed Blue team to compare robustness of policies learned via the different robust best-response operators starting from the same seed metagame on the example Graph Capture-the-Flag instance highlighted in Figure~\ref{wrap-fig:atg} in terms of worst-case regret within the admissible set of test-time flag distributions. The seed metagame for the robust $\operatorname{PSRO}$ expansion is described as follows.}

{\color{black} \textbf{Seed Metagame} 

The seed metagame (heatmap visualized in Fig.~\ref{fig:seed_metagame_3flag}) is comprised of seven scripted experts for both the teams: $\{ \operatorname{defend\_ L}, \operatorname{defend\_ M}, \operatorname{defend\_ R}, \operatorname{defend\_ LM}, \operatorname{defend\_ LR}, \operatorname{defend\_ MR}, \operatorname{rush}\}$ and $\{ \operatorname{commit\_ L}, \operatorname{commit\_ M}, \operatorname{commit\_ R}, \operatorname{split\_ LM}, \operatorname{split\_ LR}, \operatorname{split\_ MR}, \operatorname{sweep}\}$ for Red and Blue respectively. The heatmap shows the relative strengths of Red-Blue team policy matchups across the three flag types: $\{\theta_{0}, \theta_{1}, \theta_{2}\} = \{\text{\textbf{L}eft}, \text{\textbf{R}ight}, \text{\textbf{M}iddle}\}$. A flag capture is +1 for Blue and -1 for Red and a timeout without any capture is -1 for Blue and +1 for Red such that a Blue score of +0.47 under a policy matchup means Blue captures the flag 73.5$\%$ and misses 26.5$\%$ of the times such that $ 0.735\times(+1) + 0.265\times(-1) = +0.47$. Blue team operates under knowledege limited to the nominal flag distribution $\bar{\mu} = (0.10, 0.35, 0.55)$. We consider test-time type distributions from the ambiguity set $\mathcal{S}(\bar{\mu}, \delta)$ for $\delta = 0.15$ such that the Left flag remains rare under all admissible test-time type distributions up to the prescribed confidence level $\delta$.

\begin{figure}
  \centering
  \includegraphics[width=\textwidth]{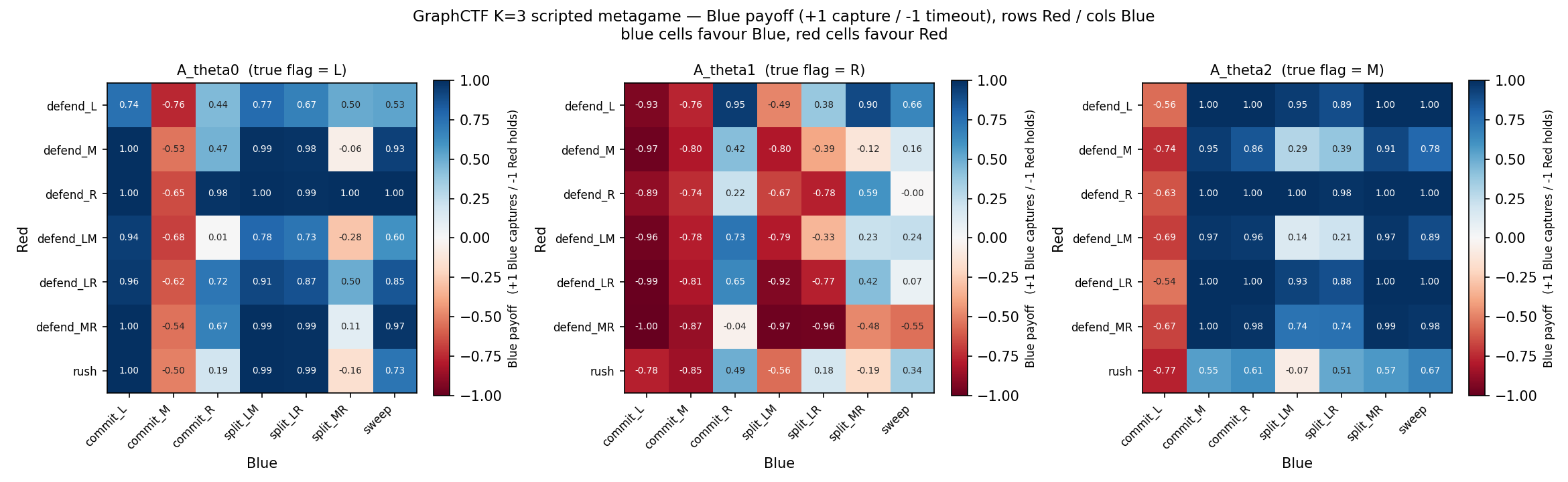}
  \caption{{\color{black} $7\times7\times3$ seed metagame payoff tensor for a three flag Capture-the-Flag game on the graph-structured map given in Figure~\ref{wrap-fig:atg}. $(\theta_{0}, \theta_{1}, \theta_{2}) = (\text{\textbf{L}eft}, \text{\textbf{R}ight}, \text{\textbf{M}iddle})$ and the three slices of the payoff tensor as depicted left to right correspond to $\theta_{0}$--$\theta_{2}$. For any slice, rows depict Red team policies, columns depict Blue team policies and the values in the cells denote Blue team's expected payoff under any matchup.}}
  \label{fig:seed_metagame_3flag}
\end{figure}

We compare the robustness properties of learned best-responses via a Bayesian best-response, MRE-based, PRMRE-based and DRO-MaxMin based robust best-response operators for multiple Red opponents. The primary metric for comparison is the worst-case regret of the learned policy (or policy mixture) within the admissible set of test-time type distributions against different Red opponents. We also compare worst-case payoff and average-case payoff of the learned policies (policy mixtures).
}

{\color{black} \textbf{Red Opponents}

We consider three different Red opponents: BNE, MRE and PR-MRE ($\delta = 0.15$, computed from the surrogate program \ref{eq:prmre_eqm_normal_form}) of the seed metagame for comparing learned robust Blue policies (policy mixtures) corresponding to the different robust best-response operators. For the seed metagame as described in the heatmaps of Figure~\ref{fig:seed_metagame_3flag}, Red BNE $\mathbf{r_0}$ is $\{ \theta_{L}:(\text{defend\_L \textbf{1.00}}), \theta_{R}:(\text{defend\_MR \textbf{1.00}}), \theta_{M}:(\text{rush \textbf{1.00}} \}$), Red MRE $\mathbf{r_{1}}$ is $\{ \theta_{L}:(\text{defend\_LM \textbf{1.00}}), \theta_{R}:(\text{defend\_MR \textbf{1.00}}), \theta_{M}:(\text{rush \textbf{1.00}}) \}$ and Red PR-MRE for $\delta=0.15$, $\mathbf{r_{2}}$ is $\{ \theta_{L}:(\text{defend\_LM \textbf{0.561}, defend\_L \textbf{0.328}, defend\_LR \textbf{0.111}}), \theta_{R}:(\text{defend\_MR \textbf{1.00}}), \theta_{M}:(\text{rush \textbf{1.00}}) \}$. Red DRO-MaxMin on the seed metagame is identical to the Red MRE and hence not included as a separate row in our block experiment.
} 

\begin{figure}
  \centering
  \includegraphics[scale=0.75]{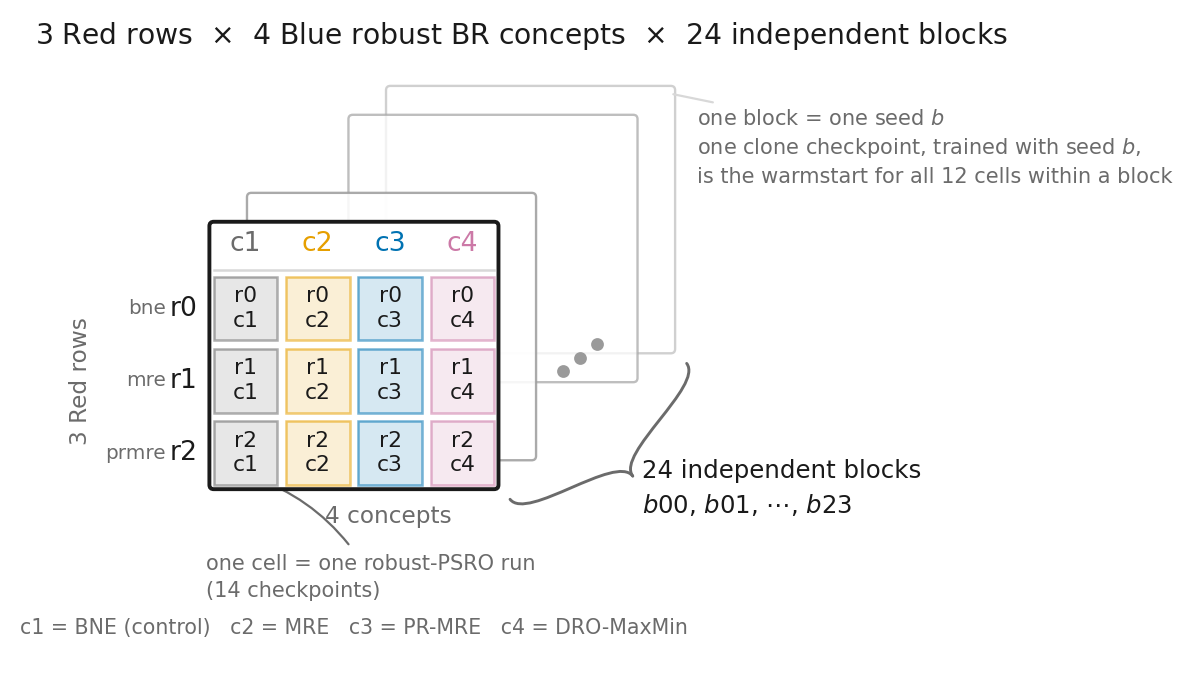}
\caption{{\color{black} The block design. One card is one block: the whole $3\times4$ grid run once, one cell being one robust-BR training run (3M steps, 14 checkpoints). c1--c4 are the BR concepts (BNE control, MRE, PR-MRE, DRO-MaxMin); r0--r2 are the three Red meta-equilibria the concepts train against (MRE and DRO-MaxMin Red meta-equilibria over the seed metagame Figure~\ref{fig:seed_metagame_3flag} is the same, hence omitted). All 12 cells of a block share one warmstart network and one training seed -- the block index $b$ is the seed, and every cell loads the same clone checkpoint -- so they begin bit-identical and differ only in (Red row, concept). With Red fixed, the expansion is a Nature-vs-Blue game solved alongside training, and the four cells of a row differ only in Nature: PR-MRE faces a regret-maximising Nature restricted to $\mathcal{S}(\bar{\mu}, \delta)$, DRO-MaxMin a payoff-minimising Nature on the same $\mathcal{S}(\bar{\mu}, \delta)$, MRE a regret-maximising Nature on the full simplex, and the BNE control none at all ($q^{*}=p_{0}$ throughout). Any difference in the learned policy within a block is therefore attributable to the concept. Blocks are independent draws of clone and RL seed; nothing is paired across them.}}
\label{fig:design-blocks}
\end{figure}

{\color{black} \textbf{Robust Expansion Experiment}

\textbf{Block-Design and Best-Response Training Setup}. Figure~\ref{fig:design-blocks} describes our block structured experiment for $n_{\textrm{blocks}} = 24$. One block is $3\times4$ i.e. $3 \text{ Reds} \times 4 \text{ Blue robust BRs}$ and corresponds to one shared seed for all the $12$ cells. The shared seed $b$ for a block uniquely defines a randomly initialized neural-network (GNN + GRU) that is trained via behavior cloning ($\operatorname{DAgger}$) of the $\operatorname{sweep}$ Blue team policy (seventh column of the seed metagame, Figure~\ref{fig:seed_metagame_3flag}). Given block $b$ and any Red opponent $r_{0}$--$r_{2}$, the robust Blue best-response expansion runs corresponding to concepts BNE $c_{0}$, MRE $c_{1}$, PR-MRE $c_{2}$ and DRO-MaxMin $c_{3}$ all warmstart from the same trained clone for the block. Besides the shared clone, all other $\operatorname{RNG}$s within a block such as for RL-based best-response training $c_{0}$--$c_{3}$ over the warmstart clone share the same seed. This paired block design controls for variation in initialization and training seeds, enabling within-block comparisons of the robust best-response operators against each fixed Red opponent. All best-responses are trained via MAPPO with the same hyperparameters across Red opponents, Blue concepts and blocks.

A training budget of $3000$K total timesteps with a learning rate of $3$e$-4$ is used with a linear decay starting at 60\% of the budget for all best-responses $c_{0}$--$c_{3}$. For the MRE, PR-MRE and DRO-MaxMin robust best-responses, Nature is re-weighed at a frequency of every 100K steps for the first 600K steps (6 checkpoints), every 200K steps for 600-2200K steps (8 checkpoints) and frozen for 2200-3000K steps with a total of 6 + 8 = 14 checkpoints added to the auxiliary Blue column pool. Thus the auxiliary pool during robust best-response training contains 7+14 = 21 policies and metrics are compared for the final returned mixtures over the 21 policies. The BNE response is trained against the nominal type distribution for the full budget with no reweighing. Note that adaptive re-weighing of the training type distribution and all other hyperparameter schedules such as learning rate decay are independent processes and are measured against progress with respect to the total training budget of 3000K timesteps. Also note that all evaluation, reported results and plots correspond to checkpoints trained till 2200K steps (same for all four Blue concepts: BNE, MRE, PR-MRE and DRO-MaxMin) and the final trained policies at the end of 3000K steps is not evaluated.
\begin{figure}
  \centering
  \includegraphics[scale=0.725]{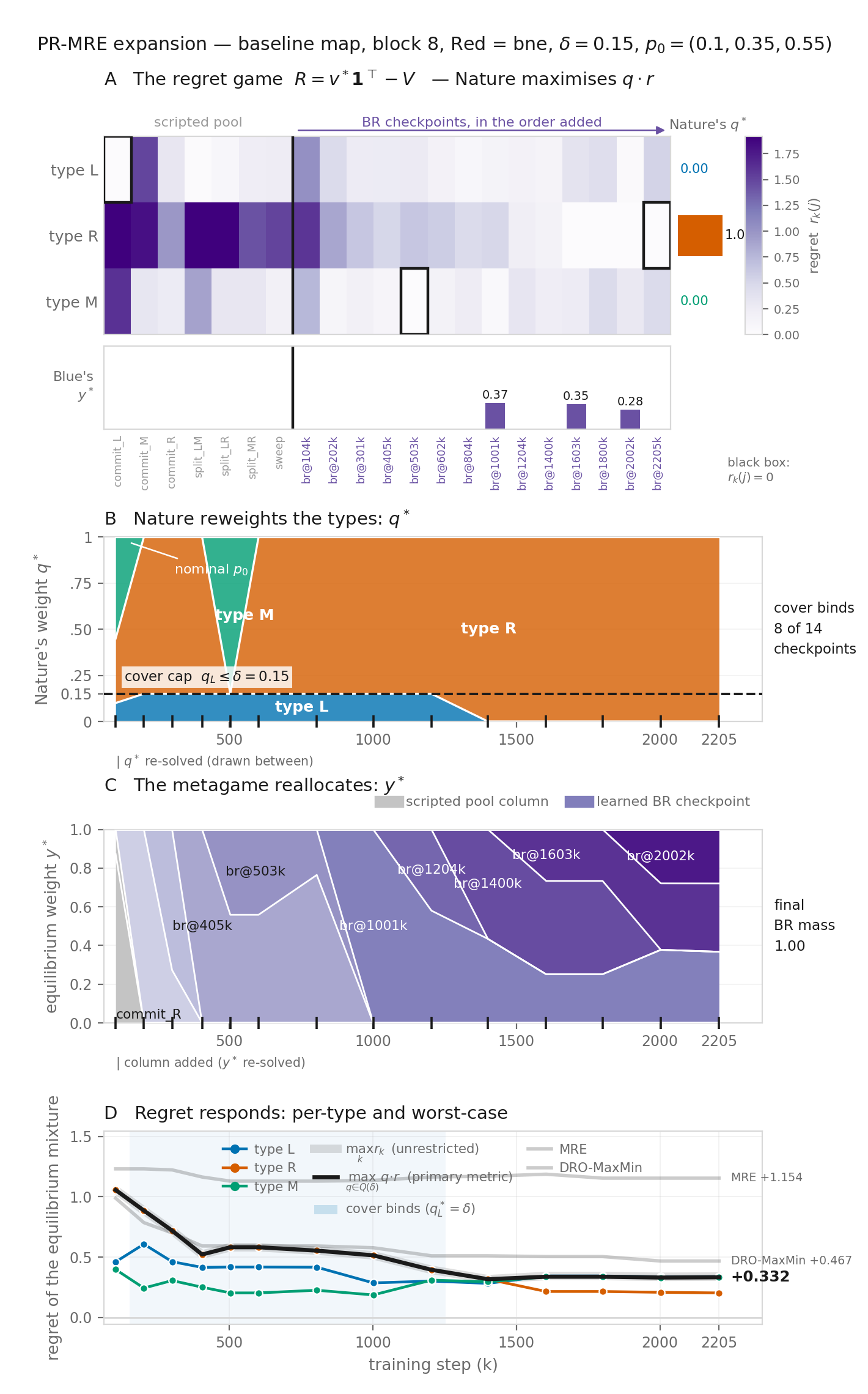}
  \caption{\color{black}{\textbf{Training mechanism, PR-MRE robust best-response}. Block 08, Red = BNE. $\bar{\mu} = (0.1, 0.35, 0.55)$, $\delta=0.15$. \textbf{Panel A}: evolution of the auxiliary regret-based game (white-purple heatmap) as trained checkpoints are added. \textbf{Panel B}: Nature's regret-maximizing adversarial type distribution over $\mathcal{S}(\bar{\mu}, \delta)$. \textbf{Panel C}: evolution of the $\operatorname{PR-MRE}$ robust population mixture. \textbf{Panel D}: per-type and worst-case regret evolution during training w.r.t. per-type benchmarks pooled across arms per block.}}
  \label{fig:prmre_mechanism}
\end{figure}

\begin{figure}
  \centering
  \includegraphics[scale=0.725]{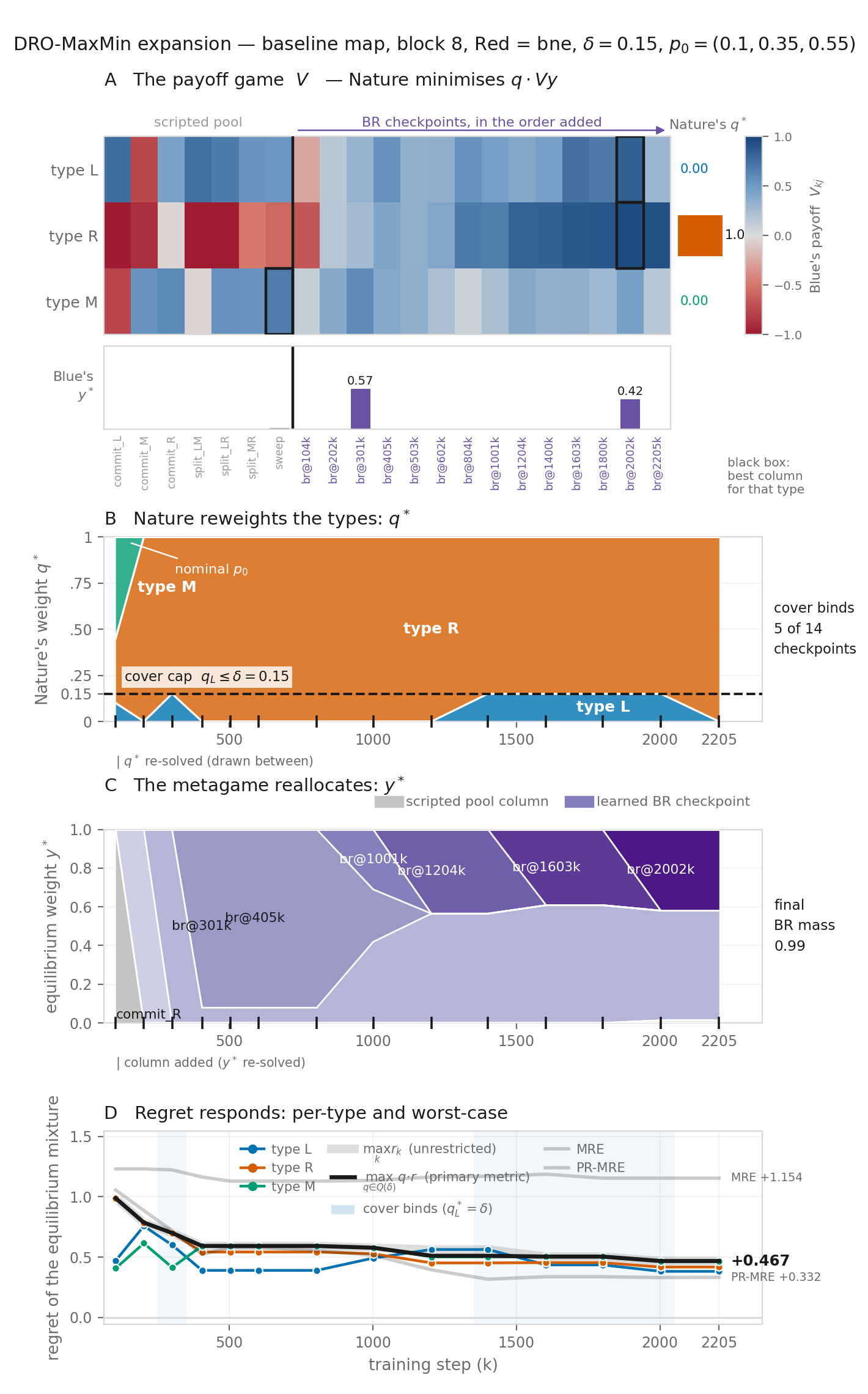}
  \caption{{\color{black} \textbf{Training mechanism, DRO-MaxMin robust best-response}. Block 08, Red = BNE. $\bar{\mu} = (0.1, 0.35, 0.55)$, $\delta=0.15$. \textbf{Panel A}: evolution of the auxiliary payoff-based game (red-blue heatmap) as trained checkpoints are added. \textbf{Panel B}: Nature's payoff-minimizing adversarial type distribution over $\mathcal{S}(\bar{\mu}, \delta)$. \textbf{Panel C}: evolution of the $\operatorname{DRO-MaxMin}$ robust population mixture. \textbf{Panel D}: per-type and worst-case regret evolution during training w.r.t. per-type benchmarks pooled across arms per block.}}
  \label{fig:dromaxmin_mechanism}
\end{figure}

\textbf{Illustration of the Mechanism (block 08, Red $r_{0}$)}. Figures~\ref{fig:prmre_mechanism}, \ref{fig:dromaxmin_mechanism} and \ref{fig:mro_mechanism} highlight the training mechanism ($\operatorname{block} 08$, Red BNE $r_{0}$) -- the auxiliary double-oracle runs for PR-MRE, DRO-MaxMin and MRE robust best-response operators respectively warmstarted from a shared clone of the $\operatorname{sweep}$ Blue team policy and the same seed. Panel A in all the figures highlights the regret game (white-purple heatmap) and the payoff (red-blue heatmap) game between Blue and Nature for PRMRE, MRE and DRO-MaxMin respectively. The first seven columns in panel A correspond to Blue's seed population and the remaining 14 columns are added during best-response training. From figure~\ref{fig:prmre_mechanism} for the $\operatorname{PR-MRE}$ robust best-response, columns are added, population-restricted benchmarks for the auxiliary regret matrix updated and Nature re-weighed at the prescribed frequency. The $\operatorname{PR-MRE}$ robust best-response trains against a regret-maximizing adversarial type distribution $q^{*}_{\operatorname{PRMRE}}$ as opposed to a payoff-minimizing type distribution $q^{*}_{\operatorname{DRO}}$ for the $\operatorname{DRO-MaxMin}$ robust best-response. Moreover, PR-MRE and DRO-MaxMin robust best-responses train against typicality-preserving adversarial type distributions from the admissible set $ \mathcal{S}(\bar{\mu}, 0.15) $ such that the nominal-rare \textbf{L}eft flag remains rare under admissible test-time distributions whereas MRE trains against a regret-maximizing Nature from the type simplex and does not incorporate any probabilistic information from the nominal.

For the run between $0$-$100$K timesteps, the three best-responses train on the nominal flag distribution against the fixed Red opponent and the training run starts diverging from step $100$K onwards -- from Panel B in Figures~\ref{fig:prmre_mechanism}, \ref{fig:dromaxmin_mechanism} and \ref{fig:mro_mechanism}, $q^{*}_{\operatorname{PRMRE}, \operatorname{100K}} = (0.15, 0.85, 0)$, $q^{*}_{\operatorname{DRO}, \operatorname{100K}} = (0, 1, 0)$ and $q^{*}_{\operatorname{MRE}, \operatorname{100K}} = (1, 0, 0)$. Panel C in the figures shows the diverging populations and population mixtures for the three auxiliary double-oracle runs. The final returned Blue mixture at the end of the auxiliary double-oracle run for the three Blue arms from panel A: PR-MRE $\{(\operatorname{br@1000K}, 0.37), (\operatorname{br@1600K}, 0.35), (\operatorname{br@2000K}, 0.28)\}$, DRO-MaxMin $\{(\operatorname{br@301K}, 0.57), (\operatorname{br@2002K}, 0.43)\}$ and MRE $\{(\operatorname{commit\_R}, 0.11), (\operatorname{sweep}, 0.13), (\operatorname{br@1800K}, 0.75)\}$. Note that checkpoints are labeled by timestep for eg. $\operatorname{br@1600K}$ and represent different actual policies since the training trajectories diverge between concepts (figs.~\ref{fig:prmre_mechanism}, \ref{fig:dromaxmin_mechanism} and \ref{fig:mro_mechanism} detail evolution of the auxiliary Blue-Nature games, Nature's adversarial type distribution and the robust mixtures).

Panel D in the figures denotes the per-type and worst-case (within $\mathcal{S}(\bar{\mu}, \delta)$, $\delta=0.15$) regret evolution for the respective Blue population mixtures as training progresses such that per-type regret is measured with respect to pooled benchmarks from the completed runs of the four Blue arms $c_{0}$--$c_{3}$ for $\operatorname{block} 08$ and Red $r_{0}$ (pooled per-type from 7 seed + 14 trained checkpoints per Blue arm). From panel D, MRE Blue that trains against the rare flag $\mathbf{L}$ for majority of the run (13 of 14 checkpoints) incurs high regret against flags $\mathbf{R}$ and $\mathbf{M}$ leading to the poorest worst-case regret amongst admissible flag distributions (+1.154) as compared to PR-MRE (+0.332) and DRO-MaxMin (+0.467) at the end of the $\operatorname{block} 08$ run against Red $r_{0}$.

\textbf{Paired-differences and Metrics}. We are concerned with three paired differences: $\operatorname{PRMRE}-\operatorname{BNE}$, $\operatorname{PRMRE}-\operatorname{MRE}$ and $\operatorname{PRMRE}-\operatorname{DRO-MaxMin}$ labeled H1--H3 respectively with the primary metric: in-band worst-case regret where in-band refers to the trapezoidal-shaped band in the three-type simplex, $\mu(\theta_{0}) \leq \delta$. In-band worst-case and average payoffs constitute secondary metrics in the analysis of our block experiment, $n_{\mathrm{blocks}} = 24$. For each metric, paired differences H1--H3 are computed for all blocks. As described before, for any block $b$, the in-band worst-case regret of the final mixture corresponding to any Blue concept $c$ against row $r$ (Red opponent) uses pooled per-type benchmarks across the Blue concepts for that block $b$ such that paired differences of worst-case regret H1--H3 are computed per-block and per-Red. 

Figure~\ref{fig:wcreg_per_row} shows paired-differences for in-band worst-case regret per-block: H1--H3 (final trained mixtures) under different Red opponents. From Figure~\ref{fig:wcreg_per_row} and for block b$08$ and Red $r_{0}$, H1 ($\operatorname{PRMRE} - \operatorname{BNE}$) in-band worst-case regret = -0.45, H2 ($\operatorname{PRMRE} - \operatorname{MRE}$) = -0.82 and H3 ($\operatorname{PRMRE} - \operatorname{DRO-MaxMin}$) = -0.13. Recall from panel D in the mechanism figs.~\ref{fig:prmre_mechanism}, \ref{fig:dromaxmin_mechanism} and \ref{fig:mro_mechanism} that depicts regret evolution during training for block 08 and Red $r_{0}$, the final wc-regret for PR-MRE, MRE and DRO-MaxMin = +0.332, +1.154 and +0.467. This is consistent with Fig.~\ref{fig:wcreg_per_row} up to rounding errors such that H2 $= +0.332 - (+1.154) = -0.822$ and H3 $ = +0.332 - (+0.467) = -0.135$. Figures~\ref{fig:wcpay_per_row} and \ref{fig:avgpay_per_row} summarize the per-block paired differences for our secondary metrics: in-band worst-case payoff and average-case payoff respectively.

\textbf{Aggregate Metrics and Results}. Table~\ref{tab:contrasts-per-row-n24} reports the aggregate results of our block experiment, H1--H3 across the Red opponents $r_{0}$--$r_{2}$ for $n_{\mathrm{blocks}} = 24$ -- mean, standard deviation and the 95\% confidence-interval for the primary and the two secondary metrics. Moreover, we report the per-block values for all metrics in Figures~\ref{fig:wcreg_per_row}, \ref{fig:wcpay_per_row} and \ref{fig:avgpay_per_row} for reference.

H1, $\operatorname{PRMRE}-\operatorname{BNE}$: $\operatorname{PRMRE}$ achieves lower mean in-band worst-case regret than $\operatorname{BNE}$ against all three Red opponents, with mean differences of $-0.183$, $-0.146$, $-0.155$ respectively and confidence intervals excluding zero. PR-MRE also achieves higher worst-case payoff against all three opponents (mean paired differences +0.214, +0.164 and +0.196 respectively with confidence intervals excluding zero), while average-case payoff differences remain inconclusive since confidence intervals cross zero. From the table, the number of blocks with a positive mean paired difference for average-case payoff, H1: $13/24$, $13/24$ and $12/24$ for the Red opponents respectively.

H2, $\operatorname{PRMRE}-\operatorname{MRE}$: PRMRE outperforms MRE in terms of lower in-band worst-case regret against the Red BNE $r_{0}$ and PRMRE $r_{2}$ with mean paired differences of -0.628 and -0.059 respectively and both confidence intervals excluding zero. The 24 block experiment remains inconclusive against Red MRE where the mean paired difference is -0.039 (negative) but the confidence interval crosses zero $[-0.093, +0.013]$. For the secondary metrics as well, PRMRE's advantage over MRE is opponent dependent. PRMRE achieves higher in-band worst-case and average-case payoffs than MRE with mean paired differences of +0.660, +0.349 against Red $r_{0}$ and +0.066, +0.056 against Red $r_{2}$ respectively with all respective confidence intervals excluding zero. Against Red $r_{1}$, PRMRE has higher mean worst-case and average-case payoffs but both the confidence intervals cross zero and hence the test remains inconclusive. Therefore in summary, against Red MRE $r_{1}$, the 24 block experiment returns inconclusive for H2 across all metrics.

H3, $\operatorname{PRMRE}-\operatorname{DRO-MaxMin}$: Against DRO-MaxMin, worst-case-regret intervals overlap zero for all three Red opponents. The secondary metrics reveal limited separation: DRO-MaxMin achieves higher worst-case payoff against Red $r_{1}$, whereas PR-MRE achieves higher average-case payoff against Red $r_{2}$. The remaining intervals overlap zero. The 24 block experiment and the map (Figure~\ref{wrap-fig:atg}) therefore do not establish a consistent practical advantage of PRMRE over DRO-MaxMin. Moreover, the unresolved comparisons do not establish equivalence: the 95\% confidence intervals for the mean worst-case regret differences are [-0.033, 0.057], [-0.012, 0.078], and [-0.036, 0.046] against Reds $r_0$--$r_2$, respectively, remaining compatible with modest differences in either direction. Figures~\ref{fig:wcreg_per_row}, \ref{fig:wcpay_per_row} and \ref{fig:avgpay_per_row} show the block-level variability underlying these aggregate results as noted before.


\begin{figure}
  \centering
  \includegraphics[scale=0.6]{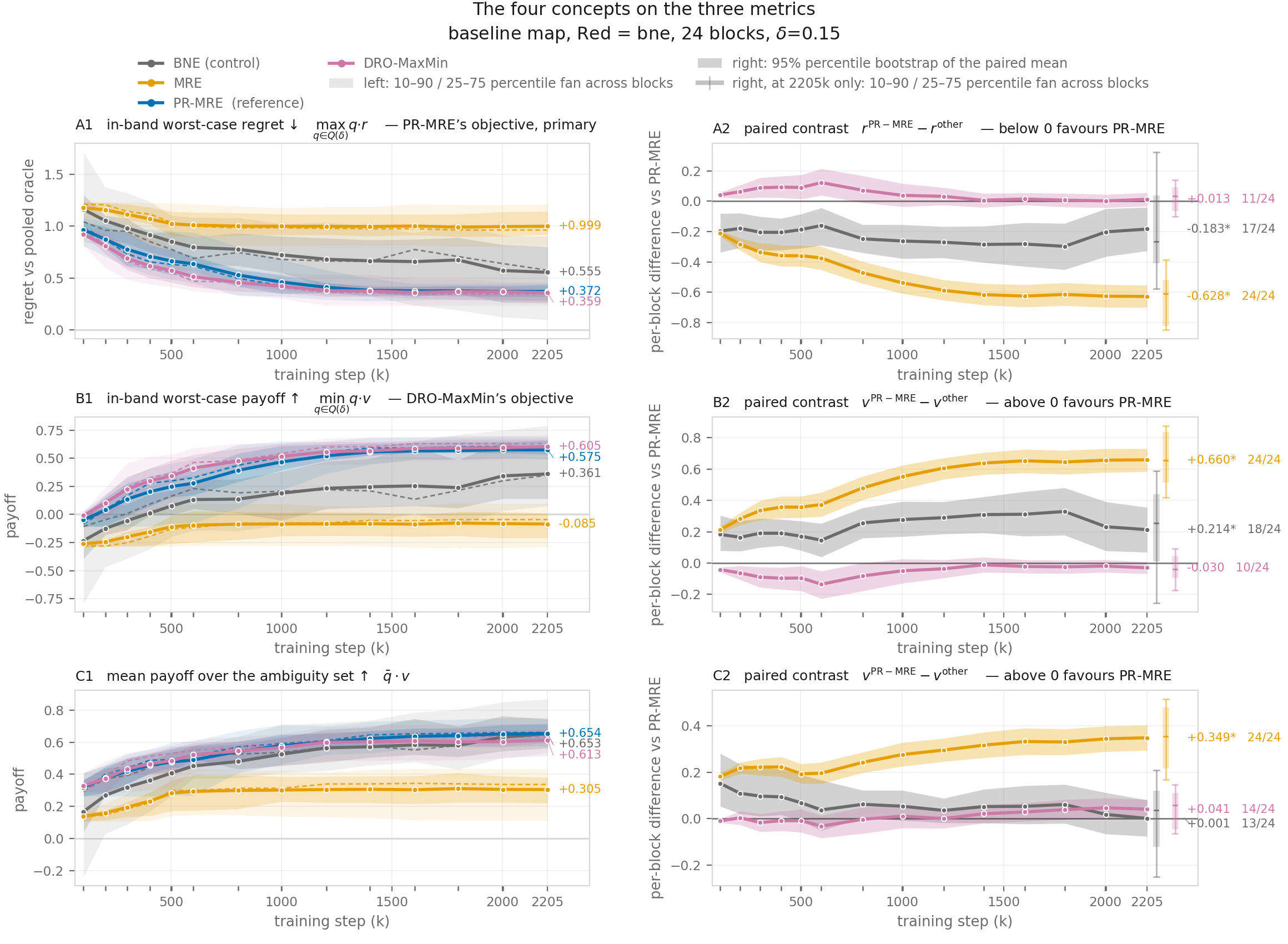}
\caption{\color{black}{\textbf{Metrics evolution during training across blocks. Baseline map (fig.~\ref{wrap-fig:atg}), Red row $r_{0}$ \texttt{BNE}}. The four concepts on the primary and the two seconday metrics, $n=24$ blocks, $\delta=0.15$. Rows: in-band worst-case regret against the per-block frozen pooled benchmark; in-band worst-case payoff; mean payoff over the ambiguity set. Arrows give the favourable direction. \textbf{Left panel}: descriptive only -- the blocks are paired, so overlapping bands neither establish nor rule out a difference. \textbf{Right panel}: the paired per-block contrast against PR-MRE: H1-H3. The solid band is the 95\% percentile bootstrap ($B=20{,}000$) of the paired mean; the bar at the final checkpoint is the 10--90 / 25--75 percentile fan of the per-block differences. Labels give the final paired mean and the blocks favouring PR-MRE.}} 
\label{fig:regretcurves-psro-grid-bne-paired}
\end{figure}

\begin{table}[t]
\centering
\small
{\color{black}\caption{\textbf{Aggregate metrics for the learned robust best-responses}: Paired contrasts against PRMRE by Red row, $n=24$ randomised complete blocks; one block is a shared warm-start clone and training seed, so the block effect cancels in the within-block difference. Negative favours PRMRE on regret, positive on payoff. \textbf{Bold} marks an interval excluding zero. $^{\dagger}$95\% pointwise block-bootstrap intervals (percentile method, 20{,}000 resamples of whole blocks). ``Sign agree'' counts blocks sharing the sign of $\bar d$. The three rows corresponding to each Red opponent share the same blocks and are therefore not independent.}}
\setlength{\tabcolsep}{5pt}
\vspace{1em}
{\color{black}
\begin{tabular}{@{}ll ccc@{}}
\toprule
& & \multicolumn{3}{c}{In-band metric} \\
\cmidrule(l){3-5}
Contrast & & WC regret & Avg.\ payoff & WC payoff \\
\midrule
\multicolumn{5}{@{}l}{\textit{Red row, $r_{0}$: BNE}} \\
\addlinespace[2pt]
\multirow{3}{*}{\hspace{1em}H1: PRMRE $-$ BNE} & $\bar d \pm \mathrm{sd}$ & $\boldsymbol{-0.183 \pm 0.368}$ & $+0.001 \pm 0.200$ & $\boldsymbol{+0.214 \pm 0.367}$ \\
 & 95\% CI$^{\dagger}$ & $\boldsymbol{[-0.325,\,-0.039]}$ & $[-0.076,\,+0.078]$ & $\boldsymbol{[+0.066,\,+0.356]}$ \\
 & sign agree & $17/24$ & $13/24$ & $18/24$ \\
\addlinespace[3pt]
\multirow{3}{*}{\hspace{1em}H2: PRMRE $-$ MRE} & $\bar d \pm \mathrm{sd}$ & $\boldsymbol{-0.628 \pm 0.189}$ & $\boldsymbol{+0.349 \pm 0.139}$ & $\boldsymbol{+0.660 \pm 0.190}$ \\
 & 95\% CI$^{\dagger}$ & $\boldsymbol{[-0.701,\,-0.552]}$ & $\boldsymbol{[+0.296,\,+0.403]}$ & $\boldsymbol{[+0.584,\,+0.732]}$ \\
 & sign agree & $24/24$ & $24/24$ & $24/24$ \\
\addlinespace[3pt]
\multirow{3}{*}{\hspace{1em}H3: PRMRE $-$ DRO-MaxMin} & $\bar d \pm \mathrm{sd}$ & $+0.013 \pm 0.115$ & $+0.041 \pm 0.109$ & $-0.030 \pm 0.099$ \\
 & 95\% CI$^{\dagger}$ & $[-0.033,\,+0.057]$ & $[-0.001,\,+0.084]$ & $[-0.069,\,+0.008]$ \\
 & sign agree & $13/24$ & $14/24$ & $14/24$ \\
\midrule
\multicolumn{5}{@{}l}{\textit{Red row, $r_{1}$: MRE}} \\
\addlinespace[2pt]
\multirow{3}{*}{\hspace{1em}H1: PRMRE $-$ BNE} & $\bar d \pm \mathrm{sd}$ & $\boldsymbol{-0.146 \pm 0.361}$ & $-0.006 \pm 0.192$ & $\boldsymbol{+0.164 \pm 0.367}$ \\
 & 95\% CI$^{\dagger}$ & $\boldsymbol{[-0.288,\,-0.005]}$ & $[-0.080,\,+0.071]$ & $\boldsymbol{[+0.019,\,+0.307]}$ \\
 & sign agree & $14/24$ & $13/24$ & $14/24$ \\
\addlinespace[3pt]
\multirow{3}{*}{\hspace{1em}H2: PRMRE $-$ MRE} & $\bar d \pm \mathrm{sd}$ & $-0.039 \pm 0.136$ & $+0.029 \pm 0.114$ & $+0.039 \pm 0.150$ \\
 & 95\% CI$^{\dagger}$ & $[-0.093,\,+0.013]$ & $[-0.014,\,+0.075]$ & $[-0.018,\,+0.100]$ \\
 & sign agree & $13/24$ & $13/24$ & $13/24$ \\
\addlinespace[3pt]
\multirow{3}{*}{\hspace{1em}H3: PRMRE $-$ DRO-MaxMin} & $\bar d \pm \mathrm{sd}$ & $+0.033 \pm 0.114$ & $+0.028 \pm 0.088$ & $\boldsymbol{-0.048 \pm 0.110}$ \\
 & 95\% CI$^{\dagger}$ & $[-0.012,\,+0.078]$ & $[-0.007,\,+0.061]$ & $\boldsymbol{[-0.092,\,-0.006]}$ \\
 & sign agree & $14/24$ & $16/24$ & $15/24$ \\
\midrule
\multicolumn{5}{@{}l}{\textit{Red row, $r_{2}$: PRMRE}} \\
\addlinespace[2pt]
\multirow{3}{*}{\hspace{1em}H1: PRMRE $-$ BNE} & $\bar d \pm \mathrm{sd}$ & $\boldsymbol{-0.155 \pm 0.303}$ & $+0.009 \pm 0.165$ & $\boldsymbol{+0.196 \pm 0.305}$ \\
 & 95\% CI$^{\dagger}$ & $\boldsymbol{[-0.276,\,-0.038]}$ & $[-0.055,\,+0.075]$ & $\boldsymbol{[+0.079,\,+0.316]}$ \\
 & sign agree & $13/24$ & $12/24$ & $15/24$ \\
\addlinespace[3pt]
\multirow{3}{*}{\hspace{1em}H2: PRMRE $-$ MRE} & $\bar d \pm \mathrm{sd}$ & $\boldsymbol{-0.059 \pm 0.140}$ & $\boldsymbol{+0.056 \pm 0.112}$ & $\boldsymbol{+0.066 \pm 0.150}$ \\
 & 95\% CI$^{\dagger}$ & $\boldsymbol{[-0.114,\,-0.004]}$ & $\boldsymbol{[+0.013,\,+0.101]}$ & $\boldsymbol{[+0.007,\,+0.126]}$ \\
 & sign agree & $15/24$ & $16/24$ & $15/24$ \\
\addlinespace[3pt]
\multirow{3}{*}{\hspace{1em}H3: PRMRE $-$ DRO-MaxMin} & $\bar d \pm \mathrm{sd}$ & $+0.005 \pm 0.105$ & $\boldsymbol{+0.051 \pm 0.099}$ & $-0.008 \pm 0.104$ \\
 & 95\% CI$^{\dagger}$ & $[-0.036,\,+0.046]$ & $\boldsymbol{[+0.013,\,+0.090]}$ & $[-0.049,\,+0.034]$ \\
 & sign agree & $13/24$ & $15/24$ & $13/24$ \\
\bottomrule
\end{tabular}
}
\label{tab:contrasts-per-row-n24}
\end{table}

\textbf{How do the learned robust BRs perform across admissible test-time type distributions?} Table~\ref{tab:contrasts-per-row-n24} reports worst- and average-case aggregate metrics across the admissible set of test-time type distributions. We visualize performance of the learned robust best-responses corresponding to different Blue concepts $c_{0}$--$c_{3}$ across the admissible test-time type distributions against the different Red opponents in Figures~\ref{fig:type_simplex_baseline_map_red_bne}, \ref{fig:simplex-psro-grid-mre-spread} and \ref{fig:simplex-psro-grid-prmre-spread}. The upper panels show the magnitude and direction of mean payoff differences, while the lower panels show agreement across blocks. The outlined trapezoidal region denotes the admissible test-time distribution set; results outside it are descriptive.

\textbf{Summary and Limitations}. To summarize, for the 24-block experiment on the considered map (Figure~\ref{wrap-fig:atg}), PR-MRE achieves lower worst-case regret than BNE against all three Red opponents, with pointwise 95\% confidence intervals excluding zero. The comparison with MRE is opponent-dependent, while worst-case-regret comparisons with DRO-MaxMin remain inconclusive against all three opponents. These comparisons are conditional on the map, the seed metagame, fixed Red mixtures, and initialization from a clone of the sweep policy. The observed H2 differences may reflect interactions between this initialization and the training objectives; alternative clone initializations were not evaluated. The auxiliary double-oracle loop uses population-restricted benchmarks and approximate RL responses, so the learned mixtures need not attain the full-game best-response guarantees. Computational constraints limited evaluation to 24 blocks. Additional independent blocks would generally improve precision, and a prospective power analysis using a prespecified meaningful effect size and estimated paired-difference variance could guide the required sample size. Narrower intervals, however, would not necessarily establish a difference.


\textbf{Motivation for a Secondary Map}. Inspection of Nature’s training distributions in block 08 (Panel B of the mechanism figures) revealed that PR-MRE and DRO-MaxMin both concentrated substantial probability on the Right flag. For example, their initial adversarial distributions were $(0.15, 0.85, 0)$ and $(0, 1, 0)$, respectively. This motivated an additional experiment on a map designed to separate regret-maximizing and payoff-minimizing type distributions. We moved the Middle flag deeper into its corridor and the Right flag closer to the basin. In the resulting seed metagame, the Middle flag yields low payoffs (difficult type) even for the best available type-conditioned response, leaving limited scope for regret. The closer Right flag (easy type) offers a higher attainable payoff, creating greater regret when Blue fails a Right flag capture. Consequently, regret-based Nature emphasizes Right, whereas payoff-based Nature emphasizes Middle, inducing different training trajectories.

\begin{figure}
  \centering
  \includegraphics[scale=0.55]{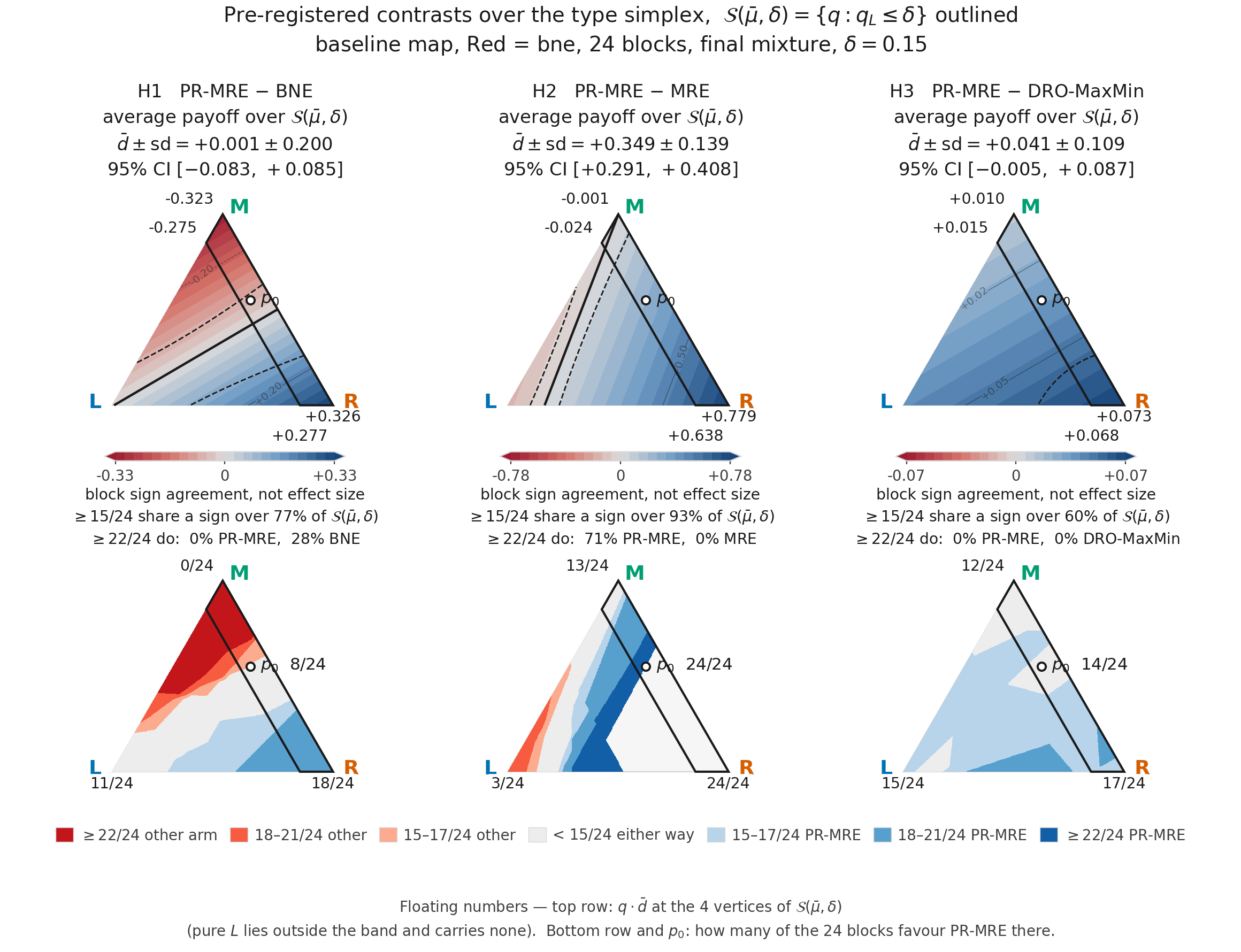}
\caption{{\color{black} \textbf{How do the learned robust best-responses perform across the type-simplex? Baseline map (fig.~\ref{wrap-fig:atg}), Red row $r_{0}$ \texttt{BNE}}. H1--H3 as heatmaps over the type simplex, $n=24$ blocks, $\delta=0.15$. Each value-vector $v \in \mathbb{R}^{\lvert\Theta\rvert}$ corresponds to the final mixture produced by the robust Blue concept (and the final trained team policy at 2200K for BNE) -- the same per-block vectors the H1--H3 table (Table~\ref{tab:contrasts-per-row-n24}) uses. \textbf{Top row}: the fill at a type distribution $q$ is $q\cdot\bar d$, the mean over blocks of $q\cdot(v^{A}(b)-v^{B}(b))$ where $v^{A}(b)$ represents the value-vector for the final trained robust mixture for Blue concept $A$ in block $b$. Averaging the vectors and averaging the heatmaps across blocks agree exactly, because $q\cdot d$ is linear in $d$; the same shortcut is invalid for the worst-case metrics, where $\max$ and mean do not commute, so those are computed per block and then averaged as reported in Table~\ref{tab:contrasts-per-row-n24}. The heavy outline is the cover region $S(\bar{\mu}, \delta)=\{q: q_L\leq\delta\}$ that every arm optimised over; outside it the paired-difference field is the same linear function $q. \bar{d}$, but no arm optimised there and no pre-registered metric lives there. The pre-registered inference remains the interval set of Table~\ref{tab:contrasts-per-row-n24}. Each panel carries its own colour scale (see the bar beneath it), so hue is not comparable across panels; the contour labels are. \textbf{Bottom row}: how strongly the blocks agree, drawn as nested majorities at 15/24, 18/24, 22/24. This is CONSENSUS, NOT SIZE: it counts signs only, so a block favouring an arm by 0.001 counts the same as one favouring it by 0.5. Magnitude is the top row. Corner and $p_0$ annotations give the count of blocks favouring PR-MRE at that type distribution.}}  \label{fig:type_simplex_baseline_map_red_bne}
\end{figure}

\textbf{Secondary Map Results}, Figure~\ref{fig:1_0_3_map}. In the seed metagame (Figure~\ref{fig:seed_metagame_3flag_103}), the adversarial type distributions are $(0.15, 0.85, 0)$ for PR-MRE and $(0.15, 0, 0.85)$ for DRO-MaxMin, ordered as (Left, Right, Middle). Thus, regret-based Nature emphasizes Right, whereas payoff-based Nature emphasizes Middle. Against BNE-Red on this map, PR-MRE achieves lower worst-case regret than DRO-MaxMin, with a mean paired difference of -0.0652 and a pointwise 95\% confidence interval of [-0.115, -0.013]. PR-MRE also achieves higher average payoff, whereas DRO-MaxMin achieves higher worst-case payoff, illustrating the tradeoff between regret-based and payoff-based robustness. The worst-case-regret comparison with MRE remains inconclusive. These results provide evidence of map-dependent separation between the learned robust best-responses.
}
\begin{table}[t]
\centering
\small
{\color{black}\caption{\textbf{Aggregate metrics for the learned robust best-responses, secondary map (Fig.~\ref{fig:1_0_3_map} and seed metagame Fig.~\ref{fig:seed_metagame_3flag_103})}: Paired contrasts against PRMRE by Red row, $n=24$ randomised complete blocks; one block is a shared warm-start clone and training seed, so the block effect cancels in the within-block difference. Negative favours PRMRE on regret, positive on payoff. \textbf{Bold} marks an interval excluding zero. $^{\dagger}$95\% pointwise block-bootstrap intervals (percentile method, 20{,}000 resamples of whole blocks). ``Sign agree'' counts blocks sharing the sign of $\bar d$. The three rows corresponding to each Red opponent share the same blocks and are therefore not independent.}}
\setlength{\tabcolsep}{5pt}
\vspace{1em}
{\color{black}
\begin{tabular}{@{}ll ccc@{}}
\toprule
& & \multicolumn{3}{c}{In-band metric} \\
\cmidrule(l){3-5}
Contrast & & WC regret & Avg.\ payoff & WC payoff \\
\midrule
\multicolumn{5}{@{}l}{\textit{Red row, $r_{0}$: BNE}} \\
\addlinespace[2pt]
\multirow{3}{*}{\hspace{1em}H1: PRMRE $-$ BNE} & $\bar d \pm \mathrm{sd}$ & $\boldsymbol{-0.321 \pm 0.542}$ & $\boldsymbol{+0.207 \pm 0.466}$ & $\boldsymbol{+0.336 \pm 0.452}$ \\
 & 95\% CI$^{\dagger}$ & $\boldsymbol{[-0.544,\,-0.121]}$ & $\boldsymbol{[+0.042,\,+0.404]}$ & $\boldsymbol{[+0.166,\,+0.526]}$ \\
 & sign agree & $19/24$ & $14/24$ & $19/24$ \\
\addlinespace[3pt]
\multirow{3}{*}{\hspace{1em}H2: PRMRE $-$ MRE} & $\bar d \pm \mathrm{sd}$ & $-0.061 \pm 0.195$ & $\boldsymbol{+0.086 \pm 0.137}$ & $\boldsymbol{+0.109 \pm 0.111}$ \\
 & 95\% CI$^{\dagger}$ & $[-0.135,\,+0.018]$ & $\boldsymbol{[+0.032,\,+0.139]}$ & $\boldsymbol{[+0.064,\,+0.150]}$ \\
 & sign agree & $16/24$ & $16/24$ & $19/24$ \\
\addlinespace[3pt]
\multirow{3}{*}{\hspace{1em}H3: PRMRE $-$ DRO-MaxMin} & $\bar d \pm \mathrm{sd}$ & $\boldsymbol{-0.065 \pm 0.130}$ & $\boldsymbol{+0.026 \pm 0.058}$ & $\boldsymbol{-0.044 \pm 0.055}$ \\
 & 95\% CI$^{\dagger}$ & $\boldsymbol{[-0.115,\,-0.013]}$ & $\boldsymbol{[+0.003,\,+0.049]}$ & $\boldsymbol{[-0.067,\,-0.023]}$ \\
 & sign agree & $14/24$ & $14/24$ & $19/24$ \\
\midrule
\bottomrule
\end{tabular}
}
\label{tab:contrasts-per-row-n24-103}
\end{table}

\section{Conclusions and Future Work}
Adversarial team games with asymmetric information are susceptible to strategic type distribution shifts in settings where the omniscient opponent has knowledge of Nature’s hidden type and potentially controls the type distribution or colludes with Nature such that it can condition its play on the Nature’s realized type. A nominal prior distribution in such strategic multi-agent interactions is “trustable” only to a certain degree in the face of potential deception. Distributionally-robust methods provide a framework to address ambiguity in the type distribution via ambiguity sets but lack guarantees for shifts outside the ambiguity set for instance strategic distribution shifts whereby Nature could be running an oracle procedure to compute a type distribution that places probability mass on high-regret types. Moreover, a worst-case approach that assumes an adversarial nature could be too conservative as it concentrates probability mass on the hardest type and loses performance under non-adversarial distribution shifts.

Motivated by the {\color{black}point-wise benchmark-relative performance} lower bound interpretation of minimax reasoning, we propose probabilistically-robust minimax-regret equilibrium as a solution concept for adversarial team games under asymmetric information for robustness to strategic distribution shifts. PR-MRE best-response discounts regret over type-space subsets with coarse probabilistic information from the nominal prior distribution. This is in contrast to minimax-regret equilibrium that discards potentially useful probabilistic information about type occurrence by minimizing worst-case regret across all types. MRE leads to conservative performance on typical types when high-regret types are known to be rare for example rare flag hypothesis locations which require a remarkably different Blue team behavior for capture as compared to the typical flag hypotheses. {\color{black}Approximate PR-MRE fixed-point computation is formulated as a robust bilinear program} for finite normal-form games and is adapted as a meta-solver within a robust double-oracle based approach, PRMRE-PSRO to learn strategically robust team policies via reinforcement learning based best-responses. For an example graph capture-the-flag game, team strategies learned via PRMRE PSRO exhibit enhanced robustness to distribution shift as compared to BNE PSRO.

Computation of sequential refinements of minimax-regret equilibria for treeplexes in extensive-form games (EFGs) remains an open question and a subject of future research. The nested optimization structure of MRE that characterizes deviation from type-conditioned optima and the {\color{black}surrogate} bilinear program for MRE in normal-form games as opposed to a linear program for BNE hints at non-trivial computational complexity such that a naive approach to compute the MRE would first solve an EFG corresponding to each type. It remains an open question whether regret decomposition across infosets holds for MRE in the EFG representation as it does for the computation of BNE in the established counterfactual regret minimization literature.


\section*{Acknowledgments}
This work was supported in part by the U.S. Army Research Laboratory through the Distributed and Collaborative Intelligent Systems and Technology (DCIST) Collaborative Research Alliance under Cooperative Agreement No.~W911NF-17-2-0181. The views and conclusions contained in this document are those of the authors and should not be interpreted as representing the official policies, either expressed or implied, of the U.S. Army Research Laboratory or the U.S. Government.

\bibliography{main}

@article{lowrank_sdp_recht_parillo,
  title={Guaranteed minimum-rank solutions of linear matrix equations via nuclear norm minimization},
  author={Recht, Benjamin and Fazel, Maryam and Parrilo, Pablo A},
  journal={SIAM review},
  volume={52},
  number={3},
  pages={471--501},
  year={2010},
  publisher={SIAM}
}

@article{ali_zhang_bimatrix_nash_sdp,
author = {Ahmadi, Amir Ali and Zhang, Jeffrey},
title = {Semidefinite Programming and Nash Equilibria in Bimatrix Games},
year = {2021},
issue_date = {Spring 2021},
publisher = {INFORMS},
address = {Linthicum, MD, USA},
volume = {33},
number = {2},
issn = {1526-5528},
url = {https://doi.org/10.1287/ijoc.2020.0960},
doi = {10.1287/ijoc.2020.0960},
journal = {INFORMS J. on Computing},
month = may,
pages = {607–628},
numpages = {22}
}

@article{moment_sos1,
author = {Ahmadi, Amir Ali and Hall, Georgina},
title = {On the Construction of Converging Hierarchies for Polynomial Optimization Based on Certificates of Global Positivity},
year = {2019},
issue_date = {November 2019},
publisher = {INFORMS},
address = {Linthicum, MD, USA},
volume = {44},
number = {4},
issn = {0364-765X},
url = {https://doi.org/10.1287/moor.2018.0962},
doi = {10.1287/moor.2018.0962},
journal = {Math. Oper. Res.},
month = nov,
pages = {1192–1207},
numpages = {16}
}

@article{moment_sos2,
author = {Chen, Tong and Lasserre, Jean-Bernard and Magron, Victor and Pauwels, Edouard},
title = {A sublevel moment-SOS hierarchy for polynomial optimization},
year = {2022},
issue_date = {Jan 2022},
publisher = {Kluwer Academic Publishers},
address = {USA},
volume = {81},
number = {1},
issn = {0926-6003},
url = {https://doi.org/10.1007/s10589-021-00325-z},
doi = {10.1007/s10589-021-00325-z},
journal = {Comput. Optim. Appl.},
month = jan,
pages = {31–66},
numpages = {36}
}

@phdthesis{fazel2002matrix,
  title={Matrix rank minimization with applications},
  author={Fazel, Maryam}
}

@inproceedings{psro_silver,
author = {Lanctot, Marc and Zambaldi, Vinicius and Gruslys, Audr\={u}nas and Lazaridou, Angeliki and Tuyls, Karl and P\'{e}rolat, Julien and Silver, David and Graepel, Thore},
title = {A unified game-theoretic approach to multiagent reinforcement learning},
year = {2017},
isbn = {9781510860964},
publisher = {Curran Associates Inc.},
address = {Red Hook, NY, USA},
booktitle = {Proceedings of the 31st International Conference on Neural Information Processing Systems},
pages = {4193–4206},
numpages = {14},
location = {Long Beach, California, USA},
series = {NIPS'17}
}

@article{harsanyi_1,
author = {Harsanyi, John C.},
title = {Games with Incomplete Information Played by “Bayesian” Players, I–III Part I. The Basic Model},
journal = {Management Science},
volume = {14},
number = {3},
pages = {159-182},
year = {1967},
doi = {10.1287/mnsc.14.3.159},

URL = { 
    
        https://doi.org/10.1287/mnsc.14.3.159
    
    

},
eprint = { 
    
        https://doi.org/10.1287/mnsc.14.3.159
    
    

}
}

@article{harsanyi_2,
 ISSN = {00251909, 15265501},
 URL = {http://www.jstor.org/stable/2628673},
 author = {John C. Harsanyi},
 journal = {Management Science},
 number = {5},
 pages = {320--334},
 publisher = {INFORMS},
 title = {Games with Incomplete Information Played by "Bayesian" Players, I-III. Part II. Bayesian Equilibrium Points},
 urldate = {2025-10-09},
 volume = {14},
 year = {1968}
}

@inproceedings{mre, author = {Hyafil, Nathanael and Boutilier, Craig}, title = {Regret minimizing equilibria and mechanisms for games with strict type uncertainty}, year = {2004}, isbn = {0974903906}, publisher = {AUAI Press}, address = {Arlington, Virginia, USA}, booktitle = {Proceedings of the 20th Conference on Uncertainty in Artificial Intelligence}, pages = {268–277}, numpages = {10}, location = {Banff, Canada}, series = {UAI '04} }

@article{ouyang2016dynamic,
  title={Dynamic games with asymmetric information: Common information based perfect bayesian equilibria and sequential decomposition},
  author={Ouyang, Yi and Tavafoghi, Hamidreza and Teneketzis, Demosthenis},
  journal={IEEE Transactions on Automatic Control},
  volume={62},
  number={1},
  pages={222--237},
  year={2016},
  publisher={IEEE}
}

@misc{konicki2026computingperfectbayesianequilibria,
      title={Computing Perfect Bayesian Equilibria, with Application to Empirical Game-Theoretic Analysis}, 
      author={Christine Konicki and Mithun Chakraborty and Michael P. Wellman},
      year={2026},
      eprint={2602.15233},
      archivePrefix={arXiv},
      primaryClass={cs.GT},
      url={https://arxiv.org/abs/2602.15233}, 
}

@inproceedings{treeplexes_efg_kroer,
author = {Chakrabarti, Darshan and Grand-Cl\'{e}ment, Julien and Kroer, Christian},
title = {Extensive-form game solving via blackwell approachability on treeplexes},
year = {2024},
isbn = {9798331314385},
publisher = {Curran Associates Inc.},
address = {Red Hook, NY, USA},
booktitle = {Proceedings of the 38th International Conference on Neural Information Processing Systems},
articleno = {1111},
numpages = {31},
location = {Vancouver, BC, Canada},
series = {NIPS '24}
}

@misc{carminati2022marriageadversarialteamgames,
      title={A Marriage between Adversarial Team Games and 2-player Games: Enabling Abstractions, No-regret Learning, and Subgame Solving}, 
      author={Luca Carminati and Federico Cacciamani and Marco Ciccone and Nicola Gatti},
      year={2022},
      eprint={2206.09161},
      archivePrefix={arXiv},
      primaryClass={cs.GT},
      url={https://arxiv.org/abs/2206.09161}, 
}

@inproceedings{tme_cor_psro,
author = {McAleer, Stephen and Farina, Gabriele and Zhou, Gaoyue and Wang, Mingzhi and Yang, Yaodong and Sandholm, Tuomas},
title = {Team-PSRO for learning approximate TMECor in large team games via cooperative reinforcement learning},
year = {2023},
publisher = {Curran Associates Inc.},
address = {Red Hook, NY, USA},
booktitle = {Proceedings of the 37th International Conference on Neural Information Processing Systems},
articleno = {1967},
numpages = {17},
location = {New Orleans, LA, USA},
series = {NIPS '23}
}

@inproceedings{sokkota2026,
    title = {Reevaluating Policy Gradient Methods for Imperfect-Information Games},
    author = {Max Rudolph and Nathan Lichtlé and Sobhan Mohammadpour and Alexandre M. Bayen and J. Zico Kolter and Amy Zhang and Gabriele Farina and Eugene Vinitsky and Samuel Sokota},
    booktitle = {International Conference on Learning Representations (ICLR)},
    year = {2026},
}

@article{mappo,
  title={The surprising effectiveness of ppo in cooperative multi-agent games},
  author={Yu, Chao and Velu, Akash and Vinitsky, Eugene and Gao, Jiaxuan and Wang, Yu and Bayen, Alexandre and Wu, Yi},
  journal={Advances in neural information processing systems},
  volume={35},
  pages={24611--24624},
  year={2022}
}

@inproceedings{mcmahan2003planning,
  title={Planning in the presence of cost functions controlled by an adversary},
  author={McMahan, H Brendan and Gordon, Geoffrey J and Blum, Avrim},
  booktitle={Proceedings of the 20th international conference on machine learning (ICML-03)},
  pages={536--543},
  year={2003}
}

@article{renou2010minimax,
  title={Minimax regret and strategic uncertainty},
  author={Renou, Ludovic and Schlag, Karl H},
  journal={Journal of Economic Theory},
  volume={145},
  number={1},
  pages={264--286},
  year={2010},
  publisher={Elsevier}
}

@article{renou2011implementation,
  author    = {Renou, Ludovic and Schlag, Karl H.},
  title     = {Implementation in minimax regret equilibrium},
  journal   = {Games and Economic Behavior},
  volume    = {71},
  number    = {2},
  pages     = {527--533},
  year      = {2011},
  issn      = {0899-8256},
  doi       = {10.1016/j.geb.2010.05.010},
  url       = {https://doi.org/10.1016/j.geb.2010.05.010}
}

@article{anagnostides2026algorithms,
  author     = {Anagnostides, Ioannis and Kalogiannis, Fivos and Panageas, Ioannis and Vlatakis-Gkaragkounis, Emmanouil-Vasileios and Mcaleer, Stephen},
  title     = {Algorithms and complexity for computing Nash equilibria in adversarial team games},
  journal   = {Games and Economic Behavior},
  volume    = {157},
  pages     = {138--152},
  year      = {2026},
  doi       = {10.1016/j.geb.2026.01.006}
}

@inproceedings{atgs_neurips_2024,
author = {Kalogiannis, Fivos and Yan, Jingming and Panageas, Ioannis},
title = {Learning equilibria in adversarial team Markov games: a nonconvex-hidden-concave min-max optimization problem},
year = {2024},
isbn = {9798331314385},
publisher = {Curran Associates Inc.},
address = {Red Hook, NY, USA},
booktitle = {Proceedings of the 38th International Conference on Neural Information Processing Systems},
articleno = {2948},
numpages = {59},
location = {Vancouver, BC, Canada},
series = {NIPS '24}
}

@article{xiao2026solving,
  author     = {Xiao, Jinheng and Qiu, Chen and Xu, Yingying and Zhang, Jiajia and Qi, Shuhan and Wang, Xuan},
  title     = {Solving equilibrium for adversarial team games utilizing fictitious team play with refined team plans},
  journal   = {Expert Systems with Applications},
  volume    = {298},
  part      = {B},
  pages     = {129496},
  year      = {2026},
  issn      = {0957-4174},
  doi       = {10.1016/j.eswa.2025.129496},
  url       = {https://www.sciencedirect.com/science/article/pii/S0957417425031112}
}

@InProceedings{pmlr-v235-zhang24b,
  title = 	 {{DAG}-Based Column Generation for Adversarial Team Games},
  author =       {Zhang, Youzhi and An, Bo and Zeng, Daniel Dajun},
  booktitle = 	 {Proceedings of the 41st International Conference on Machine Learning},
  pages = 	 {58563--58578},
  year = 	 {2024},
  editor = 	 {Salakhutdinov, Ruslan and Kolter, Zico and Heller, Katherine and Weller, Adrian and Oliver, Nuria and Scarlett, Jonathan and Berkenkamp, Felix},
  volume = 	 {235},
  series = 	 {Proceedings of Machine Learning Research},
  month = 	 {21--27 Jul},
  publisher =    {PMLR},
  url = 	 {https://proceedings.mlr.press/v235/zhang24b.html}
}

@ARTICLE{Celli2018-yn,
  title     = "Computational results for extensive-form adversarial team games",
  author    = "Celli, Andrea and Gatti, Nicola",
  journal   = "Proc. Conf. AAAI Artif. Intell.",
  publisher = "Association for the Advancement of Artificial Intelligence
               (AAAI)",
  volume    =  32,
  number    =  1,
  month     =  apr,
  year      =  2018
}

@InProceedings{pmlr-v162-carminati22a,
  title = 	 {A Marriage between Adversarial Team Games and 2-player Games: Enabling Abstractions, No-regret Learning, and Subgame Solving},
  author =       {Carminati, Luca and Cacciamani, Federico and Ciccone, Marco and Gatti, Nicola},
  booktitle = 	 {Proceedings of the 39th International Conference on Machine Learning},
  pages = 	 {2638--2657},
  year = 	 {2022},
  editor = 	 {Chaudhuri, Kamalika and Jegelka, Stefanie and Song, Le and Szepesvari, Csaba and Niu, Gang and Sabato, Sivan},
  volume = 	 {162},
  series = 	 {Proceedings of Machine Learning Research},
  month = 	 {17--23 Jul},
  publisher =    {PMLR},
  url = 	 {https://proceedings.mlr.press/v162/carminati22a.html}
}

@inproceedings{
anagnostides2023efficiently,
title={Efficiently Computing Nash Equilibria in Adversarial Team Markov Games},
author={Ioannis Anagnostides and Vaggos Chatziafratis and Fivos Kalogiannis and Manolis Vlatakis and Stelios Stavroulakis},
booktitle={International Conference on Learning Representations (ICLR)},
year={2023},
url={https://openreview.net/forum?id=mjzm6btqgV}
}

@article{zinkevich2007regret,
  title={Regret minimization in games with incomplete information},
  author={Zinkevich, Martin and Johanson, Michael and Bowling, Michael and Piccione, Carmelo},
  journal={Advances in neural information processing systems},
  volume={20},
  year={2007}
}

@inproceedings{slumbers_rae_psro,
author = {Slumbers, Oliver and Mguni, David Henry and Blumberg, Stefano B. and McAleer, Stephen and Yang, Yaodong and Wang, Jun},
title = {A game-theoretic framework for managing risk in multi-agent systems},
year = {2023},
publisher = {JMLR.org},
booktitle = {Proceedings of the 40th International Conference on Machine Learning},
articleno = {1329},
numpages = {29},
location = {Honolulu, Hawaii, USA},
series = {ICML'23}
}

@inproceedings{mcaleer2023teampsro,
  title={Team-PSRO for Learning Approximate TMECor in Large Team Games via Cooperative Reinforcement Learning},
  author={McAleer, Stephen and Farina, Gabriele and Zhou, Guanzhi and Wang, Mingzhi and Yang, Yining and Sandholm, Tuomas},
  booktitle={Advances in Neural Information Processing Systems},
  volume={36},
  pages={56872--56897},
  year={2023}
}
\bibliographystyle{tmlr}

\appendix
\section{Appendix}
\label{sec:appendix_proof_existence}
{\color{black}
\subsection{Proof of Existence of PR-MRE for General Asymmetric Information Games}
\label{sec:}
\begin{proof}
Fix $(\bar\mu,\delta)$ and write
$\mathcal S:=\mathcal S(\bar\mu,\delta)$.
For a Red team strategy $\pi_r(\cdot)\in\Pi_r$,
define the expected regret of Blue strategy
$\pi_b\in\Pi_b$ under $\mu\in\mathcal S$ as
\begin{equation}
\mathcal G(\pi_r(\cdot),\pi_b,\mu)
:=
\sum_{\theta\in\Theta}
\mu(\theta)\mathcal E(\pi_r(\cdot),\pi_b;\theta),
\end{equation}
where $\mathcal E$ is the type-conditioned
exploitability defined in Definition~\ref{def:exp}.

By continuity of $V$ and compactness of $\Pi_b$,
Berge's maximum theorem implies that
$\mathcal E(\pi_r(\cdot),\pi_b;\theta)$
is continuous in $(\pi_r(\cdot),\pi_b)$.
Consequently, $\mathcal G$ is jointly continuous,
convex in $\pi_b$ for fixed $\pi_r(\cdot)$,
and linear in $\mu$.

For a fixed Red strategy $\pi_r(\cdot)$,
consider the Blue--Nature minimax problem
\begin{align}
\min_{\pi_b\in\Pi_b}\max_{\mu\in\mathcal S}
\mathcal G(\pi_r(\cdot),\pi_b,\mu).
\end{align}
Since $\Pi_b$ and $\mathcal S$ are nonempty,
compact, and convex, Sion's minimax theorem gives
\begin{align}
\min_{\pi_b\in\Pi_b}\max_{\mu\in\mathcal S}
\mathcal G(\pi_r(\cdot),\pi_b,\mu)
=
\max_{\mu\in\mathcal S}\min_{\pi_b\in\Pi_b}
\mathcal G(\pi_r(\cdot),\pi_b,\mu)
=:v(\pi_r(\cdot)).
\end{align}

Define the Blue and Nature optimal-strategy
correspondences by
\begin{align}
\operatorname{BR}^{(b)}_{\mathrm{PRMRE}}
(\pi_r(\cdot);\bar\mu,\delta)
&:=
\operatorname*{arg\,min}_{\pi_b\in\Pi_b}
\max_{\mu\in\mathcal S}
\mathcal G(\pi_r(\cdot),\pi_b,\mu),
\\
\operatorname{BR}^{(n)}_{\mathrm{PRMRE}}
(\pi_r(\cdot);\bar\mu,\delta)
&:=
\operatorname*{arg\,max}_{\mu\in\mathcal S}
\min_{\pi_b\in\Pi_b}
\mathcal G(\pi_r(\cdot),\pi_b,\mu).
\end{align}
The first correspondence coincides with the
PR-MRE Blue best-response operator defined in
Definition~\ref{def:prmre}. The second characterizes Nature's
optimal strategy in the associated Blue--Nature
zero-sum regret game.

Both extrema are attained by continuity and
compactness. The first objective is continuous
and convex in $\pi_b$, while the second is
continuous and concave in $\mu$. Therefore,
both correspondences have nonempty, compact,
and convex values.

Define the Blue--Nature saddle-point
correspondence
\begin{align}
\mathcal H(\pi_r(\cdot))
:={}&
\operatorname{BR}^{(b)}_{\mathrm{PRMRE}}
(\pi_r(\cdot);\bar\mu,\delta)
\nonumber\\
&\times
\operatorname{BR}^{(n)}_{\mathrm{PRMRE}}
(\pi_r(\cdot);\bar\mu,\delta).
\end{align}
By the minimax equality, every
$(\pi_b,\mu)\in\mathcal H(\pi_r(\cdot))$
satisfies
\begin{align}
\mathcal G(\pi_r(\cdot),\pi_b,\mu')
&\leq
\mathcal G(\pi_r(\cdot),\pi_b,\mu)
\nonumber\\
&\leq
\mathcal G(\pi_r(\cdot),\pi_b',\mu),
\label{eq:prmre_existence_saddle}
\end{align}
for all $\pi_b'\in\Pi_b$ and $\mu'\in\mathcal S$.
Conversely, these saddle-point inequalities
imply optimality in both minimax problems.
Thus $\mathcal H(\pi_r(\cdot))$ is precisely
the Blue--Nature saddle-point set and has
nonempty, compact, convex values.

We next establish that $\mathcal H$ has a
closed graph. Suppose
$\pi_r^k(\cdot)\to\pi_r(\cdot)$ and
$(\pi_b^k,\mu^k)\to(\pi_b,\mu)$, with
$(\pi_b^k,\mu^k)\in
\mathcal H(\pi_r^k(\cdot))$.
For every $\widetilde\pi_b\in\Pi_b$ and
$\widetilde\mu\in\mathcal S$,
\begin{align}
\mathcal G(\pi_r^k(\cdot),\pi_b^k,\widetilde\mu)
&\leq
\mathcal G(\pi_r^k(\cdot),\pi_b^k,\mu^k)
\nonumber\\
&\leq
\mathcal G(\pi_r^k(\cdot),\widetilde\pi_b,\mu^k).
\end{align}
Passing to the limit by joint continuity
establishes the saddle-point inequalities
at $(\pi_r(\cdot),\pi_b,\mu)$.
Hence $(\pi_b,\mu)\in\mathcal H(\pi_r(\cdot))$,
and $\mathcal H$ has a closed graph.

Next, consider the informed Red team's
type-conditioned best-response correspondence
$\operatorname{BR}^{(r)}(\pi_b;\cdot)$
from Definition~\ref{assumption:red_team_rationality}, interpreted as the product
of the typewise best-response sets:
\begin{align}
\operatorname{BR}^{(r)}(\pi_b;\cdot)
:=
\prod_{\theta\in\Theta}
\operatorname*{arg\,min}_{\pi_r(\theta)
\in\Pi_r(\theta)}
V(\pi_r(\theta),\pi_b;\theta).
\end{align}
Continuity of $V$, compactness of
$\Pi_r(\theta)$, and convexity of $V$
in $\pi_r(\theta)$ imply that each
type-conditioned best-response set is
nonempty, compact, and convex.
Its graph is closed by continuity of $V$.
Since $\Theta$ is finite,
$\operatorname{BR}^{(r)}(\pi_b;\cdot)$
also has nonempty, compact, convex values
and a closed graph.

Now let
\begin{align}
\mathcal K
&:=\Pi_b\times\mathcal S\times\Pi_r,
\\
\mathcal F(\pi_b,\mu,\pi_r(\cdot))
&:=
\mathcal H(\pi_r(\cdot))
\times
\operatorname{BR}^{(r)}(\pi_b;\cdot).
\end{align}
The set $\mathcal K$ is nonempty, compact,
and convex. The correspondence
$\mathcal F:\mathcal K\rightrightarrows
\mathcal K$ has nonempty, compact, convex
values and a closed graph.
Since its range lies in the compact set
$\mathcal K$, it is upper hemicontinuous.

By Kakutani's fixed-point theorem, there
exists
$(\pi_b^*,\mu^*,\pi_r^*(\cdot))\in\mathcal K$
such that
\begin{align}
(\pi_b^*,\mu^*)
&\in\mathcal H(\pi_r^*(\cdot)),
\\
\pi_r^*(\cdot)
&\in\operatorname{BR}^{(r)}(\pi_b^*;\cdot).
\end{align}
By the definition of $\mathcal H$,
\begin{align}
\pi_b^*
&\in
\operatorname{BR}^{(b)}_{\mathrm{PRMRE}}
(\pi_r^*(\cdot);\bar\mu,\delta),
\\
\mu^*
&\in
\operatorname{BR}^{(n)}_{\mathrm{PRMRE}}
(\pi_r^*(\cdot);\bar\mu,\delta).
\end{align}
In particular, $\pi_b^*$ minimizes the
worst-case expected regret over
$\mathcal S(\bar\mu,\delta)$ against
$\pi_r^*(\cdot)$, while
$\pi_r^*(\cdot)$ is a type-conditioned
best response to $\pi_b^*$.

Consequently,
\begin{align}
\pi_b^*
&\in
\operatorname{BR}^{(b)}_{\mathrm{PRMRE}}
(\pi_r^*(\cdot);\bar\mu,\delta),
\\
\pi_r^*(\cdot)
&\in
\operatorname{BR}^{(r)}(\pi_b^*;\cdot).
\end{align}
These are precisely the PR-MRE fixed-point
conditions in Eq.~(\ref{eq:prmre_fixed_point}),
proving existence.
\end{proof}
}
\subsection{Robustness Certificates for MRE and PR-MRE Best-response Operators}
\label{sec:appendix_proofs}
\begin{proof}[\textbf{Proof of Lemma~\ref{lemma:ulb_mre}}]
\begin{align}\label{eq:mre_perf_lb_proof}
    \mathbb{E}_{\theta \sim \mu}[V(\pi_{r}(\cdot), \pi_{b}; \theta] &= \sum_{\theta \in \Theta} \mu(\theta) V(\pi_{r}(\cdot), \pi_{b}; \theta) = \sum_{\theta \in \Theta} \mu(\theta) \left[ \sup_{\pi^{*} \in \Pi_{b}} V(\pi_{r}(\cdot), \pi^{*}; \theta) - \mathcal{E}(\pi_{r}(\cdot), \pi_{b}; \theta)  \right], \\
    &= \sum_{\theta \in \Theta} \mu(\theta)  \sup_{\pi^{*} \in \Pi_{b}} V(\pi_{r}(\cdot), \pi^{*}; \theta) - \sum_{\theta \in \Theta} \mu(\theta) \mathcal{E}(\pi_{r}(\cdot), \pi_{b}; \theta), \\ 
    &\geq \sum_{\theta \in \Theta} \mu(\theta)  \sup_{\pi^{*} \in \Pi_{b}} V(\pi_{r}(\cdot), \pi^{*}; \theta) - \sup\limits_{q  \in \Delta(\lvert\Theta\rvert)} \underbrace{\left(\sum_{\theta \in \Theta} q(\theta) \mathcal{E}(\pi_{r}(\cdot), \pi_{b}; \theta) \right)}_{ \bar{\mathcal{E}}(\pi_{b}, q)}, \\
    &= \underbrace{\sum_{\theta \in \Theta} \mu(\theta)  \sup_{\pi^{*} \in \Pi_{b}} V(\pi_{r}(\cdot), \pi^{*}; \theta) - \sup\limits_{\theta \in \Theta} \mathcal{E}(\pi_{r}(\cdot), \pi_{b}; \theta)}_{\mathcal{V}^{(-)}(\pi_{b}, \mu)}.
\end{align}
Therefore, $\mathbb{E}_{\theta \sim \mu}[V(\pi_{r}(\cdot), \pi_{b}; \theta] \geq \mathcal{V}^{(-)}(\pi_{b}, \mu) =  \sum\limits_{\theta \in \Theta} \mu(\theta) \sup\limits_{\pi^{*} \in \Pi_{\mathrm{b}}} V(\pi_{r}, \pi^{*}; \theta) - \mathcal{E}(\pi_{r}(\cdot), \pi_{b}, \Theta)$ for all $\pi_{b} \in \Pi_{b}$ and $\mu(\cdot) \in \Delta(\lvert\Theta\rvert)$. Moreover, the lower bound is tight such that equality holds for $q = e_{\operatorname{argsup}\limits_{\theta \in \Theta} \mathcal{E}(\pi_{r}(\cdot), \pi_{b}; \theta)}$ where $e_{k}$ is the $k$-th canonical vector in $\mathbb{R}^{\lvert\Theta\rvert}$.
\end{proof}
\begin{proof}[\textbf{Proof of Lemma~\ref{lemma:prmre_rob_property}}]
We recall that,
\begin{align}
\mathcal{V}_{\mathcal{S}}^{(-)}( \pi, \mu ) = \sum\limits_{\theta \in \Theta} \mu(\theta) \sup\limits_{\pi^{*} \in \Pi_{\mathrm{b}}} V(\pi_{r}, \pi^{*}; \theta) - \sup\limits_{q \in \mathcal{S}(\bar{\mu}, \delta)} \bar{\mathcal{E}}(\pi, q),
\end{align}
where $\bar{\mathcal{E}}(\pi, q) = \sum\limits_{\theta \in \Theta} q(\theta) \mathcal{E}(\pi_{r}(\cdot), \pi; \theta)$. Now,
    \begin{align}
    \mathcal{V}_{\mathcal{S}}^{(-)}( \pi_{\mathrm{PRMRE}}, \mu ) &= \sum\limits_{\theta \in \Theta} \mu(\theta) \sup\limits_{\pi^{*} \in \Pi_{\mathrm{b}}} V(\pi_{r}, \pi^{*}; \theta) - \sup\limits_{q \in \mathcal{S}(\bar{\mu}, \delta)} \bar{\mathcal{E}}(\pi_{\mathrm{PRMRE}}, q), \\
    &= \sum\limits_{\theta \in \Theta} \mu(\theta) \sup\limits_{\pi^{*} \in \Pi_{\mathrm{b}}} V(\pi_{r}, \pi^{*}; \theta) - \inf\limits_{\pi' \in \Pi_{b}}\sup\limits_{q \in \mathcal{S}(\bar{\mu}, \delta)} \bar{\mathcal{E}}(\pi' , q), \\
    &\geq \sum\limits_{\theta \in \Theta} \mu(\theta) \sup\limits_{\pi^{*} \in \Pi_{\mathrm{b}}} V(\pi_{r}, \pi^{*}; \theta) - \sup\limits_{q \in \mathcal{S}(\bar{\mu}, \delta)} \bar{\mathcal{E}}(\pi'' , q), \\ &= \mathcal{V}_{\mathcal{S}}^{(-)}(\pi'', \mu),
    \end{align}
for all $ \pi'' \in \Pi_{b}$. The same holds for $\pi_{\mathrm{MRE}} \in \arginf\limits_{\pi \in \Pi_{b}} \sup\limits_{\theta \in \Theta} \mathcal{E}(\pi_{r}(\cdot), \pi; \theta)$ such that $\mathcal{V}_{\mathcal{S}}^{(-)}(\pi_{\mathrm{PRMRE}}, \mu) \geq \mathcal{V}_{\mathcal{S}}^{(-)}(\pi_{\mathrm{MRE}}, \mu)$. Hence proved that $\pi_{\mathrm{PRMRE}}$ provides the tightest {\color{black} point-wise benchmark-relative performance certificate} for type distributions restricted to the threat model.
\end{proof}

\subsection{Rank-constrained Semidefinite Reformulation and Convex Relaxation}
\label{sec:prmre_sdp_relax}
In this section, we lift the robust bilinear program \ref{eq:prmre_eqm_normal_form} to a semidefinite program (SDP) by reformulating the constraints via an introduced matrix positive semidefinite (psd) variable. Consider a symmetric psd matrix $Z(\theta) \in \mathbb{S}^{+}_{\lvert\mathcal{X}\rvert+\lvert\mathcal{Y}\rvert+1}$ such that,
\begin{align}\label{eq:Z_theta_def}
    Z(\theta) = \left[\begin{array}{lll} 1 & x(\theta)^{\top} & y^{\top} \\ x(\theta) & X(\theta) & W(\theta)^{\top} \\ y & W(\theta) & Y\end{array}\right] \succeq 0,
\end{align}
with the first row and column equal to the concatenated $x$ and $y$ vectors lying in the respective probability simplices and block matrices $X$ and $Y$ at the diagonals and an off-diagonal block matrix $W$. We follow a zero-based indexing convention to refer to the rows and columns of $Z(\theta)$ such that $Z_{1,1}(\theta) = X(\theta)$, $Z_{2,1}(\theta) = W$ and so on. Consider a reformulation of constraints (\ref{eq:y_switch})--(\ref{eq:x_switch}) with respect to the block matrices of the psd matrix $Z(\theta)$,
\begin{align}
    x(\theta)A(\theta)e_{j} - \langle A(\theta), Z_{2,1}(\theta)  \rangle \leq \{ \lambda^{\top} \mathcal{A}_{\mathscr{L}} &- \lambda^{\top}_{\mathrm{p}} + \lambda_{\mathrm{u}}\mathbf{1} \}_{\theta} \text{ for all } j \in \left[\lvert\mathcal{Y}\rvert\right] \text{ and } \theta \in \Theta,
\end{align}
and $- \langle A(\theta), Z_{2,1}(\theta)  \rangle \geq -e_{k} A(\theta) y \text{ for all } k \in [\lvert\mathcal{X}\rvert] \text{ and } \theta \in \Theta$ where $\langle\cdot,\cdot\rangle$ is the trace operator such that $\langle A, B \rangle = \operatorname{trace}(AB)$.

The off-diagonal block matrix $Z_{2,1}(\theta) = W(\theta)$ has the interpretation of capturing the cross-terms between the $x(\theta)$ and $y$ vectors such that $W(\theta) \approx yx(\theta)^{\top}$. This enables us to lift the bilinear term $x(\theta)^{\top} A(\theta) y$ and express it in terms of a block matrix of a positive semidefinite matrix variable as, $x(\theta)^{\top} A(\theta) y = \langle A(\theta), yx(\theta)^{\top} \rangle \approx \langle A(\theta), W(\theta) \rangle$. This is an instance of lifting a polynomial program into an SDP by the introduction of psd matrix variables (see \cite{moment_sos1, moment_sos2} for a survey on Moment-SoS hierarchy for polynomial optimization). {\color{black} Any rank-1 solution $\{Z(\theta_{i})\}$ to the lifted program returns an exact solution of the non-convex SDP and hence recovers an exact solution $\{ \mathbf{x}(\cdot), y \}$ of the robust bilinear program \ref{eq:prmre_eqm_normal_form} as noted in Proposition~\ref{prop:rank_one_recovery_normal_form_games}.} 

\textbf{Remark on Asymmetric Information between Red and Blue Teams}. Note that $Z_{2,0}(\theta) = y \text{ for all } \ \theta \in \Theta$ from the definition (\ref{eq:Z_theta_def}) of $Z(\theta)$. This is the robust Blue decision $y$ under asymmetric information whereas Red can condition its play on the hidden type as reflected by $x(\theta)$ in $Z(\theta)$. Therefore, psd matrices $Z(\theta)$ are constrained to be consistent in their $y$ decision for all types $\theta \in \Theta$ and this is added as an additional consistency constraint in the lifted program.

\textbf{Convex Relaxation}.
Rank-constrained SDPs are an important class of non-convex problems that are generally NP-hard to solve \cite{lowrank_sdp_recht_parillo}. The rank constraint appears because of the cross-terms in the bilinear program that was lifted via the use of matrix PSD variables. We relax the rank constraint into a convex semidefinite program for our application using McCormick cuts.
\textit{McCormick Cuts and Valid Inequalities}: We begin our exposition by stating the following proposition that characterizes the desired rank-one correlations between the lifted matrix variables and the original decision variables.
\begin{proposition}\label{prop:cross_term}
    If $\operatorname{rank}(Z(\theta)) = 1$, $X(\theta) = x(\theta)x^{\top}(\theta)$, $Y = yy^{\top}$ and $W(\theta) = yx^{\top}(\theta)$.
\end{proposition}
Since we relax the rank-constraint, we add {valid inequalities} in the form of McCormick cuts to the block matrices $X(\theta)$, $Y$ and $W(\theta)$ to encode the correlation between the cross-terms as stated in Proposition~\ref{prop:cross_term}. This is a valid approach since all feasible $\operatorname{rank}$-$1$ solutions to the rank-constrained program satisfy these inequalities and is a way of pruning out high-rank solutions that do not obey the correlational structure possessed by low-rank solutions. For example, we add the following set of inequalities for $W_{i,j}(\theta)$, since it is supposed to be a proxy for $y_{i}x_{j}$ such that $W_{i,j} \approx y_{i}x_{j}$,
\begin{align}
    W_{i,j}(\theta) &\leq x_{j}(\theta), W_{i,j}(\theta) \leq y_{i}, \\
    W_{i,j}(\theta) &\geq x_{j}(\theta) + y_{i} - 1
\end{align}
for all $i \in [\lvert m \rvert]$ and $j \in [\lvert n \rvert]$. We add similar cuts for the $i,j$th terms ($i \neq j$) of the block-matrices $X(\theta)$ and $Y$ which are supposed to be proxies for $x_{i}x_{j}$ and $y_{i}y_{j}$ respectively. For the diagonal terms of the $X(\theta)$ and $Y$ blocks, we enforce constraints of the form $0 \leq X_{ii}(\theta) \leq x_{i}(\theta)$ and $0 \leq Y_{ii} \leq y_{i}(\theta)$ which follow from $x(\theta)$ and $y$ belonging to probability simplices.

The solution $\{Z(\theta_{i})\}$ obtained from the relaxed program via the valid inequalities provided by McCormick cuts is not guaranteed to be $\operatorname{rank}$-1 and the relaxation might not be lossless. Therefore, as part of our method, we project the obtained solution for $Z(\theta)$'s on the manifold of $\operatorname{rank}$-1 matrices. From the classical Eckart–Young–Mirsky theorem, the optimal $\operatorname{rank}$-$r$ approximant of a $m\times n$ rectangular matrix $M$ is given by the largest $r$ singular values of $M$ and the corresponding eigenvectors. In our method, we obtain $\operatorname{rank}$-1 approximations for $\{Z(\theta_{i})\}$ via the largest eigenvector for each $\theta_{i}$. Let $v_{\mathrm{max}}(Z(\theta_{i})) \in \mathbb{R}^{\lvert\mathcal{X}\rvert+\lvert\mathcal{Y}\rvert+1}$ denote the eigenvector corresponding to the largest singular value of $Z(\theta_{i})$. We re-normalize $v_{\mathrm{max}}(\theta_{i})$ such that the first-entry becomes unity as $v_{\mathrm{max}}(\theta_{i}) = \left[ 1, x(\theta_{i}), y\right]^{\top}$ and obtain $x(\theta_{i})$ and $y$. Note that the $x(\theta)$ and $y$ obtained in this fashion might not lie in their respective probability simplices (due to negative entries for example) and will require an additional projection step. In our implementation, we clip the negative entries at $0$ and re-normalize such that the sum of all entries equals $1$.

\subsection{Additional Illustrative Example}
\label{sec:illus_example_appendix}
\todo{penalty as a mild punisher row; preserves specialist and generalist structure but changes the regret and payoff landscape across types.}
Consider the following $2\times5\times3$ normal-form game such that rows corresponding to red action $r_{0}$ are the same as the $1\times5\times3$ game from Table~\ref{tab:1row-game}. Red action $r_{1}$ is an additional \textit{punisher} action that punishes Blue specialists $b_{0}$, $b_{1}$ and $b_{2}$ on their specialized types respectively and also the generalists $b_{3}$, $b_{4}$ across all types hence modifying the regret and payoff landscape for the game. We note that $b_{0}$, $b_{1}$ and $b_{2}$ retain their specialist roles for $\theta_{0}$, $\theta_{1}$ and $\theta_{2}$ respectively as before as can be verified from the payoff values in the game below. Unlike the main-text example, the informed Red player now has multiple strategic choices, so that both the robust Blue best-responses and the corresponding informed Red best-responses differ across equilibrium fixed-points.
\begin{table}[h]
\caption{Example $2{\times}5{\times}3$ game. Columns $\{b_{0} \cdots b_{4}\}$ are Blue actions, rows $\{r_{0}, r_{1}\}$ correspond to Red actions and a total of three $2 \times 5$ matrices stacked one on top of another along a third type-axis for $\theta \in \{\theta_{0}, \theta_{1}, \theta_{2}\}$ with a nominal distribution $\bar{\mu} = (0.25, 0.70, 0.05)$ such that each cell shows the Blue payoff $V(r, b; \theta)$ for the $(r, b, \theta)$ matchup.
}
\label{tab:2row-game}
\centering
\[
\setlength{\arraycolsep}{6pt}
\renewcommand{\arraystretch}{1.3}
\begin{array}{r |c|c|c|c|c| l}
 \multicolumn{1}{r}{}
   & \multicolumn{1}{c}{b_0} & \multicolumn{1}{c}{b_1} & \multicolumn{1}{c}{b_2}
   & \multicolumn{1}{c}{b_3} & \multicolumn{1}{c}{b_4} & \multicolumn{1}{l}{} \\[4pt]
\cline{2-6}
 r_0 &  1.5 &  0.6 & -0.5 &  0.4 &  0.7 & \multirow{2}{*}{$\left.\rule{0pt}{4.6ex}\right\}\theta_0, 0.25$} \\
\cline{2-6}
 r_1 &  0.7 &  0.1 & -0.5 &  0.1 &  0.1 & \\
\cline{2-6}
\noalign{\vskip 7pt}
\cline{2-6}
 r_0 &  0.3 &  1.0 & -0.8 &  0.8 &  0.8 & \multirow{2}{*}{$\left.\rule{0pt}{4.6ex}\right\}\theta_1, 0.70$} \\
\cline{2-6}
 r_1 & -0.1 &  0.3 & -0.8 &  0.2 & -0.1 & \\
\cline{2-6}
\noalign{\vskip 7pt}
\cline{2-6}
 r_0 & -0.6 & -0.8 &  1.5 &  0.5 & -0.1 & \multirow{2}{*}{$\left.\rule{0pt}{4.6ex}\right\}\theta_2, 0.05$} \\
\cline{2-6}
 r_1 & -0.4 & -0.8 &  0.5 &  0.0 & -0.3 & \\
\cline{2-6}
\end{array}
\]
\end{table}

$\operatorname{BNE}$, $\operatorname{DRO-Ball}$, $\operatorname{MaxMin}$ and $\operatorname{DRO-MaxMin}$ are all computed via linear-programs as before whereas $\operatorname{MRE}$ and $\operatorname{PR-MRE}$ ($\delta = 0.15$) are computed via semidefinite relaxations to the corresponding bilinear and robust bilinear programs as derived in Section~\ref{sec:computation}. The evaluation protocol for the winner-map, payoff difference heatmaps and the metrics table is the same as before such that the computed Blue decisions are evaluated against the Red $\operatorname{BNE}$ corresponding to the test type distribution that probes over the entire type-simplex. The robust Blue decisions corresponding to the above computed fixed-points are: BNE $(0, 1, 0, 0, 0)$, DRO-Ball $(0, 0, 0, 1, 0)$, MaxMin $(0.05, 0, 0.14, 0.81, 0)$, DRO-MaxMin $(0.11, 0, 0, 0.89, 0)$, MRE $(0.42, 0, 0.29, 0.29, 0)$ and  PR-MRE $(0.56, 0, 0, 0.45, 0)$. We note that no Blue decision has $b_{4}$ in its support including $\operatorname{PR-MRE}$ in contrast to the $1\times5\times3$ game from before. This is due to the additional punisher row that penalizes $b_{4}$ more drastically than $b_{3}$ over typical types $\theta_{0}$ and $\theta_{1}$ as seen in Table~\ref{tab:2row-game}. The computed $\operatorname{MRE}$ Blue decision allocates meaningful probability mass on the rare-type specialist $b_{2}$ same as before, at the expense of catastrophic performance on the typical set $\mu(\theta_{2}) \leq 0.15$ (-0.222 worst-case and -0.022 average payoff from Figure~\ref{fig:2row_tables}), as opposed to the computed $\operatorname{PR-MRE}$ Blue that discounts regret incurred on the rare-type and hedges across the trapezoidal band representing the typical subset of the type-simplex under confidence-level $\delta = 0.15$ (see Figure~\ref{fig:2row_triangles}). 

{\color{black} \textbf{Equilibrium residual of the recovered candidate PR-MRE strategy}. As discussed before, the program \ref{eq:prmre_eqm_normal_form} does not return exact $\operatorname{MRE}$ and $\operatorname{PR-MRE}$ fixed-points and all the results discussed here are for candidate strategy profiles recovered from a convex relaxation of a lifted rank-constrained semidefinite reformulation of the surrogate program. We characterize the gap of the recovered strategy from a $\operatorname{PR-MRE}$ fixed-point by computing the equilibrium residual. For the example game in Table~\ref{tab:2row-game}, the recovered semidefinite matrix from our convex relaxation resides in $\mathcal{S}^{+}_{8}$ (lifted space, $\mathcal{S}^{+}_{\lvert\mathcal{X}\rvert + \lvert\mathcal{Y}\rvert + 1}$ psd cone) is rank-$3$. The three leading eigenvalues are approximately $(2.5469,0.15565,0.14649)$, with the remaining four below $7\times10^{-8}$. The computed Red residual $\eps_{r}$ is atmost $7.5\times10^{-9}$ across the three-types. To compute the Blue residual, we fix the extracted Red strategy and solve the PR-MRE robust best-response problem exactly (see Proposition~\ref{prop:prmre_br_nfg_exact} that states that the formulated program returns exact robust best-responses). The extracted Blue strategy has worst-case expected regret $0.335429$, versus the exact robust best-response value $0.335000$, giving a residual $\eps_{b}$ of $4.3\times10^{-4}$. Thus, although the relaxed SDP solution is not rank one, the extracted strategy satisfy the PR-MRE fixed-point condition approximately, with a residual of $\operatorname{max}(\eps_{r}, \eps_{b}) = 4.7\times10^{-4}$.
}
 
\textbf{Results}. From the winner-map in Figure~\ref{fig:2row_triangles}, $\operatorname{PR-MRE}$ dominates the trapezoidal band and has the lowest worst-case regret and the highest average-case payoff (0.200) for type distributions chosen at random from the structured type ambiguity set as opposed to $\operatorname{DRO-MaxMin}$ which has a higher worst-case payoff (0.135) as compared $\operatorname{PR-MRE}$ (-0.005) but loses out in the average-case. The regret-based objective of $\operatorname{PR-MRE}$ provides a performance guarantee in the form of a point-wise uniform lower bound (see Lemma~\ref{lemma:prmre_rob_property}) for all type distributions belonging to the structured ambiguity set that results in a favorable average-case performance over the trapezoidal band in the example game.
\begin{figure}[!h]
  \centering
  \includegraphics[scale=0.25]{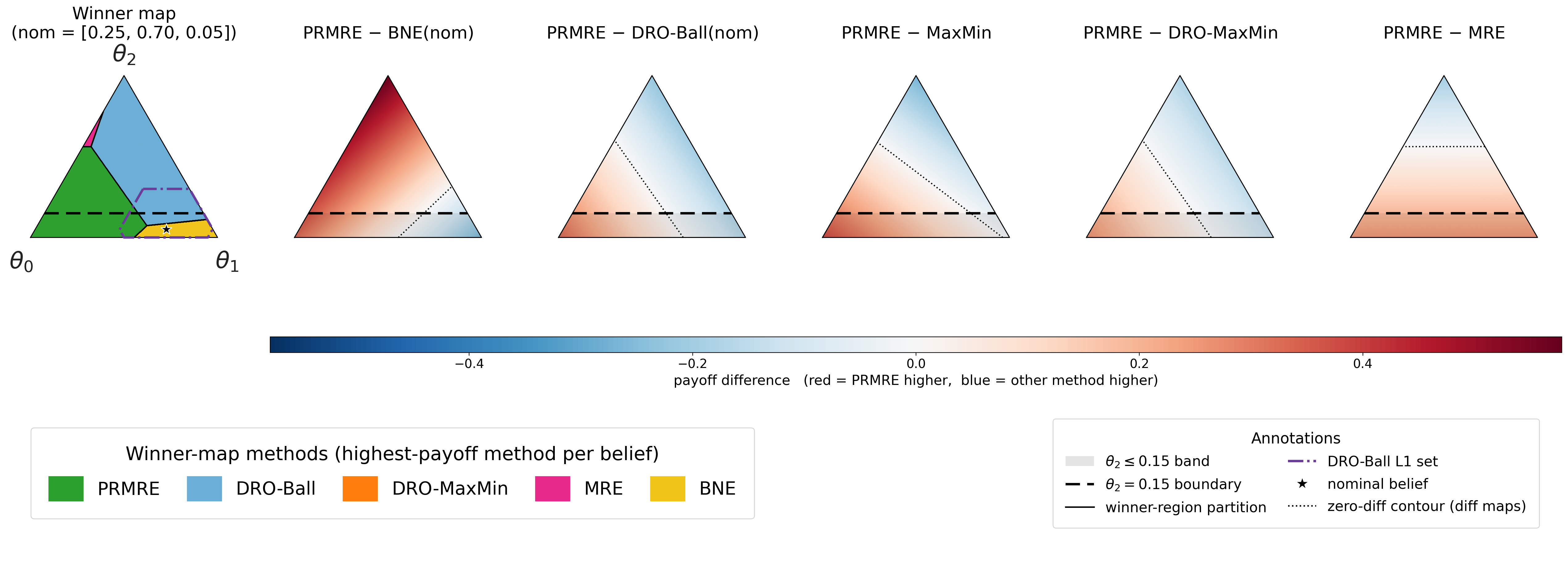}
  \vspace{-1.5em}
  \caption{Comparison of $\operatorname{PRMRE}$ ($\delta = 0.15$) with $\operatorname{BNE}$, $\operatorname{DRO-Ball}$, $\operatorname{MaxMin}$, $\operatorname{DRO-MaxMin}$ and $\operatorname{MRE}$  for the synthetic three-type game in Table~\ref{tab:2row-game}: winner map (first triangle) and payoff difference heatmaps (remaining triangles). Triangles represent three-dimensional simplex in barycentric coordinates with vertices corresponding to full probability mass on the respective types. Black dotted line and the trapezoidal band beneath it represent the high-confidence under nominal structured type-ambiguity set $\mu(\theta_{2}) \leq 0.15$, the star in the winner map corresponds to the nominal distribution $\bar{\mu} = (0.25, 0.70, 0.05)$ and the dotted hexagon centered at the star corresponds to an $L1$-ball with radius $\rho=0.5$ around it. For the trapezoidal band, $\operatorname{PR-MRE}$ dominates all other methods which can be seen from the green volume share in the winner-map and the red volume share in the payoff difference heatmaps. $\operatorname{BNE}$ and $\operatorname{DRO-Ball}$ perform better close to the nominal ($\theta_{1}$-corner, mass 0.70 under nominal) but lose out on performance to $\operatorname{PR-MRE}$ for a type distribution chosen at random within the trapezoidal band. $\operatorname{PR-MRE}$ has the best average}
  \label{fig:2row_triangles}
\end{figure}
\begin{figure}[!ht]
  \centering
  \includegraphics[scale=0.3]{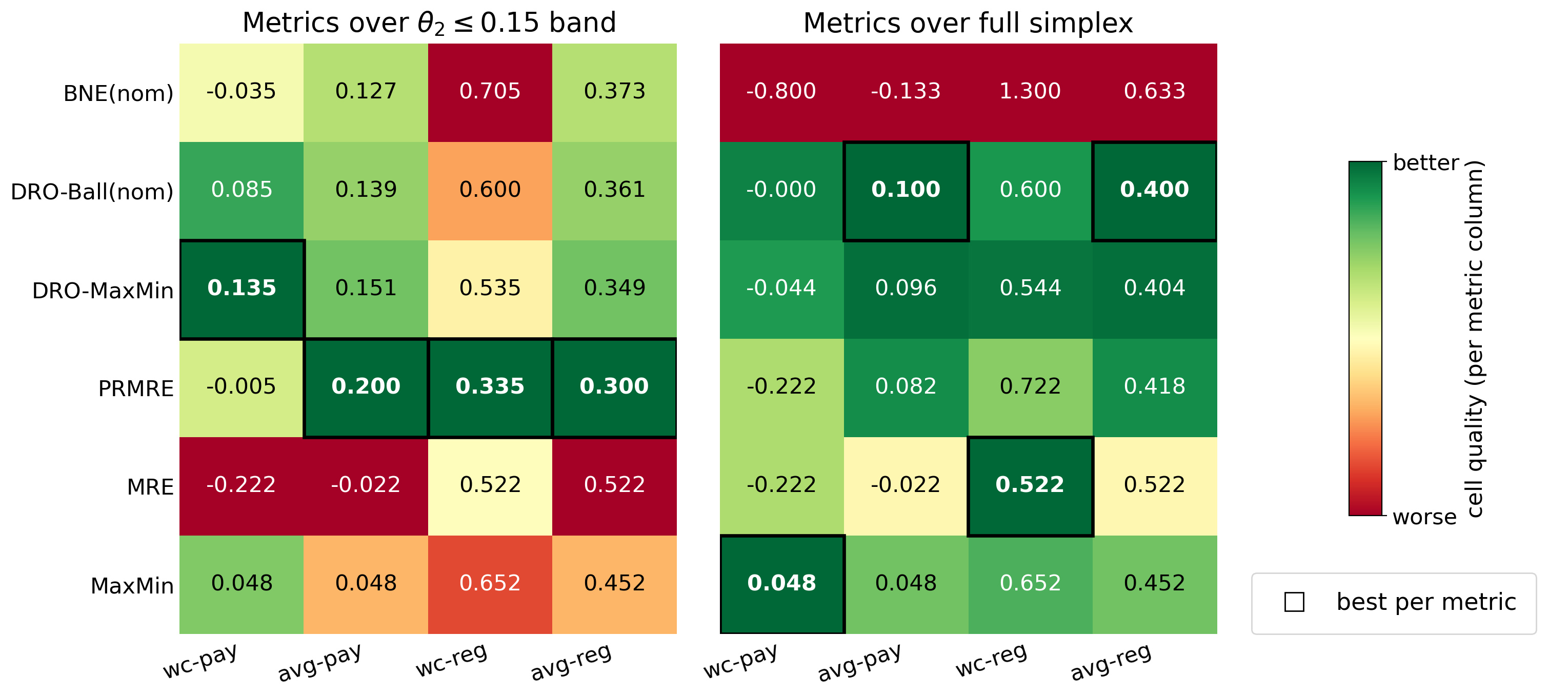}
  \caption{Comparison of $\operatorname{PRMRE}$ ($\delta = 0.15$) with $\operatorname{BNE}$, $\operatorname{DRO-Ball}$, $\operatorname{MaxMin}$, $\operatorname{DRO-MaxMin}$ and $\operatorname{MRE}$  for the synthetic three-type game in Table~\ref{tab:2row-game} with respect to average and worst-case, payoff and regret over the structured ambiguity set $\mu(\theta_{2}) \leq 0.15$ and the entire type-simplex. Over the $\mu(\theta_{2}) \leq 0.15$ trapezoidal band (see Figure~\ref{fig:2row_triangles}), $\operatorname{PR-MRE}$ has the best worst-case regret as compared to any other method. $\operatorname{DRO-MaxMin}$ on the other hand beats $\operatorname{PR-MRE}$ on worst-case payoff over the $\mu(\theta_{2}) \leq 0.15$ band which is consistent with its payoff-based objective (0.135 vs -0.005) but loses average-case performance for a type distribution chosen at random from the trapezoidal band such that $\operatorname{PR-MRE}$ attains the highest average payoff (0.200) as compared to any other method (0.151 for $\operatorname{DRO-MaxMin}$). The regret-based objective of the $\operatorname{PR-MRE}$ best-response provides the best point-wise lower bound on payoff for type distributions belonging to the structured ambiguity set (as stated in Lemma~\ref{lemma:prmre_rob_property}) which in this example game manifests as a superior average-case performance for a type-distribution chosen at random from the $\mu(\theta_{2}) \leq 0.15$ structured ambiguity set.}
  \label{fig:2row_tables}
\end{figure}

\subsection{GNN-GRU-MAPPO Architecture and Training Details}
\label{sec:psro_details_appendix}
\begin{table}[t]
\centering
\caption{{\color{black}Observation of one Blue agent (\texttt{obs\_version}~3). Every agent sees only its own 3-hop ego-subgraph of the 103-node graph, zero-padded to $N_{\max}=24$ nodes and $E_{\max}=70$ edges so the tensor shape is fixed; the visibility masks say which slots are real. Maximum degree is $d_{\max}=4$, giving $d_{\max}+2=6$ actions (move to a neighbour, stay, tag).}}
\small
\begin{tabular}{@{}lll@{}}
\toprule
Key & Shape & Contents \\
\midrule
\texttt{x} & $24 \times 17$ & node features of the ego-subgraph (padded) \\
\texttt{edge\_index} & $2 \times 70$ & ego-subgraph edges (padded) \\
\texttt{edge\_attr} & $70 \times 1$ & edge features --- allocated but all zero; the MPNN discards them \\
\texttt{node\_visibility\_mask} & $24$ & which of the $N_{\max}$ node slots are real \\
\texttt{edge\_visibility\_mask} & $70$ & which of the $E_{\max}$ edge slots are real \\
\texttt{num\_visible\_nodes} & $1$ & count of real nodes \\
\texttt{num\_visible\_edges} & $1$ & count of real edges \\
\texttt{agent\_node\_local\_idx} & $1$ & index of the ego node within the subgraph \\
\texttt{agent\_node\_mask} & $103$ & one-hot of the ego node over the FULL graph \\
\texttt{neighbor\_local\_idx} & $4$ & subgraph indices of the ego node's neighbours \\
\texttt{neighbor\_mask} & $4$ & which neighbour slots are real \\
\texttt{action\_mask} & $6$ & legal actions: $d_{\max}$ moves, stay, tag \\
\bottomrule
\end{tabular}
\label{tab:obs-keys}
\end{table}

\begin{table}[t]
\caption{{\color{black}The $F=17$ columns of \texttt{x}. ``Ego-only'' features are written on the agent's own row and left at zero elsewhere. Distances are in hops, normalised by the graph diameter. Feature~16 is the agent-index column: it is what lets one shared parameter set imitate role-dependent scripted policies.}}
\centering
\small
\begin{tabular}{@{}llll@{}}
\toprule
\# & Feature & Scope & Meaning \\
\midrule
0 & \texttt{is\_ego} & per-node & this is the agent's own node \\
1 & \texttt{is\_frontier} & per-node & node resolves a flag hypothesis \\
2 & \texttt{frontier\_resolution\_mass} & per-node & how much belief it would resolve \\
3 & \texttt{is\_regular} & per-node & neither ego nor frontier \\
4 & \texttt{is\_flag\_discovered\_here} & per-node & true flag seen at this node \\
5 & \texttt{distance\_to\_ego} & per-node & hops to the agent \\
6 & \texttt{distance\_to\_home\_flag} & per-node & hops to own flag \\
7 & \texttt{distance\_to\_away\_flag} & per-node & hops to enemy flag; $0$ while unknown \\
8 & \texttt{enemy\_flag\_known} & ego-only & hypothesis resolved yet \\
9 & \texttt{distance\_to\_frontier} & per-node & hops to nearest frontier node \\
10 & \texttt{has\_opponent} & per-node & an opponent occupies this node \\
11 & \texttt{has\_teammate} & per-node & a teammate occupies this node \\
12 & \texttt{min\_opp\_distance} & ego-only & hops to nearest opponent \\
13 & \texttt{min\_teammate\_distance} & ego-only & hops to nearest teammate \\
14 & \texttt{structural\_degree} & per-node & node degree in the full graph \\
15 & \texttt{visibility\_bit} & per-node & real node vs.\ padding \\
16 & \texttt{agent\_index} & ego-only & agent id in $[0,1]$; breaks the role tie \\
\bottomrule
\end{tabular}
\label{tab:obs-features}
\end{table}
{\color{black}
\paragraph{Architecture.}
Each agent encodes its ego-subgraph with a three-layer message-passing network: $17\!\to\!64\!\to\!64\!\to\!64$, sum aggregation with self-loops, message $\mathrm{MLP}([h_i,h_j])$ and a \texttt{ReLU} after each layer. Edge features are allocated in the observation but are identically zero and the message function ignores them, so propagation is purely topological. A single \texttt{GRUCell} ($64$ hidden units) is then applied to the ego-node embedding $h_{\mathrm{agent}}$, carrying belief about the flag hypothesis across the $16$-step episode; it is stateful during rollout and re-initialised to zeros at each gradient update, so back-propagation through time spans one full episode without truncation. Moves are scored relationally --- the logit for stepping to neighbour $k$ is $\mathrm{MLP}([h_{\mathrm{agent}},h_{\mathrm{neighbour}_k}])$, masked to the real neighbours --- and two further logits for \emph{stay} and \emph{tag} are read off $h_{\mathrm{agent}}$ alone. The critic is $V=\mathrm{MLP}([h_{\mathrm{agent}},\ \mathrm{meanpool},\ \mathrm{maxpool}])$ over that same ego-subgraph: it is conditioned on one agent's observation and sees a teammate only when that teammate falls inside the agent's own 3-hop radius. 
}

\begin{table}[t]
\caption{{\color{black}Training hyperparameters. Every cell of the grid uses these values: the four BR concepts and the three Red rows differ only in Nature's objective $q^{*}$ (MRE, PR-MRE and DRO-MaxMin; BNE is trained for Nature fixed at nominal) and in the fixed Red opponent, never in the optimiser, the budget or the evaluation protocol.} \newline}
\centering
\small
{\color{black}
\begin{tabular}{@{}lll@{}}
\toprule
Hyperparameter & Value & Source \\
\midrule
\multicolumn{3}{@{}l}{\textit{PPO optimisation}} \\
\addlinespace[2pt]
learning rate & $3\times10^{-4}$ & set \\
LR schedule & \texttt{linear} & set \\
LR decay start (fraction of budget) & 0.6 & set \\
rollout length $n_{\mathrm{steps}}$ (per env) & 256 & set \\
minibatch size & 512 & set \\
epochs per update & 4 & set \\
parallel envs & 8 & set \\
GRU & yes & set \\
deterministic CUDA kernels & yes & set \\
discount $\gamma$ & 0.99 & SB3 2.6.0 default \\
GAE $\lambda$ & 0.95 & SB3 2.6.0 default \\
PPO clip range & 0.2 & SB3 2.6.0 default \\
value clip range & none & SB3 2.6.0 default \\
entropy coefficient & 0 & SB3 2.6.0 default \\
value coefficient & 0.5 & SB3 2.6.0 default \\
gradient clip & 0.5 & SB3 2.6.0 default \\
normalise advantages & yes & SB3 2.6.0 default \\
target KL & none & SB3 2.6.0 default \\
\midrule
\multicolumn{3}{@{}l}{\textit{Policy and recurrence}} \\
\addlinespace[2pt]
policy & \texttt{GraphCTF\dots v4} & set \\
node embedding / hidden dim & 64 / 64 & policy \\
BPTT window \texttt{seq\_len} & 16 (= full episode) & default \\
\midrule
\multicolumn{3}{@{}l}{\textit{Budget and protocol}} \\
\addlinespace[2pt]
training budget (env steps) & 3{,}000{,}000 & set \\
episode horizon (cycles) & 16 & set \\
evaluation episodes per checkpoint & 40 & set \\
checkpoint schedule (env steps) & \texttt{0:100000,600000:200000} & set \\
\midrule
\multicolumn{3}{@{}l}{\textit{$q^{*}$ safeguards (identical in all arms)}} \\
\addlinespace[2pt]
floor $q_k\geq$ & 0.1 & default \\
softmax temperature $\tau$ & 0.25 & default \\
LP/softmax blend & 0.5 & default \\
anneal $p_0\to q^{*}$ over & 0.3 & default \\
freeze $q^{*}$ over last & 0.2 & set \\
EMA on per-type values & 0.2 & default \\
\bottomrule
\end{tabular}
}
\label{tab:hparams}

\end{table}
\subsection{Three Flag Graph-CtF Experiment, $\delta > 0$}
\begin{figure}[h]
  \centering
  \includegraphics[width=\textwidth]{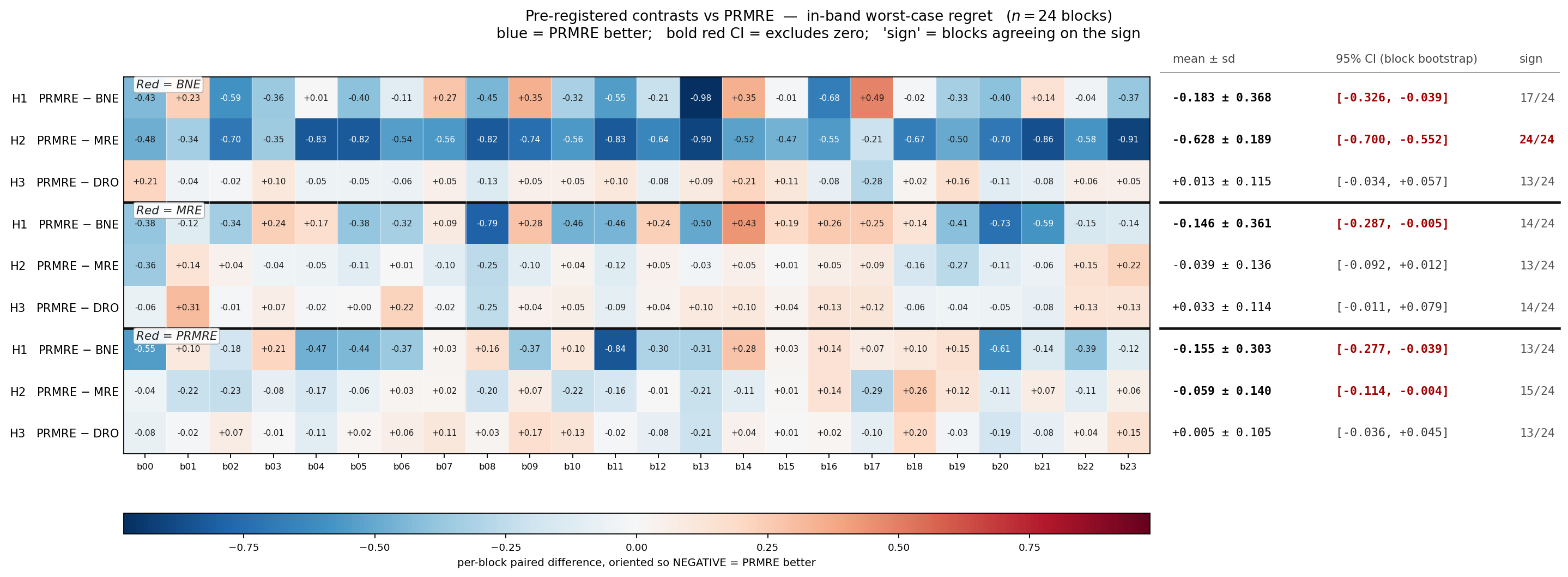}
  \caption{\textbf{In-band worst-case regret}: concept paired-differences across the 24 blocks. Baseline map, Fig.~\ref{wrap-fig:atg}.}   
  \label{fig:wcreg_per_row}
\end{figure}

\begin{figure}[h]
  \centering
  \includegraphics[width=\textwidth]{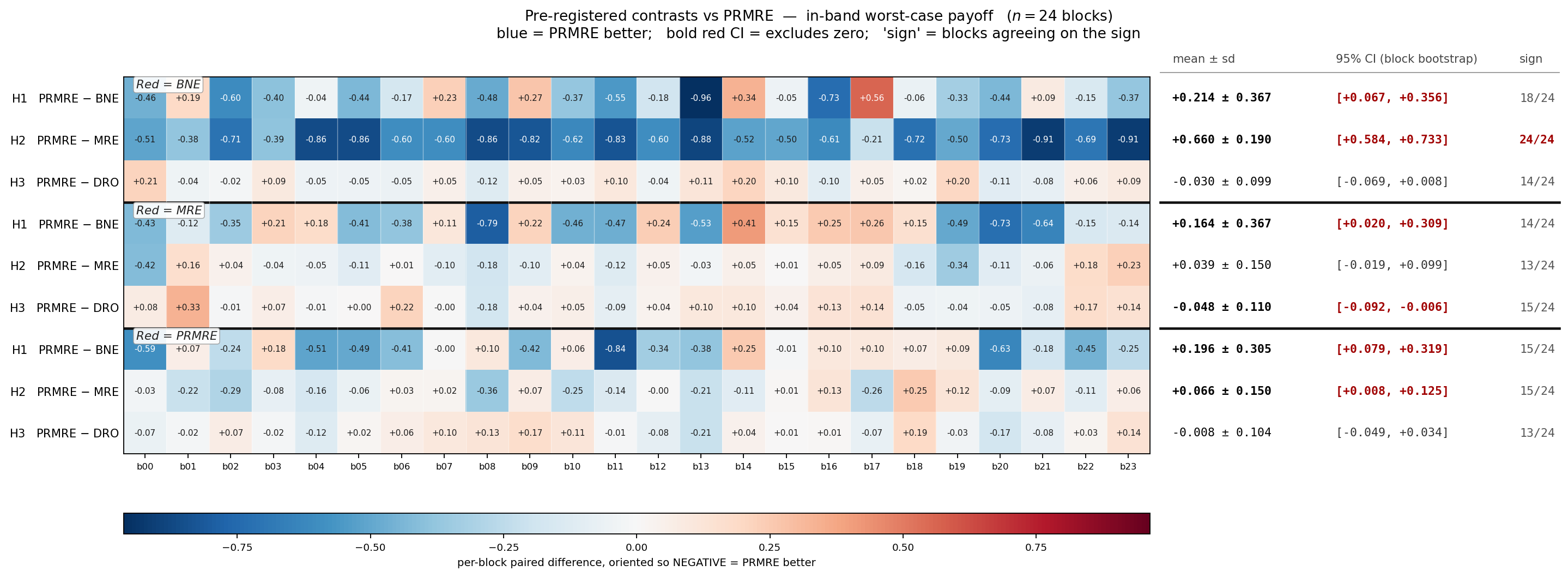}
  \caption{\textbf{In-band worst-case payoff}: concept paired-differences across the 24 blocks. Baseline map, Fig.~\ref{wrap-fig:atg}.}
  \label{fig:wcpay_per_row}
\end{figure}

\begin{figure}[t]
  \centering
  \includegraphics[width=\textwidth]{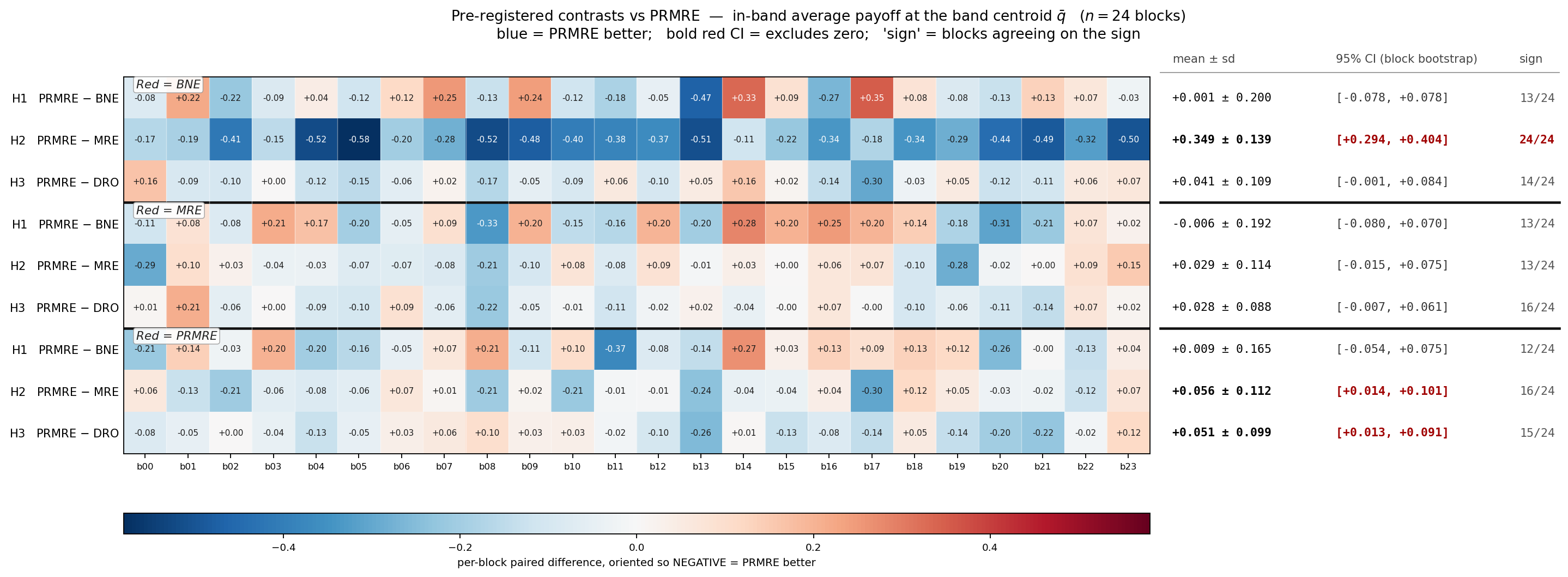}
  \caption{\textbf{In-band average-case payoff}: concept paired-differences across the 24 blocks. Baseline map, Fig.~\ref{wrap-fig:atg}}
  \label{fig:avgpay_per_row}
\end{figure}

\begin{figure}
  \centering
  \includegraphics[scale=0.725]{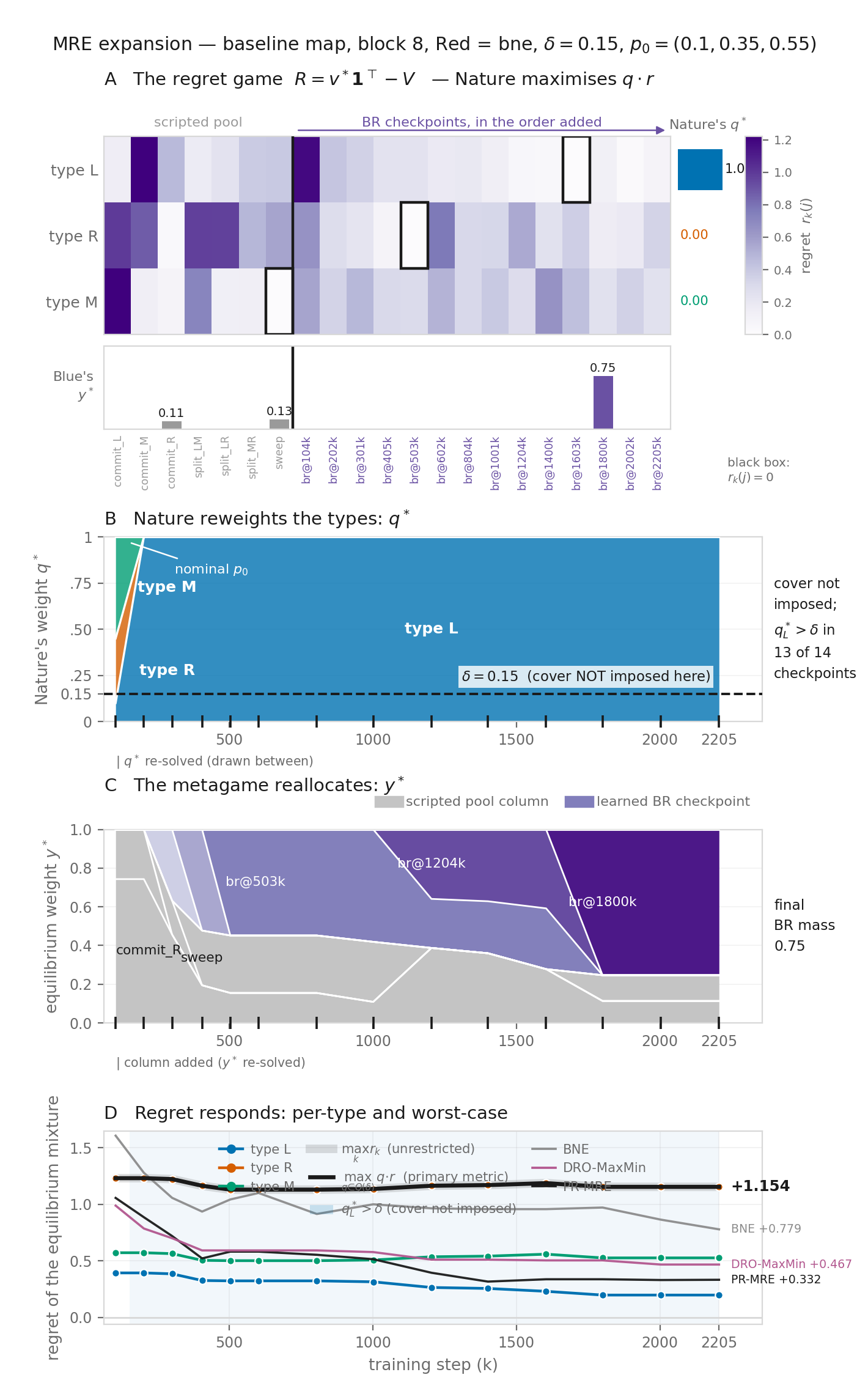}
  \caption{\color{black}{\textbf{Training mechanism, MRE robust best-response}. Block 08, Red = BNE. $\bar{\mu} = (0.1, 0.35, 0.55)$, $\delta=0.15$. \textbf{Panel A}: evolution of the auxiliary regret-based game (white-purple heatmap) as trained checkpoints are added. \textbf{Panel B}: Nature's regret-maximizing adversarial type distribution over $\Delta(\lvert\Theta\rvert)$. \textbf{Panel C}: evolution of the $\operatorname{MRE}$ robust population mixture. \textbf{Panel D}: per-type and worst-case regret evolution during training w.r.t. per-type benchmarks pooled across arms per block.}}
  \label{fig:mro_mechanism}
\end{figure}

\begin{figure}[t]
  \centering
  \includegraphics[width=\textwidth]{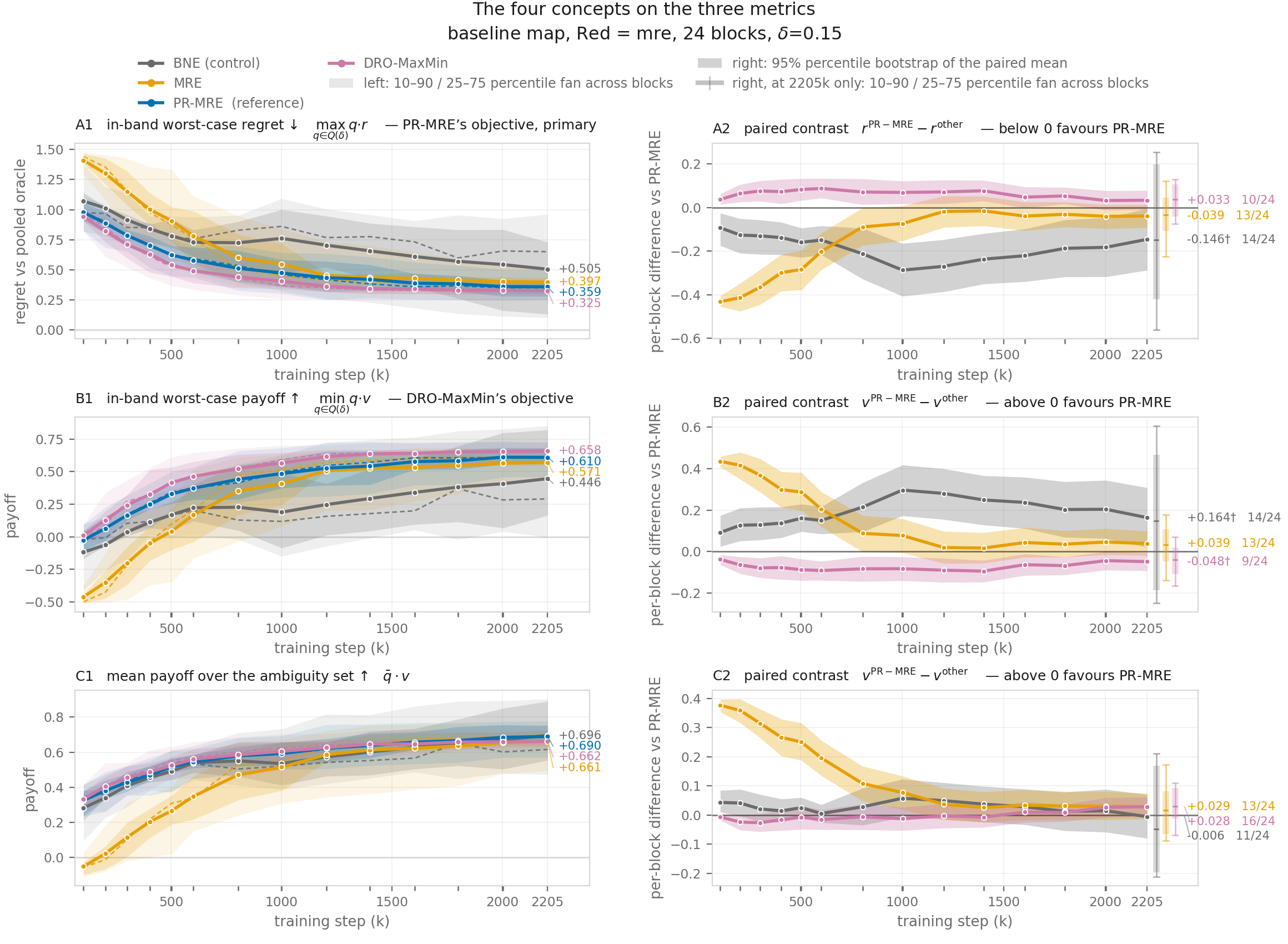}
\caption{{\color{black} \textbf{Metrics evolution during training across blocks. Baseline map (fig.~\ref{wrap-fig:atg}), Red row $r_{1}$ \texttt{MRE}}. The four concepts on the primary and the two seconday metrics, $n=24$ blocks, $\delta=0.15$. Rows: in-band worst-case regret against the per-block frozen pooled benchmark; in-band worst-case payoff; mean payoff over the ambiguity set. Arrows give the favourable direction. \textbf{Left panel}: descriptive only -- the blocks are paired, so overlapping bands neither establish nor rule out a difference. \textbf{Right panel}: the paired per-block contrast against PR-MRE: H1-H3. The solid band is the 95\% percentile bootstrap ($B=20{,}000$) of the paired mean; the bar at the final checkpoint is the 10--90 / 25--75 percentile fan of the per-block differences. Labels give the final paired mean and the blocks favouring PR-MRE.}}
\label{fig:regretcurves-psro-grid-mre-paired}
\end{figure}

\begin{figure}[t]
  \centering
  \includegraphics[width=\textwidth]{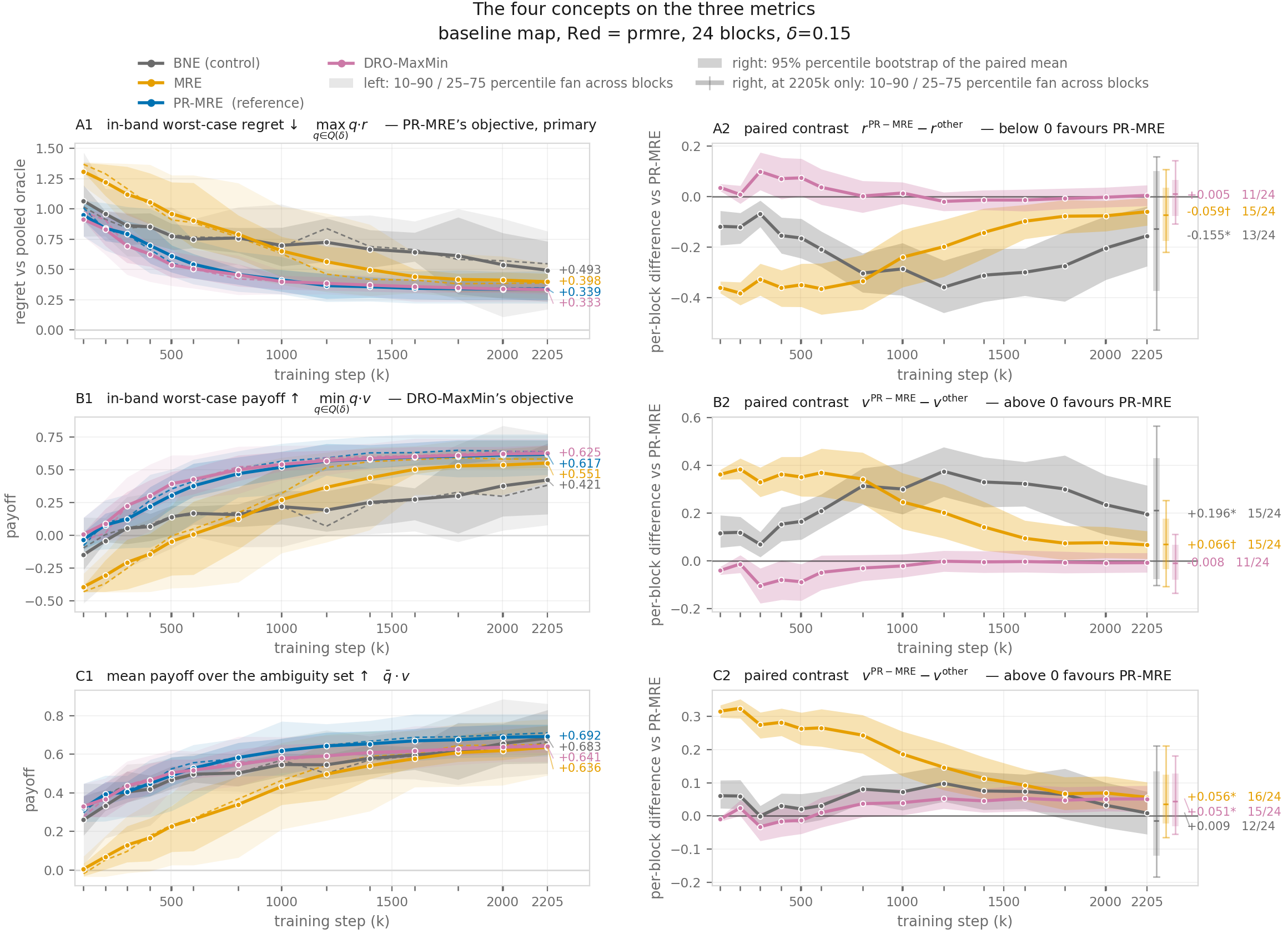}
\caption{{\color{black} \textbf{Metrics evolution during training across blocks. Baseline map (fig.~\ref{wrap-fig:atg}), Red row $r_{2}$ \texttt{PRMRE}}. The four concepts on the primary and the two seconday metrics, $n=24$ blocks, $\delta=0.15$. Rows: in-band worst-case regret against the per-block frozen pooled benchmark; in-band worst-case payoff; mean payoff over the ambiguity set. Arrows give the favourable direction. \textbf{Left panel}: descriptive only -- the blocks are paired, so overlapping bands neither establish nor rule out a difference. \textbf{Right panel}: the paired per-block contrast against PR-MRE: H1-H3. The solid band is the 95\% percentile bootstrap ($B=20{,}000$) of the paired mean; the bar at the final checkpoint is the 10--90 / 25--75 percentile fan of the per-block differences. Labels give the final paired mean and the blocks favouring PR-MRE.}}
\label{fig:regretcurves-psro-grid-prmre-paired}
\end{figure}

\begin{figure}[t]
  \centering
  \includegraphics[width=\textwidth]{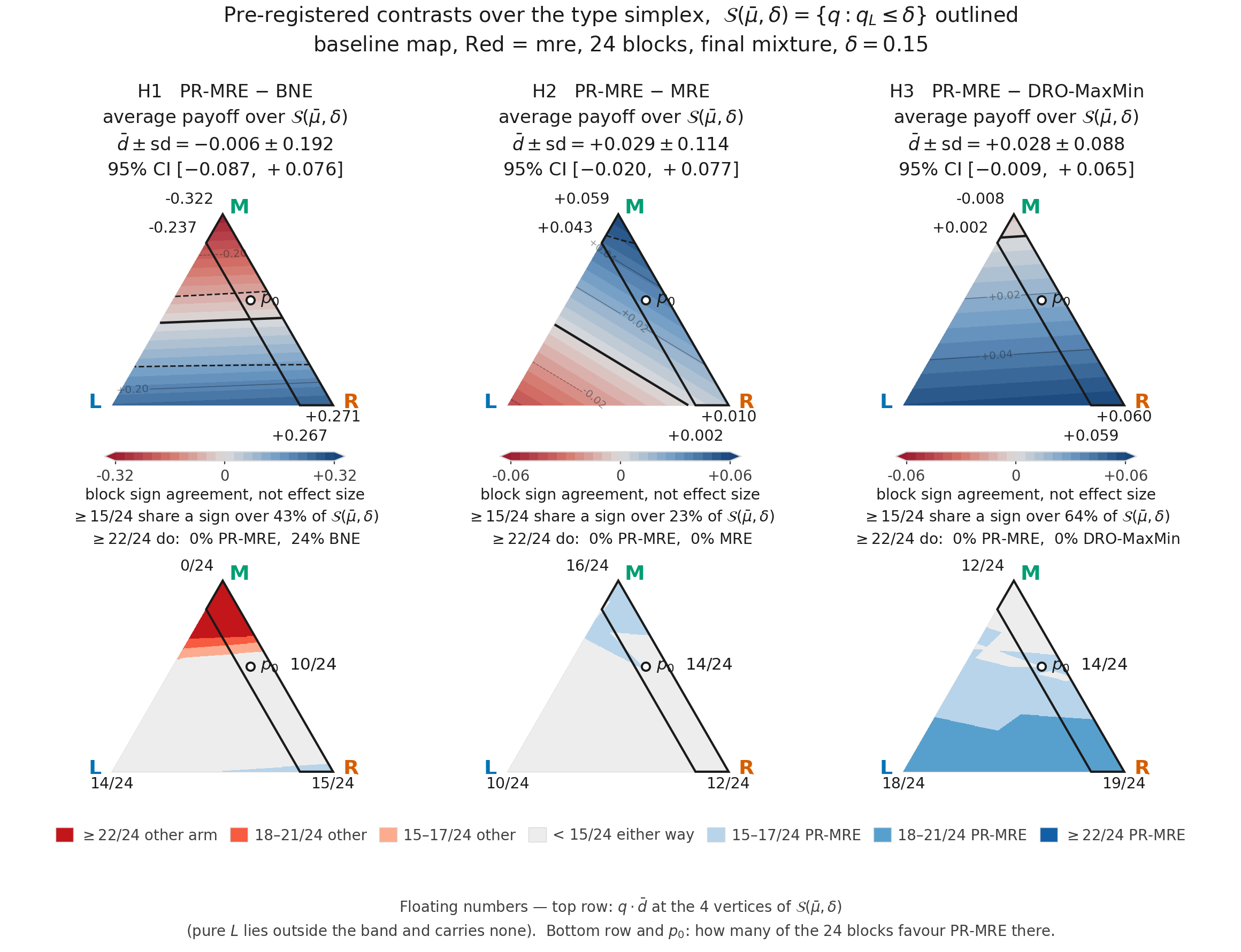}
\caption{{\color{black} \textbf{How do the learned robust best-responses perform across the type-simplex? Baseline map (fig.~\ref{wrap-fig:atg}), Red row $r_{1}$ \texttt{MRE}}. H1--H3 as heatmaps over the type simplex, $n=24$ blocks, $\delta=0.15$. Each value-vector $v \in \mathbb{R}^{\lvert\Theta\rvert}$ corresponds to the final mixture produced by the robust Blue concept (and the final trained team policy at 2200K for BNE) -- the same per-block vectors the H1--H3 table (Table~\ref{tab:contrasts-per-row-n24}) uses. \textbf{Top row}: the fill at a type distribution $q$ is $q\cdot\bar d$, the mean over blocks of $q\cdot(v^{A}(b)-v^{B}(b))$ where $v^{A}(b)$ represents the value-vector for the final trained robust mixture for Blue concept $A$ in block $b$. Averaging the vectors and averaging the heatmaps across blocks agree exactly, because $q\cdot d$ is linear in $d$; the same shortcut is invalid for the worst-case metrics, where $\max$ and mean do not commute, so those are computed per block and then averaged as reported in Table~\ref{tab:contrasts-per-row-n24}. The heavy outline is the cover region $S(\bar{\mu}, \delta)=\{q: q_L\leq\delta\}$ that every arm optimised over; outside it the paired-difference field is the same linear function $q. \bar{d}$, but no arm optimised there and no pre-registered metric lives there. The pre-registered inference remains the interval set of Table~\ref{tab:contrasts-per-row-n24}. Each panel carries its own colour scale (see the bar beneath it), so hue is not comparable across panels; the contour labels are. \textbf{Bottom row}: how strongly the blocks agree, drawn as nested majorities at 15/24, 18/24, 22/24. This is CONSENSUS, NOT SIZE: it counts signs only, so a block favouring an arm by 0.001 counts the same as one favouring it by 0.5. Magnitude is the top row. Corner and $p_0$ annotations give the count of blocks favouring PR-MRE at that type distribution.}}
\label{fig:simplex-psro-grid-mre-spread}
\end{figure}

\begin{figure}[t]
  \centering
  \includegraphics[width=\textwidth]{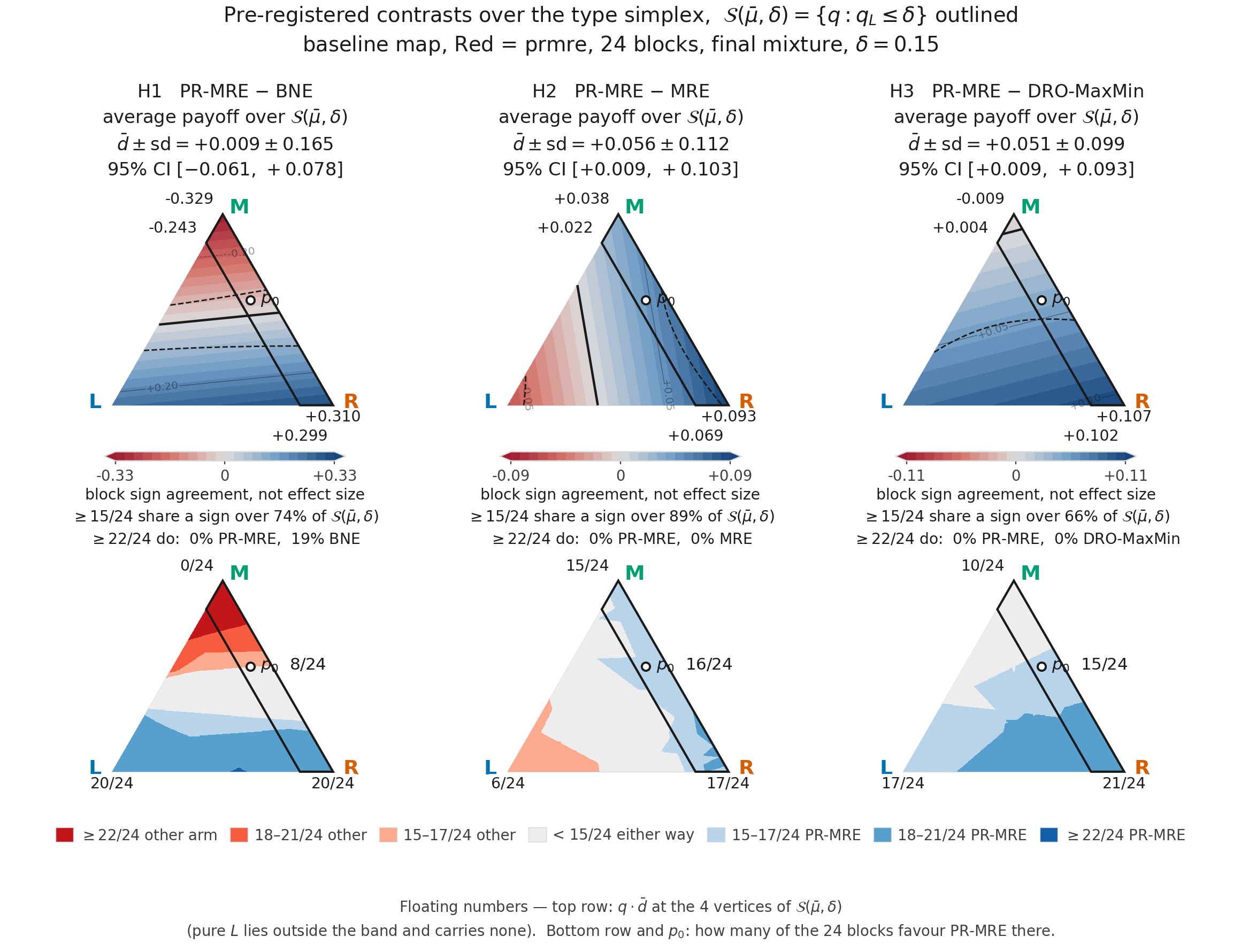}
\caption{{\color{black} \textbf{How do the learned robust best-responses perform across the type-simplex? Baseline map (fig.~\ref{wrap-fig:atg}), Red row $r_{2}$ \texttt{PRMRE}}. H1--H3 as heatmaps over the type simplex, $n=24$ blocks, $\delta=0.15$. Each value-vector $v \in \mathbb{R}^{\lvert\Theta\rvert}$ corresponds to the final mixture produced by the robust Blue concept (and the final trained team policy at 2200K for BNE) -- the same per-block vectors the H1--H3 table (Table~\ref{tab:contrasts-per-row-n24}) uses. \textbf{Top row}: the fill at a type distribution $q$ is $q\cdot\bar d$, the mean over blocks of $q\cdot(v^{A}(b)-v^{B}(b))$ where $v^{A}(b)$ represents the value-vector for the final trained robust mixture for Blue concept $A$ in block $b$. Averaging the vectors and averaging the heatmaps across blocks agree exactly, because $q\cdot d$ is linear in $d$; the same shortcut is invalid for the worst-case metrics, where $\max$ and mean do not commute, so those are computed per block and then averaged as reported in Table~\ref{tab:contrasts-per-row-n24}. The heavy outline is the cover region $S(\bar{\mu}, \delta)=\{q: q_L\leq\delta\}$ that every arm optimised over; outside it the paired-difference field is the same linear function $q. \bar{d}$, but no arm optimised there and no pre-registered metric lives there. The pre-registered inference remains the interval set of Table~\ref{tab:contrasts-per-row-n24}. Each panel carries its own colour scale (see the bar beneath it), so hue is not comparable across panels; the contour labels are. \textbf{Bottom row}: how strongly the blocks agree, drawn as nested majorities at 15/24, 18/24, 22/24. This is CONSENSUS, NOT SIZE: it counts signs only, so a block favouring an arm by 0.001 counts the same as one favouring it by 0.5. Magnitude is the top row. Corner and $p_0$ annotations give the count of blocks favouring PR-MRE at that type distribution.}}
\label{fig:simplex-psro-grid-prmre-spread}
\end{figure}

\subsection{Secondary Map}
\begin{figure}
  \centering
  \includegraphics[scale=0.45]{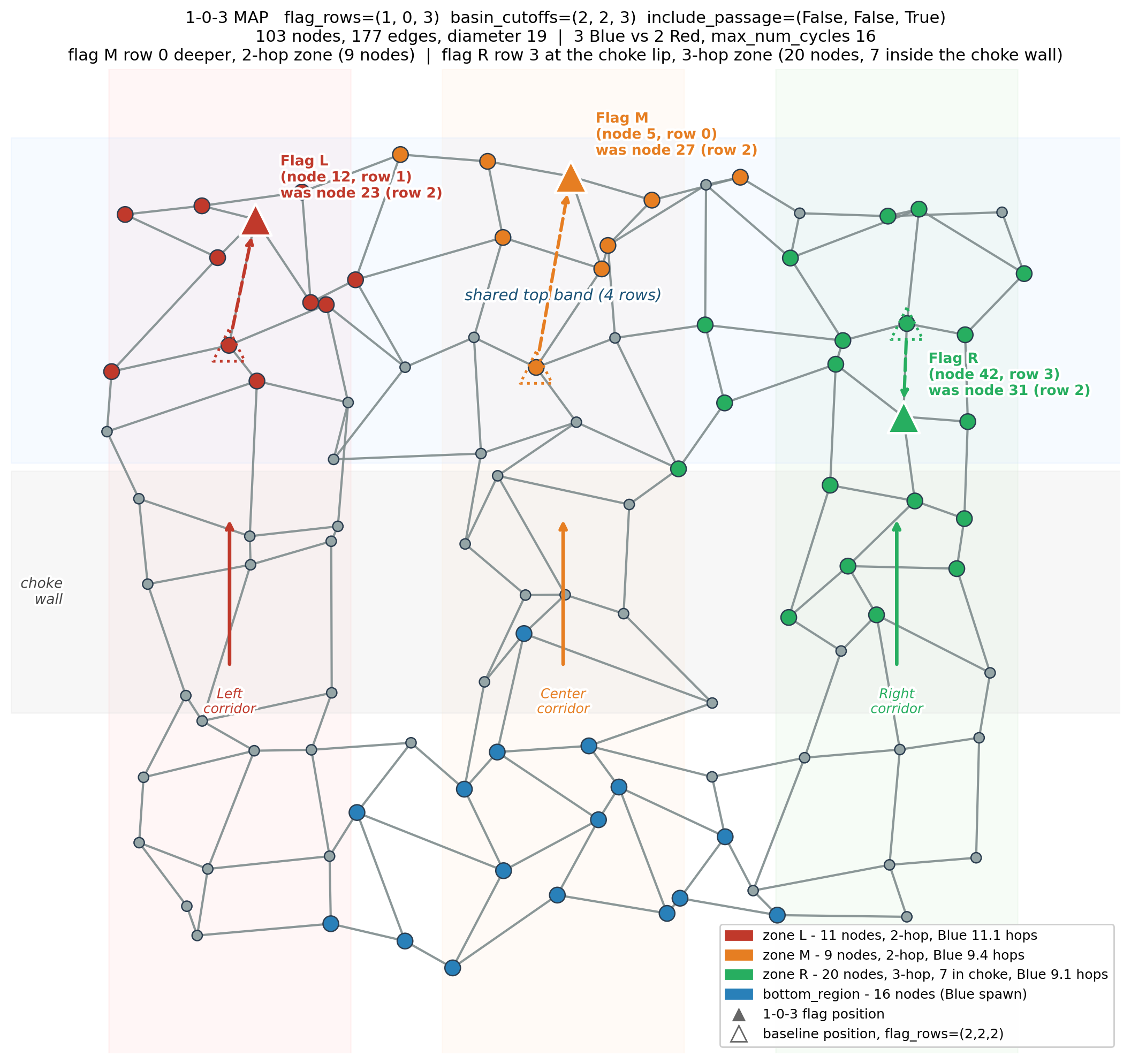}
\caption{{\color{black}\textbf{Secondary Map}. Flag M moves deeper into the top-basin creating a low-payoff and low-regret hard-type, flag R moves closer to the corridor creating high-payoff and relatively higher-regret type and flag L (rare type) also moves row deeper into the top basin. A payoff-based adversarial type distribution samples flag M (nominal $\bar{\mu} = (0.1, 0.35, 0.55)$ and $\delta = 0.15$) versus a regret-based adversarial type distribution that samples the higher regret type, flag R.}}
  \label{fig:1_0_3_map}
\end{figure}
\begin{figure}
  \centering
  \includegraphics[width=\textwidth]{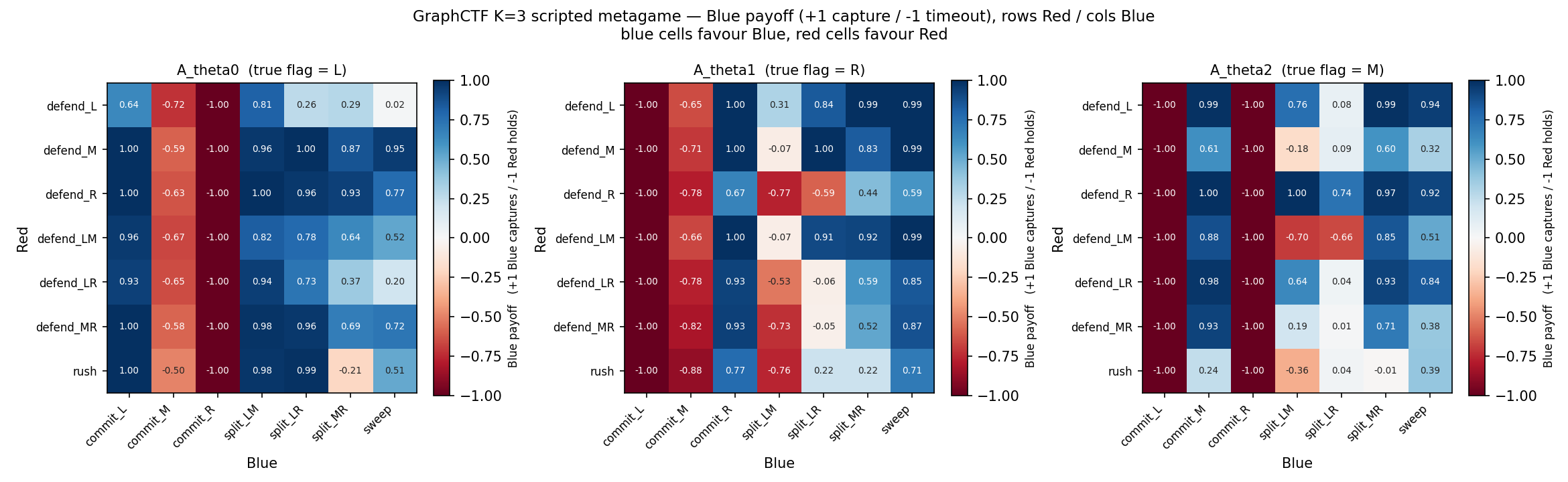}
  \caption{{\color{black} $7\times7\times3$ seed metagame payoff tensor for a three flag Capture-the-Flag game on the graph-structured map given in Figure~\ref{wrap-fig:atg}. $(\theta_{0}, \theta_{1}, \theta_{2}) = (\text{\textbf{L}eft}, \text{\textbf{R}ight}, \text{\textbf{M}iddle})$ and the three slices of the payoff tensor as depicted left to right correspond to $\theta_{0}$--$\theta_{2}$. For any slice, rows depict Red team policies, columns depict Blue team policies and the values in the cells denote Blue team's expected payoff under any matchup.}}
  \label{fig:seed_metagame_3flag_103}
\end{figure}

\begin{figure}[t]
  \centering
  \includegraphics[width=\textwidth]{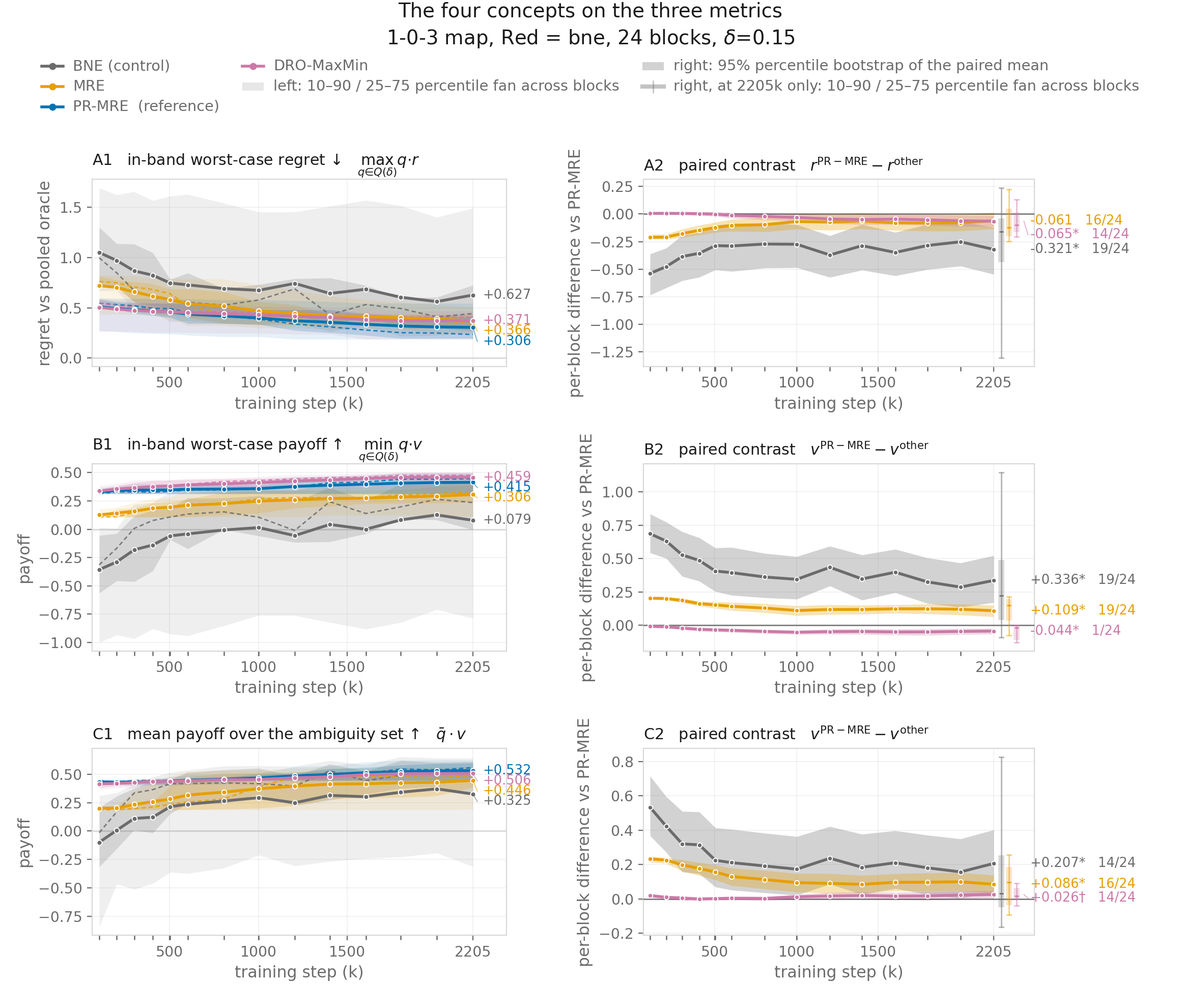}
\caption{{\color{black} \textbf{Metrics evolution during training across blocks. Secondary map (fig.~\ref{fig:1_0_3_map}), Red row $r_{0}$ \texttt{BNE}}. The four concepts on the primary and the two seconday metrics, $n=24$ blocks, $\delta=0.15$. Rows: in-band worst-case regret against the per-block frozen pooled benchmark; in-band worst-case payoff; mean payoff over the ambiguity set. Arrows give the favourable direction. \textbf{Left panel}: descriptive only -- the blocks are paired, so overlapping bands neither establish nor rule out a difference. \textbf{Right panel}: the paired per-block contrast against PR-MRE: H1-H3. The solid band is the 95\% percentile bootstrap ($B=20{,}000$) of the paired mean; the bar at the final checkpoint is the 10--90 / 25--75 percentile fan of the per-block differences. Labels give the final paired mean and the blocks favouring PR-MRE.}}
\label{fig:regretcurves-psro-grid-103-bne-paired}
\end{figure}

\begin{figure}[t]
  \centering
  \includegraphics[width=\textwidth]{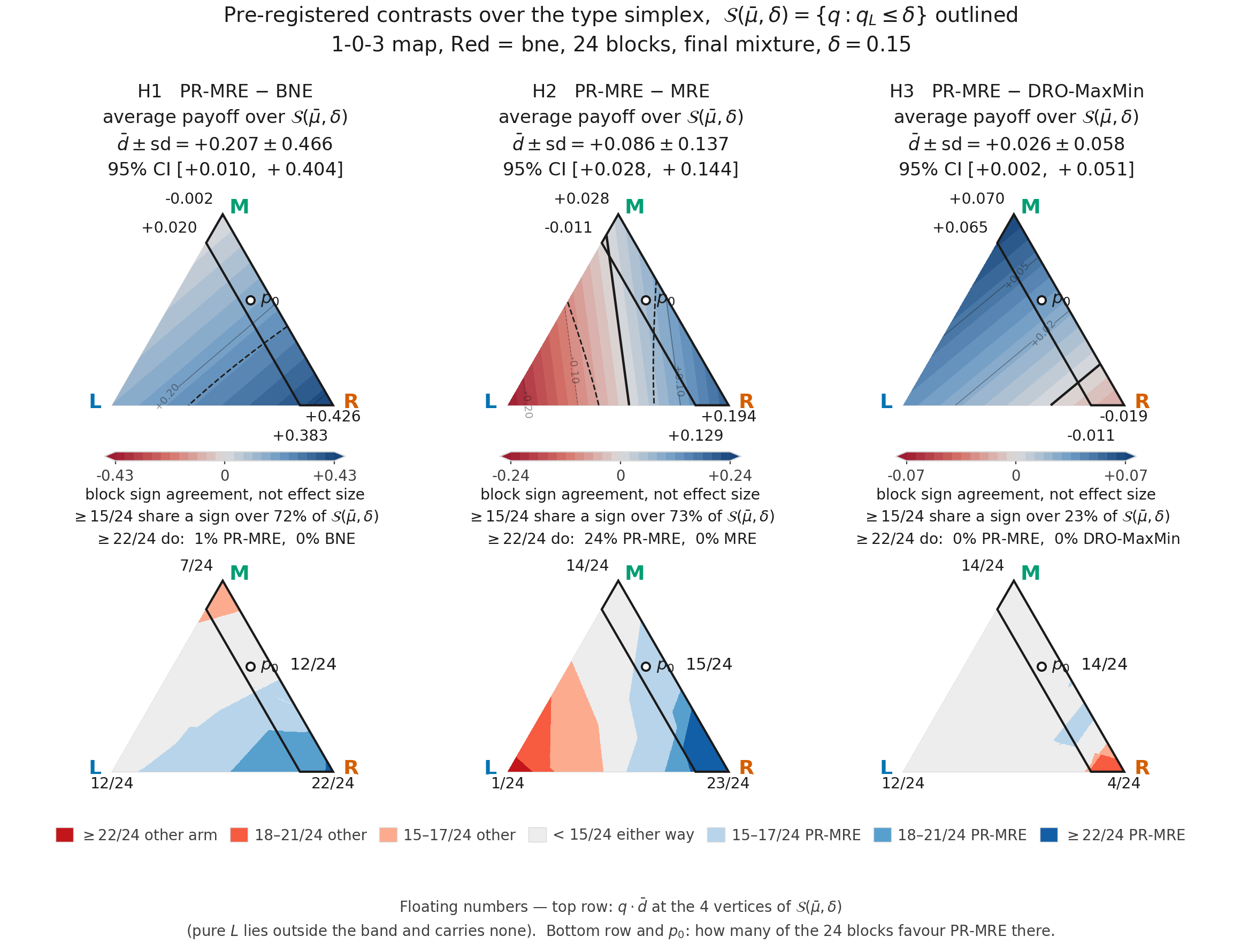}
\caption{{\color{black} \textbf{How do the learned robust best-responses perform across the type-simplex? Secondary map (fig.~\ref{fig:1_0_3_map}), Red row $r_{0}$ \texttt{BNE}}. H1--H3 as heatmaps over the type simplex, $n=24$ blocks, $\delta=0.15$. Each value-vector $v \in \mathbb{R}^{\lvert\Theta\rvert}$ corresponds to the final mixture produced by the robust Blue concept (and the final trained team policy at 2200K for BNE) -- the same per-block vectors the H1--H3 table (Table~\ref{tab:contrasts-per-row-n24-103}) uses. \textbf{Top row}: the fill at a type distribution $q$ is $q\cdot\bar d$, the mean over blocks of $q\cdot(v^{A}(b)-v^{B}(b))$ where $v^{A}(b)$ represents the value-vector for the final trained robust mixture for Blue concept $A$ in block $b$. Averaging the vectors and averaging the heatmaps across blocks agree exactly, because $q\cdot d$ is linear in $d$; the same shortcut is invalid for the worst-case metrics, where $\max$ and mean do not commute, so those are computed per block and then averaged as reported in Table~\ref{tab:contrasts-per-row-n24-103}. The heavy outline is the cover region $S(\bar{\mu}, \delta)=\{q: q_L\leq\delta\}$ that every arm optimised over; outside it the paired-difference field is the same linear function $q. \bar{d}$, but no arm optimised there and no pre-registered metric lives there. The pre-registered inference remains the interval set of Table~\ref{tab:contrasts-per-row-n24-103}. Each panel carries its own colour scale (see the bar beneath it), so hue is not comparable across panels; the contour labels are. \textbf{Bottom row}: how strongly the blocks agree, drawn as nested majorities at 15/24, 18/24, 22/24. This is CONSENSUS, NOT SIZE: it counts signs only, so a block favouring an arm by 0.001 counts the same as one favouring it by 0.5. Magnitude is the top row. Corner and $p_0$ annotations give the count of blocks favouring PR-MRE at that type distribution. }}
\label{fig:simplex-psro-grid-103-bne-spread}
\end{figure}

\end{document}